\documentclass{acScientificArticleTemplate}

\pdfoutput=1

\makeindex

\begin{document}


\renewcommand{\thefootnote}{\fnsymbol{footnote}}

\title{No-Arbitrage Pricing, Dynamics and Forward Prices \\of Collateralized Derivatives}


\author{Alessio Calvelli\footnotemark[3]}
\footnotetext[3]{The opinions here expressed are solely those of the authors and do not represent in any way those of their
employers.}


\address{Banco BPM\\
Financial Engineering\\
Piazza Meda, 4 -- 20121 Milan (Italy)\\\vspace{0.01cm}
\href{mailto:alessio.calvelli@bancobpm.it}{\emph{\texttt{alessio.calvelli@bancobpm.it}}}\\\vspace{-0.1cm}
\href{https://www.linkedin.com/in/alessio-calvelli-888b3612/}{\emph{\texttt{linkedin/alessio.calvelli}}}}

%
\date{June 18, 2024}

\makeatletter
\pdfbookmark[1]{Title page}{title}
\makeatother

\maketitle

\begin{center}
\footnotesize{v4\footnotemark[1] June 18, 2024}\\
\end{center}

\footnotetext[1]{ArXiv version history. v1 (18 Aug 2022): article submitted. v2 (9 Mar 2023): changes of minor nature in Section \ref{sec:multiccyCollDeriv} (in particular, added Proposition \ref{prop:fxforward} and Corollary \ref{cor:multiccyDiscountingExplain}). 
v3 (12 June 2024): updated contact details, corrected typos, added references \cite{bu10}-\cite{bu18}-Proposition \ref{prop:bu18} and corresponding parts of Appendix \ref{app:BasicDividendModels}. v4 (this version): added Proposition \ref{prop:totalReturn}.}

\vspace*{15pt}

\begin{abstract}
\textbf{Abstract}. This paper analyzes the pricing of collateralized derivatives, i.e.~contracts where counterparties are not only subject to financial derivatives cash flows but also to collateral cash flows arising from a collateral agreement. We do this along the lines of the brilliant approach of the first part of Moreni and Pallavicini \cite{MP2017}: in particular we extend their framework where underlyings are continuous processes driven by a Brownian vector, to a more general setup where underlyings are semimartingales (and hence jump processes). 
First of all, we briefly derive from scratch the theoretical foundations of the main subsequent achievements, i.e.~the extension of the classical No-Arbitrage theory to dividend paying semimartingale assets, where by \emph{dividend} we mean any cash flow earned/paid from holding the asset.
In this part we merge, in the same treatment and under the same notation, the principal known results with some original ones. Then we extend the approach of \cite{MP2017} in different directions and we derive not only
the pricing formulae but also the dynamics and forward prices of collateralized derivatives (extending the achievements of the first part of Gabrielli \emph{et al.} \cite{GPS2019}). Finally, we study some important applications (Repurchase Agreements, Securities Lending and Futures contracts) of previously established theoretical frameworks, obtaining some results that are commonly used in practitioners literature, but often not well understood.
\end{abstract}


\vspace*{1cm}
\tableofcontents

\markboth{A.~Calvelli}{No-Arbitrage Pricing, Dynamics and Forward Prices of Collateralized Derivatives}

\renewcommand{\thefootnote}{\arabic{footnote}}

\section{Introduction}

\subsection{Motivation\label{sec:Intro}}

It is well known that, when entering into any financial contract establishing some (optional) future financial transactions, a risk for a party signing this contract is the risk that 
her counterparty defaults and fails to pay some due future financial flows specified in the contract. In order to tackle this issue, very often the counterparties agree to sign a collateral agreement, i.e.~an annex of the financial contract where they engage themselves to a process referred to as \emph{collateralization}: the counterparty (or a third independent party acting as custodian) running the credit risk receives the \emph{collateral}, i.e.~cash or liquid securities, to cover some or all of this risk: the rationale being that in case of default she can sell loaned securities or seize the loaned cash to offset her uncovered positions. 

The main objective of this paper is to analyze the pricing of collateralized (financial) derivatives, i.e.~contracts where counterparties are not only subject to financial derivatives cash flows but also to collateral cash flows arising from the collateral agreement.
We do this along the lines of the brilliant approach of the first part of \cite{MP2017}, a paper that adapted the results of \cite{PPB2012} to multiple currencies in case of perfect collateralization. The findings of \cite{PPB2012} were also subsequently obtained in \cite{BFP19} in terms of Backward Stochastic Differential Equations (BSDEs). 
Finally,  we mention the recent publication of \cite{BBFPR22} that compares this approach, settled in practitioners literature, with the more elegant \vvirg{replication portfolio approach} generally used in academics literature.

First of all, we derive from scratch the theoretical foundations of the main subsequent achievements, i.e.~the extension of the classical No-Arbitrage theory to dividend paying semimartingale assets, where by \emph{dividend} we mean any cash flow earned/paid from holding the asset.
In this part we merge, in the same treatment and under the same notation, the principal known results with some original ones. 
Then we extend the framework of \cite{MP2017} where underlyings are continuous processes driven by a Brownian vector, to a more general setup where underlyings are semimartingales. Therefore, all processes can jump:  this is coherent not only with the fact that the financial derivative price is inherently a jump process with jumps coinciding with intermediate cash flows, but also with the fact that the collateral value process is intrinsically purely discontinuous (see Remark \ref{rk:DiscreteCollateralization}). We allow (not only jumps but also all) interest rates be stochastic and we derive both
the pricing formulae and also the dynamics and forward prices of collateralized derivatives (extending the achievements of the first part of \cite{GPS2019}) and finally we study some important applications (Repurchase Agreements, Securities Lending and Futures contracts) of previously established theoretical frameworks.   

As in \cite{MP2017} (but this choice is common also to other papers, see e.g.~\cite{vp2010}) we do not take into account the residual possibility of a loss on a collateralized contract due to the default of the counterparty: this is a simplifying hypothesis to better understand the mechanics of collateralization. In fact, since the collateralization strongly reduces the bilateral counterparty risk, we take a step forward and assume that it \emph{eliminates} it: this assumption becomes more realistic the more the collateralization process is performed continuously (we will refer to this case as \emph{continuous margin calls}). In cases where the underlying assets are continuous and driven by a Brownian vector as in \cite{PPB2012}, the continuous collateralization implies perfect collateralization, meaning that the collateral perfectly covers the close-out amount (i.e.~the residual value of the financial derivative at default time) and hence that the counterparty risk is literally eliminated. 
When introducing jumps in the underlyings -- even in cases of continuous
collateralization -- since the collateral is by definition a predictable process, it could differ from
the close-out amount which is by nature optional (i.e. only adapted but not predictable). In fact, the close-out amount is in some way dependent on the financial derivative price, which is an optional process (since the latter is in some way dependent on the underlyings’ quotes).
However, the possible difference between the collateral value and the close-out amount could be covered by the fact that, generally, the financial derivatives are over-collateralized (i.e.~the collateral value is set to a quantity strictly greater than the financial derivative price: see Section 
\ref{sec:CollaterDerivatives}). So, even within the framework of the present article, we could still achieve perfect collateralization if, during the continuous collateralization process, the collateral value is set in such a way that the over-collateral covers 
all the jumps (particularly the unpredictable ones) of the derivative's price. If this is not the case, the residual counterparty risk (when it is non-negligible) should be taken into account and the results of the present paper should be considered as an approximation.

The pros of the aforementioned simplification are that we obtain clearer formulae and that we rigorously explain some results that are commonly used in practitioners literature, but often not well understood -- see Section \ref{sec:Applications}. On the other hand, the introduction of the counterparty risk would be quite straightforward since this subject is broadly explored and the literature dealing with it is well settled. In general, the approach of \cite{MP2017} (and hence our approach) already contains the essential elements of the approach of \cite{PPB2012} which tackles the calculation of 
all valuation adjustments -- not only the collateral ones. 

For obtaining all the achievements described above, the common thread of the paper will be to identify, thanks to No-Arbitrage conditions, Risk-Neutral martingales in progressively more challenging contexts where our intuition could be increasingly lost: 
as a first example we will see the martingale  corresponding to a non-dividend paying asset, then we will strive to recognize the martingale corresponding to a dividend-paying asset,   
as a third step we will detect a martingale linked to a collateralized derivative, finally we will find martingales in more specific contracts. All these efforts in searching for martingales are motivated from the fact that 
martingales have some nice properties and, primarily, since they have well defined dynamical features and specific connections with expected values: among all expected values we are particularly interested in pricing ones.

This paper is organized as follows: the next subsection consists in a brief technical setup, Section \ref{sec:GenTheory} develops the general theory to be used in the remainder of the paper, 
Section \ref{sec:CollaterDerivatives} is dedicated to the presentation and analysis of collateralized derivatives, in Section \ref{sec:Applications} we describe some specific applications of previous sections, the last section 
outlines some concluding remarks.


\subsection{Technical Setup}


All processes of the present paper are semimartingales: we refer to Appendix \ref{sec:SemiMartingales} for semimartingales notation, a list of relevant results and some references.
As in  \cite{pp2001} we will assume that we are given a filtered complete probability space $(\Omega;\mathbb{F};(\fst_t )_{t\geq 0};\p)$
where $\fst_t\subseteq \mathbb{F}$ for any $t\geq 0$ and $\p$ is the so called \vvirg{real world probability} or \vvirg{physical measure}. We further assume that $\fst_s \subset \fst_t$ if $s<t$; $\fst_0$ contains all the $\p$-null
sets of $\fst$; and also that $\bigcap_{s>t} \fst_s\equiv\fst_{t+}=\fst_t$  by hypothesis. This last property is called
the right continuity of the filtration. These hypotheses, taken together, are known
as the \emph{usual hypotheses} (when the usual hypotheses hold, one knows that every
martingale has a version which is RCLL, one of the most important consequences of
these hypotheses).

\begin{remark}
The reader who is unfamiliar with jump processes and semimartingale theory could read only a subset of the following results interpreting all processes as \emph{continuous} processes: 
under this simplifying hypothesis, for any processes $X,Y$ and time $u$, one has $X_{u-}=X_u$ and $\Delta X_u=0$ and Quadratic Variation/Covariation equal to their predictable versions: 
$[X,X]_u=\langle X,X\rangle_u$ and $[X,Y]_u=\langle X,Y\rangle_u$. 
\end{remark}

Moreover, all vectors of the paper are considered as column vectors and we denote with $\vec{0}$ the vector with all components equal to zero (its 
dimension will be clear from the context). We also use the notation $a\wedge b:=\min\{a,b\}$ and $a \vee b := \max\{a,b\}$. We define with $\Theta_{T}(u):=\1_{u\geq T}$ the Heaviside step function centered at $T$ and with $\delta(u-T)$ the Dirac mass centered at $T$, where $u,T\in\re$ and we have $\partial_u \Theta_{T}(u)=\delta(u-T)$ (distributional derivative). We also use the convention $\int_t^T:=\int_{(t,T]}$.

Unless stated otherwise, any interest rate process is stochastic, predictable, continuous and bounded: for an interest rate process $x:=(x_t)_{t\geq 0}$ we denote the corresponding bank account with $B_t^x := \esp^{\int_0^t x_u \diff u}$ (then also any bank account is bounded).
Clearly $B^r\equiv B$ where $r$ is the domestic spot risk-free interest rate process. 
For $t\leq T$ we define the $T$-zero coupon bond price process associated with interest rate $x$ as
\[
P_t^x(T) :=\e_t\left[\frac{B_t^x}{B_T^x}\right]
\]
where $\e_t[\cdot]$  stands for the expectation under measure $\q$ conditioned to $\fst_t$, and we denote with $\q$ the domestic Risk-Neutral measure (the measure with numéraire $B$). Of course $P_t(T)\equiv P_t^r(T)$. Finally
$\q^T$ is the domestic $T$-forward measure (the measure with numéraire $P_\cdot(T)$) and $\e_t^T[\cdot]$  stands for the expectation under measure $\q^T$ conditioned to $\fst_t$.


\section{General Theory\label{sec:GenTheory}}



\subsection{Non-Dividend Paying Assets}


We start with two cornerstones of Asset Pricing: see \cite{ds94,Bjork,pp2001} or Theorem 2.1.4.~of \cite{jyc2009} and references therein for proofs and an explanation of
the condition of No Free Lunch with Local Vanishing Risk (NFLVR): the less interested reader can understand this condition as \vvirg{no-arbitrage}.

\begin{theorem}[First Fundamental Theorem of Asset Pricing ($\mathbf{1}^{\text{\textbf{st}}}$-FTAP)]\label{th:fftap}
Consider the market model of non-dividend paying underlying processes $S^0, S^1,\ldots, S^n$ under the filtered probability space $(\Omega, \mathbb{F}, \{\fst_t\}_{t\geq 0}, \p)$ satisfying the usual hypotheses.  Assume that
$S^0_t>0$ $\p$-a.s.~for all $t\geq 0$ and that $S^0, S^1,\ldots, S^n$ are locally bounded semimartingales. Then the following conditions are equivalent:
\begin{romanlist}
\item The model satisfies the NFLVR;
\item There exists a measure $\q^0\sim \p$ such that the deflated processes
\[
\left\{\frac{S^0_t}{S^0_t}\right\}_{t\geq 0}, 
\left\{\frac{S^1_t}{S^0_t}\right\}_{t\geq 0}, \ldots, \left\{\frac{S^n_t}{S^0_t}\right\}_{t\geq 0}
\]
are local martingales under $\q^0$, which is called Equivalent Martingale Measure (EMM). Then we call the process $S_0$ as the numéraire of the measure $\q^0$.
\end{romanlist}
\end{theorem}

\noindent The following result establishes that the dynamics of asset prices have to be semimartingales.

\begin{theorem}[From \cite{jyc2009}] Let $S$ be an adapted RCLL process. If $S$ is locally bounded
and satisfies the NFLVR property for simple integrands, then $S$ is a semimartingale.
\end{theorem}

We will see (in a more general setting) that  not only the discounted prices of securities are local martingales, but also any self-financing strategy and then price, and in particular
prices of financial derivatives. As anticipated, in the special case where $S^0\equiv B$, the bank account process, we call $\q\equiv \q^0$ as the Risk-Neutral measure.

\begin{theorem}[$\mathbf{2}^{\text{\textbf{nd}}}$-FTAP]
Assume that the model satisfies the NFLVR condition and consider a fixed numéraire $S^0$. Then the market is complete iff the EMM $\q^0$, corresponding to numéraire $S^0$, is unique.
\end{theorem}

\noindent In case $\q^0$ is unique, then any price process is unique.


\subsection{Dividend-Paying Assets\label{sec:AssetPricingDivs}}


\subsubsection[$1^{\text{st}}$-FTAP]{$\mathbf{1}^{\text{\textbf{st}}}$-FTAP\label{sec:FirstThAssetPricing}}


We consider a market with underlying assets $S=(S^0,S^1,\ldots, S^n)$ where each $S^i$ is a locally bounded semimartingale. Define with $D:=(D^0,D^1,\ldots, D^n)$ the  cumulative dividend vector process, where each $D^i$ is a locally bounded semimartingale representing the undiscounted cumulative net (after taxes) paid by asset $S^i$ from inception. In particular,
\[
D_t^i = D_0^i + \int_0^t \diff D_u^i
\]
and $\diff D_u^i$ are the net dividends paid by $S^i$ in the interval $\diff u$. One can imagine $D^i$ as an account (at zero interest rates) that grows with dividend payments $\diff D_u^i$. 
We fix $S^0\equiv B$ so that clearly $\diff D^0_t=\Delta D^0_t=0$ for any $t$ and $\q^0\equiv \q$ (the Risk Neutral measure with numéraire $B$).

As we read in \cite{DC88}, the convention we choose is for the dividend or asset price change at $t$ to be included in the cumulative dividend process or price at $t$. In technical terms, this says that for each asset $i$, 
both $D^i$ and $S^i$ are assumed to be right-continuous (recall that semimartingales have RCLL paths). The lump net dividend paid out at time $t$ by security $S^i$ is thus $\Delta D_t^i := D_t^i - D_{t-}^i$, and $\Delta S_t^i := S_t^i - S_{t-}^i$ is the jump of the asset: in case $S^i$ jumps only due to the cash (i.e.~lump) dividends we have the following change in price as the stock goes ex-dividend: $\Delta S_t^i=-\Delta D_t^i$ (more on this at Remark \ref{rk:DividendDropping}).
 
If $D^i$ is absolutely continuous, the dividend rate $\partial_t D_t^i$ exists for almost all $t$ and $D_t^i$ is its integral. As another
special case, if dividends occur only in lumps, then $D^i$ is a random step function. Because it is of no real use to us to have a dividend payment at $t = 0$, we assume, without loss of generality,
that $D_0^i=D_{0-}^i = 0$ and $S_{0-}^i = S_0^i$ so that $\Delta D^i_0=\Delta S^i_0=0$.

\begin{remark}
As we will see later on, we will extend the concept of dividends to any intermediate (after taxes) coupon of the asset $S$, in this sense the coupon is a \vvirg{dividend} of a financial derivative. Hence, we do not require the non-negativity of the dividend process as it should be for dividends in a strict sense. 
\end{remark}

Define also the gain process as
\be\label{eq:gainProcess}
G_t:=(G_t^0, \ldots, G_t^n) \qquad G_t^i := S_t^i + D_t^i
\ee
that is so called  since $\diff G_t^i = \diff S_t^i + \diff D_t^i$ is the gain at time $t$ one experiences holding the asset $S^i$: it is the sum between the capital gain and the dividend gain (either can be negative). 


We build a portfolio of semimartingale vector quantities $\varphi=(\varphi^0,\varphi^1,\ldots,\varphi^n)$ with value
\be\label{eq:DefSelfFinPort}
\Psi_t(\varphi)  := \varphi_t\cdot S_t :=\sum_{i=0}^n \varphi_t^i\, S_t^i  
\ee
which, with no-dividend jumps (more on this later), is self-financing if we set
\bes
\diff \Psi_t \setto  \varphi_{t} \cdot \diff G_t 
\ees
where $G$ is defined in \eqref{eq:gainProcess} and $\diff G_t:=(\diff G^0_t, \ldots, \diff G^n_t)$. The self-financing portfolio constraint means that the only change in portfolio value comes from capital gains and dividend gains, whatever the trading strategy. The trading strategy can move value between the stock and cash accounts but not 
create or destroy value.
In fact, heuristically, integrating in the infinitesimally small interval $(t, t+\delta t]$, we have
\bes
\Psi_{t+\delta t} &:=& \varphi_{t+\delta t}\cdot S_{t+\delta t} \\
&=& \Psi_t + \int_{t}^{t+\delta t}  \diff \Psi_u \\
&\stackrel{\textit{self-fin.}}{=}& \Psi_t + \int_{t}^{t+\delta t}  \varphi_{u} \cdot \diff G_u \\
&\approx& \Psi_t +  \varphi_t\cdot \Big[  G_{t+\delta t} - G_{t} \Big]\\
&:=& \varphi_t\cdot \Big[S_{t+\delta t} + \Big(D_{t+\delta t} - D_{t}\Big)\Big]
\ees
where the first term of the last equation represents the allocation at time $t$ that naturally evolves due to the underlyings move in $t+\delta t$, and the second term represents the dividends that drop in the interval $(t, t+\delta t]$. The wealth that is produced in the last equation due to the allocation at time $t$ and the market move \emph{must} be totally reallocated with the new quantities $\varphi_{t+\delta t}$: see the first equation.

The previous heuristical reasoning has some issues in a semimartingale framework:
\begin{romanlist}
\item The integrand of the self-financing condition must be a locally bounded \emph{predictable} process: 
\begin{itemize}
\item This is needed from a technical point of view in order to have well posed semimartingale integrals;
\item From a Mathematical Finance point of view, this fact has some no-arbitrage implications: see Examples 14.1/14.5 in \cite{ap2011} or Example 8.1 in \cite{tc2005}.
\end{itemize}
\item The previous heuristical reasoning does not consider the presence of jumps.
\end{romanlist}

In order to tackle the first issue we modify the above condition with the left limit modification of the trading strategy (recall from Proposition \ref{prop:leftLimitOFXisPredictable} 
that the process $\varphi_-$ is predictable):
\be\label{eq:selffinWithJumps}
\diff \Psi_t \setto  \varphi_{t-} \cdot \diff G_t 
\ee
and the left limit is also coherent with the fact that the stock holder earns the cash dividend ($\Delta D_t:=(\Delta D^0_t, \ldots, \Delta D^n_t)$) and is subject to the asset jump ($\Delta S_t:=(\Delta S^0_t, \ldots, \Delta S^n_t)$) on her position $\varphi$ 
written one instant before these jumps materialize (hence at time $t-$). In order to tackle the second issue we have the following proposition.

\begin{proposition}
The system (\ref{eq:DefSelfFinPort})-(\ref{eq:selffinWithJumps}) implies the jump-self-financing condition, i.e.~for any $t$,
\beq\label{eq:JumpDividendAndSelffinPort}
\Delta \varphi_t \cdot S_{t} =  \varphi_{t-} \cdot \Delta D_t
\eeq
which says that the (possibly unpredictable) jump dividend gain (rhs of the above equation) must be absorbed by a (jump) increase in asset quantities (left-hand side of the above equation where the asset value has gone ex-dividend). 
\end{proposition}
\begin{proof}
By \eqref{eq:linearityOfJumpOperator}-\eqref{eq:JumpsOfSemiMgIntegral}, the self-financing condition \eqref{eq:selffinWithJumps} implies
\beq\label{eq:jumpsOfPsi}
\Delta \Psi_t =  \varphi_{t-} \cdot \Delta G_t 
\eeq
where one observes that the trader earns the jump gain $\Delta G_t=\Delta S_t+\Delta D_t$ under position $\varphi_{t-}$ (the position on her book one instant before the jumps realization). Therefore
\[
\Psi_t = \Psi_{t-} + \Delta \Psi_t = \varphi_{t-} \cdot S_{t-} +  \varphi_{t-} \cdot \Delta G_t = \varphi_{t-} \cdot ( S_{t} +  \Delta D_t),
\]
on the other hand  $\Psi_{t} := \varphi_{t} \cdot S_{t}$ by \eqref{eq:DefSelfFinPort}, then through equating these two equivalent expressions one obtains \eqref{eq:JumpDividendAndSelffinPort}. We explore also an alternative derivation of \eqref{eq:JumpDividendAndSelffinPort}:
using \eqref{eq:linearityOfJumpOperator}-\eqref{eq:jumpOperatorProductRule}, from \eqref{eq:DefSelfFinPort}:
\bes
\Delta \Psi_t &:=& \Delta( \varphi_{t} \cdot S_{t} ) = \varphi_{t-} \cdot \Delta S_{t} + S_{t-} \cdot \Delta \varphi_{t} + \Delta S_{t} \cdot \Delta \varphi_{t}\\
&=& \varphi_{t-} \cdot \Delta S_{t} + S_{t} \cdot \Delta \varphi_{t}
\ees
which can be compared to \eqref{eq:jumpsOfPsi} to obtain again condition \eqref{eq:JumpDividendAndSelffinPort}.
\end{proof}

\begin{corollary}
In case the dividend vector process $D$ is null or continuous, we have $\Delta \varphi_t=0$ for all $t$, or $\varphi_t=\varphi_{t-}$. In this case our first guess of the
self-financing condition is correct: we have $\diff \Psi_t \setto  \varphi_{t} \cdot \diff G_t$, which is the standard self-financing condition in the literature: see 
e.g.~\cite{pp2001} or \cite{jyc2009}.
\end{corollary}

\begin{remark}
Under condition (\ref{eq:JumpDividendAndSelffinPort}), even if the integrals in (\ref{eq:selffinWithJumps}) are well posed, the trading strategy vector process $\varphi$ generally looses the predictability feature in cases where either $S$ or $D$ are optional (the process $\varphi_-$ is predictable but $\varphi$ is optional). 
This is problematic, since it seems reasonable that the trading strategy be predictable: the trading strategy represents the trader's holdings at time $t$, and this should be based on information obtained at times strictly before $t$, but not $t$ itself. In other words, the trader cannot be aware of all jumps, even the unpredictable ones.\\
In order to tackle this issue we could redefine the trading strategy with a LCRL%
\footnote{Note that a left-continuous trading strategy $\theta$ could have (predictable) discontinuities of type $\theta_{t+}-\theta_t$ and this is also coherent with discrete left-continuous rebalancing of the portfolio:  e.g.~at rebalancing times $0=t_0< t_1<\ldots$ we could set $\theta_t\setto\theta_1 \1_{t=0}+ \sum_{i=1}^\infty \theta_i \1_{t_{i-1} < t \leq t_{i}}$ for some predictable vector random variables $\theta_i$.}
 vector process $\theta$ so that process $\theta$ is predictable. Then, as in \cite{DD2001,DC88}, we have this new system for all $t\geq 0$:
\beq
\label{eq:PredictableSelffinancingSystem}
\begin{cases}
\theta_t=\theta_{t-}\\
\Psi_t(\theta)  := \theta_t \cdot S_{t} + \theta_t \cdot\Delta D_t\\
\diff \Psi_t(\theta) \setto  \theta_{t} \cdot \diff G_t 
\end{cases}
\eeq
In practice one could set $\theta_t$ in the following way:
\[
\Psi_{t-}(\theta) :=\theta_{t-} \cdot (S_{t-}+\Delta D_{t-}) = \theta_{t} \cdot S_{t-}\stackrel{\textit{self-fin.}}{=} \Psi_{0}(\theta)+\int_0^{t-} \theta_{u} \cdot \diff G_u 
\]
where the second equality is by Remark \ref{rk:JumpsAtTMinus} and left continuity of the trading strategy: therefore one should set freely the current value of vector $\theta_t$ redistributing the whole portfolio value at time $t-$ at the last equality (coherent with the self-financing condition and all past strategies) using the underlying market values $S_{t-}$ at time $t-$ (see third equality). Once the trading strategy $\theta_t$ is set, jumps of the asset and/or of the dividend process may arrive as surprises and perturb the portfolio value from $\Psi_{t-}$ to $\Psi_{t}$: see the second equation of (\ref{eq:PredictableSelffinancingSystem}).\\
Moreover,
\bes
\Psi_{t-}(\theta)+\Delta \Psi_t(\theta) &=& (\theta_{t-}\cdot S_{t-} + \theta_{t-}\cdot \Delta D_{t-}) + (\theta_{t} \cdot \Delta G_t)\\
&=& \theta_{t}\cdot ( S_{t-} + \Delta S_t + \Delta D_t)\\
&=& \theta_{t}\cdot ( S_{t} + \Delta D_t) =: \Psi_t(\theta)
\ees
where the first equality is by self-financing condition and (\ref{eq:JumpsOfSemiMgIntegral}), and  the second equality is by Remark \ref{rk:JumpsAtTMinus} and left continuity of the trading strategy: the wealth is conserved even in cases of dividend jumps. 
In framework (\ref{eq:PredictableSelffinancingSystem}), one can think to process $\varphi$ as a computation tool, where we set for any $t\geq 0$
\beq\label{eq:hToThetaMap}
\begin{cases}
\varphi_{t-}  \setto  \theta_t\\
\Delta \varphi_t\cdot S_{t} \setto  \theta_{t} \cdot \Delta D_t  
\end{cases}
\eeq
with initial condition $\varphi_0=\theta_0$: one can easily check that this system corresponds to (\ref{eq:DefSelfFinPort})-(\ref{eq:selffinWithJumps})-(\ref{eq:JumpDividendAndSelffinPort}). In the remainder of the section
we will have this last framework in mind  ($\varphi$ is a computation tool and $\theta$ is the trading strategy) but, with a slight abuse of notation, we will often refer to process $\varphi$ as the trading strategy.
\end{remark}


As we read in \cite{pp2001}, we must avoid problems that arise from the classical doubling strategy. Here, a player bets \$1 at a fair bet. If
he wins, he stops. If he loses, he next bets \$2. Whenever he wins, he stops, and his
profit is \$1. If he continues to lose, he continues to play, each time doubling his bet.
This strategy leads to a certain gain of \$1 without risk. However, the player needs to
be able to tolerate arbitrarily large losses before he might gain his certain profit. Of
course no one has such infinite resources to play such a game. Mathematically one
can eliminate this type of problem by requiring trading strategies to give martingales
that are bounded below by a constant. Thus the player’s resources, while they can be
huge, are nevertheless finite and bounded by a non-random constant. This leads to the
next definition.

\begin{definition}
A predictable self-financing strategy $\theta$ is said to be admissible on the time interval $[0,T]$ for $T>0$ if $\theta_0=\vec{0}$ and there
exists a constant $a$ such that $\Psi_{t}(\theta) \geq -a$, a.s.~for every $0\leq t \leq T$.
\end{definition}


\begin{definition}
An arbitrage opportunity on the time interval $[0,T]$ is an
admissible self-financing strategy $\theta$ such that $\Psi_{0}(\theta)=0$, $\Psi_{T}(\theta)\geq 0$ a.s.~and $\p(\Psi_{T}(\theta)>0)> 0$.
\end{definition}

\noindent Note that, since $\p\sim \q$, the last condition is equivalent to $\q(\Psi_{T}(\theta)>0)> 0$. A model is \emph{arbitrage free} if  there does not exist arbitrage opportunities in it. As anticipated in the previous section, the NFLVR condition is a slight/technical relaxation of the arbitrage free condition. We now want to re-state the portfolio and  the self-financing conditions in a more convenient way.

\begin{proposition}\label{prop:selffinancingPortfWithDivs}
Let $\Psi_t(\varphi)$ be defined in (\ref{eq:DefSelfFinPort}). Defining also $\widetilde{\Psi}_t:=\Psi_t/B_t$, 
then the self-financing conditions (\ref{eq:selffinWithJumps})-(\ref{eq:JumpDividendAndSelffinPort}) are equivalent to 
\beq\label{eq:dynOfPsi}
\diff \widetilde{\Psi}_t =  \varphi_{t-} \cdot \diff \widetilde{G}_t 
\eeq
where for any $i$, we define the deflated gain process
\be\label{eq:GTilde}
\widetilde{G}_t^i := \frac{S_t^i}{B_t} + \int_0^t \frac{\diff D_u^i}{B_u}.
\ee
We also recall that $S^0\equiv B$ and therefore $\widetilde{G}^0_t = 1$ for all $t$. 
\end{proposition}
\begin{proof}
This a particular case of the more general Lemma \ref{le:selfFinPortWithDivsGeneralEMM}: so we refer to the proof of this Lemma for the case $b=0$ and then $\beta\equiv B$.
Note that in this case some of the calculations are simplified from the fact that $B$ is continuous with finite variation.
\end{proof}

We start by using the bank account as numéraire. We concentrate on asset $S^i$ with $i\in\{1,\ldots,n\}$. Extending the analysis of \cite{Bjork}, our program is now as follows:
\begin{itemize}
\item Consider the \emph{buy-and-hold} self-financing portfolio where we hold exactly one unit of the
asset $S^i$, and invest all net dividends $D^i$ in the bank account. Denote the value
process of this portfolio by $Y^i_t:=\Psi_t(\varphi^{bh(i)})=\varphi^0_t B_t + S^i_t$ with $\varphi_t^{bh(i)}:=(\varphi^0_t, 0, 0, 1, 0, \ldots)$ where the unitary long position corresponds to the $i$-th asset recalling that $i>0$.
Hence, all the dividends dropped by the single asset are continuously rebalanced in the bank account position: at any $t$  in fact $Y^i_t =\varphi^0_t B_t + S^i_t$. 
\item  The point is now that the portfolio $Y^i$ can be viewed as a non-dividend
paying asset.
\item  Thus, the process $Y^i/B$ should be a local martingale under the Risk-Neutral measure.
\end{itemize}

Also, we make a standing assumption that the random variable $X = \int_0^T B_u^{-1}\diff D_u^i$ is $\q$-integrable.

\begin{lemma}\label{le:BDiscountedWealth}
Recalling the buy-and-hold portfolio definition $Y^i_t:=\varphi^0_t B_t + S^i_t$, define also its deflated version as $\widetilde{Y}_t^{i} :=\frac{Y^i_t}{B_t}$. Then the portfolio $Y^i$ 
respects the self-financing conditions (\ref{eq:selffinWithJumps})-(\ref{eq:JumpDividendAndSelffinPort}) if
\beq\label{eq:tildeYit}
\widetilde{Y}_t^{i} = \varphi^0_0+ \widetilde{G}_t^i
\eeq
where $\varphi^0_0$ is the initial cash endowment and
\beq\label{eq:h0t}
\varphi^0_t = \varphi^0_0+ \int_0^t \frac{\diff D_u^i}{B_u}.
\eeq
Moreover, $Y^i_t = \theta^0_t B_t + S^i_t +\Delta D^i_t$ where
\beqs
\theta^0_t =\theta^0_0+ \int_0^{t-} \frac{\diff D_u^i}{B_u}.
\eeqs
\end{lemma}
\begin{proof}
This is a particular case of the more general Lemma \ref{le:BuyAndHoldGeneralEMM}: so we refer to the proof of this Lemma for the case $b=0$ and then $\beta\equiv B$.
\end{proof}

\begin{remark}
Lemma \ref{le:BDiscountedWealth} (and its generalization of Lemma \ref{le:selfFinPortWithDivsGeneralEMM})
are measure independent and only due to the construction of a portfolio under the self-financing condition. The measure specification will be crucial instead for Theorem \ref{th:fftapWithDivs} 
(and its generalization in Theorem \ref{th:fftapWithDivsGeneralEMM}).
\end{remark}


\begin{theorem}[$\mathbf{1}^{\text{\textbf{st}}}$-FTAP with Dividends under the RN Measure]
For any $i$, we have that  $\widetilde{G}^i$ is $\q$-local martingale.
 \label{th:fftapWithDivs}
\end{theorem}
\begin{proof}
$\widetilde{Y}^{i}$ is a non-dividend paying asset deflated by the $\q$-measure numéraire. Hence, extending the market to this portfolio, due to Theorem \ref{th:fftap}, it must be a $\q$-local martingale. $\widetilde{G}^i$ 
is then equivalent to $\widetilde{Y}^{i}$ for $\varphi_0^0\setto 0$.
\end{proof}

\begin{remark}
We adopted a notation that is quite standardized but could be misleading: for a dividend paying asset $S^i$ we underline that $\widetilde{G}^i \neq \frac{G^i}{B}$ 
and that neither $\frac{S^i}{B}$ nor%
\footnote{%
One could naively think that $G^i:=S^i+D^i$ represented the position of being long on the asset and reinvesting its dividends: this is not correct since the right way to implement this strategy is 
described in Lemma \ref{le:BDiscountedWealth}. Moreover, as we will see in Proposition \ref{prop:RiskNeutralDrift}, under the Risk Neutral measure:
\[
\diff G^i_t := \diff S^i_t + \diff D^i_t = (r_t\,S_{t-}^i\diff t + \diff M^i_t -\diff D^i_t) + \diff D^i_t = r_t\,S_{t-}^i\diff t + \diff M^i_t \neq r_t\,G_{t-}^i\diff t + \diff M^i_t 
\]
and the last dynamics is the one that would guarantee that $\frac{G^i}{B}$ were a $\q$-local martingale.
}
 $\frac{G^i}{B}$ are Risk-Neutral local martingales as they would in cases where the asset had a null dividend process. 
\end{remark}


\begin{corollary}
$\widetilde{\Psi}_t$ is a local martingale under $\q$.
\end{corollary}
\begin{proof}
Recalling \eqref{eq:dynOfPsi}, then $\widetilde{\Psi}_t$ is a $\q$-local martingale being a sum of $\q$-local martingales $\widetilde{G}_t^{i}$ (by the previous theorem).
\end{proof}

\begin{remark}
In \cite{FPP19}, the authors prove that there are no sure profits via flash strategies if and
only if the deflated gain process components $\widetilde{G}^i$ do not have \emph{predictable} (resp.~\emph{fully predictable}) jumps, i.e.~if there does not exist any predictable time at which
the direction (and even the size, in the case of fully predictable jumps) of the jump is known just before the occurrence of the jump. Recalling the last result of Proposition \ref{prop:RiskNeutralDrift}, this is equivalent to assuming the absence of predictable jumps in all process $M$ components.
\end{remark}

The above results can be extended to any numéraire. From \cite{ap2011} (where the interested reader can refer for a presentation of the numéraire topic) 
we know that in order to be a numéraire a process should fit the following characteristics.

\begin{definition}[From \cite{ap2011}]\label{def:QPriceProcess}
The stochastic process $X$ is a \vvirg{$\q$-price process}, if 
\begin{romanlist}
\item $X_t>0$ for all $t$;
\item $X_t/B_t$ is a $\q$-strict martingale.
\end{romanlist}
\end{definition}

\noindent Let $\beta\equiv S^b$ for $b\in\{0,1,\ldots,n\}$ be a non-dividend paying (so $D^b=0$) semimartingale scalar process with $\beta_t>0$ a.s.~for all $t$. 
Then $\beta$ is a $\q$-local martingale due to Theorem \ref{th:fftapWithDivs}: we assume that $\beta$ is not only a local martingale but a martingale, 
so that $\beta$ is a \vvirg{$\q$-price process}. Then $\beta$ is eligible to be a numéraire, i.e.~we can define a Radon-Nikodym derivative
\[
L_t:=\e_t\left[\frac{\diff \q^\beta}{\diff \q\phantom{b}} \right] = \frac{\beta_t}{B_t}\,\frac{B_0}{\beta_0}
\] 
where $L$ is a $\q$-martingale with $\e[L_t]=L_0=1$ for any $t\geq 0$. Clearly under the degenerate case in which $b=0$ this is not a change of measure.

\begin{proposition}\label{le:selfFinPortWithDivsGeneralEMM}
Define $\widetilde{\Psi}_t^\beta:= \Psi_t/\beta_t$ where $\Psi_t$ is defined in (\ref{eq:DefSelfFinPort}). The self-financing conditions (\ref{eq:selffinWithJumps})-(\ref{eq:JumpDividendAndSelffinPort}) can also be written as:
\be\label{eq:dynOfPsiBeta}
\diff \widetilde{\Psi}_t^\beta = \varphi_{t-}\cdot \diff \widetilde{G}_t^{\beta} 
\ee
where $\widetilde{G}_t^{\beta}=(\widetilde{G}_t^{\beta,0},\widetilde{G}_t^{\beta,1},\ldots,\widetilde{G}_t^{\beta,n})$ and we define
\be\label{eq:GainProcessGeneralMeasure}
\widetilde{G}_t^{\beta,i} :=  \frac{S_t^i}{\beta_t}+\int_0^t \left\{\,\frac{\diff D_u^i}{\beta_{u-}}+\diff\left[ D^i, \frac{1}{\beta}\right]_u\,\right\} 
\ee
Clearly $\widetilde{G}_t^{\beta,b}=1$ for all $t$.
\end{proposition}
\begin{proof}
Note that $1/\beta$ is strictly positive since $\beta$ is strictly positive by hypothesis. Define $\TS^i:=S^i/\beta$ for any $i$, then
\[
\diff \TS^i_t = S^i_{t-} \diff (\beta^{-1})_t + \beta_{t-}^{-1}\diff S^i_t + \diff[S^i,\beta^{-1}]_t
\]
and so, using self-financing condition \eqref{eq:selffinWithJumps},
\bes
\diff \widetilde{\Psi}_t^\beta  &=& \Psi_{t-} \diff (\beta^{-1})_t + \beta_{t-}^{-1}\diff \Psi_t + [\diff\Psi_t,\diff(\beta^{-1})_t]\\
&=& \left(\varphi_{t-}\cdot S_{t-}\right) \diff (\beta^{-1})_t + \beta_{t-}^{-1} \, \varphi_{t-}\cdot \diff G_t +  \sum_{i=0}^n \varphi^i_{t-} \,[\diff G_t^i,\diff(\beta^{-1})_t]\\
&=&  \varphi^b_{t-} \Bigg\{ \beta_{t-} \diff (\beta^{-1})_t + \beta_{t-}^{-1}\diff \beta_t + [\diff \beta_t,\diff(\beta^{-1})_t] \Bigg\} \\
&\phantom{=}&+ \sum_{i=0, \,i\neq b}^n \varphi^i_{t-} \Bigg\{ S^i_{t-} \diff (\beta^{-1})_t  + \beta_{t-}^{-1} \diff (S^i+D^i)_t + [\diff (S^i+D^i)_t ,\diff(\beta^{-1})_t]\Bigg\}\\
&=& \sum_{i=0, \,i\neq b}^n \varphi^i_{t-} \Bigg\{ \diff \widetilde{S}^i_{t} + \beta_{t-}^{-1} \diff D^i_t + \diff[ D^i , \beta^{-1}]_t\Bigg\}\\
&=& \sum_{i=0, \,i\neq b}^n \varphi^i_{t-}  \diff \widetilde{G}_t^{\beta,i}
\ees
where at the forth equality we exploited the fact that in the first curly bracket we have $\diff (\beta\,\beta^{-1})_t = 0$ and therefore we obtain 
\eqref{eq:dynOfPsiBeta} -- recalling that $\diff \widetilde{G}_t^{\beta,b}=0$ for all $t$. Now we check ex-post that condition \eqref{eq:JumpDividendAndSelffinPort} holds: from the self-financing condition \eqref{eq:dynOfPsiBeta}, using \eqref{eq:linearityOfJumpOperator}-\eqref{eq:JumpsOfSemiMgIntegral}-\eqref{eq:JumpsOfQuadraticCov},
\bes
\Delta \widetilde{\Psi}_t^\beta  &=&  \varphi_{t-} \cdot \Delta \widetilde{G}_t^\beta := \sum_i \varphi_{t-}^i\, \Delta  \left( S_t^i\,\beta_t^{-1} +\int_0^t \Big\{\,\beta_{u-}^{-1}  \diff D_u^i +\diff\left[ D^i, \beta^{-1}\right]_u\,\Big\} \right)\\
&=& \sum_i \varphi_{t-}^i\, \left( S_{t-}^i \,\Delta (\beta_t^{-1}) + \beta_{t-}^{-1} \,\Delta S_{t}^i + \Delta S_t^{i} \,\Delta (\beta_t^{-1}) + \beta_{t-}^{-1}\,\Delta D_t^i + \Delta  D^i_t\,\Delta (\beta^{-1}_t)\right)\\
&=& (\varphi_{t-} \cdot S_{t-})\, \Delta (\beta_t^{-1}) + \beta_t^{-1} \,\varphi_{t-}\cdot( \Delta S_{t} + \Delta D_t)
\ees
on the other hand, using \eqref{eq:linearityOfJumpOperator}-\eqref{eq:jumpOperatorProductRule},
\bes
\Delta \widetilde{\Psi}_t^\beta &:=&  \Delta \left(\frac{\Psi_t}{\beta_t}\right) :=  \Delta \left(\beta_t^{-1}\, (\varphi_t\cdot S_t)\right)  = \beta_{t-}^{-1}\, \Delta \left( \varphi_t\cdot S_t\right) + (\varphi_{t-}\cdot S_{t-})\,\Delta (\beta_t^{-1}) + 
\Delta \left( \varphi_t\cdot S_t\right)\,\Delta (\beta_t^{-1})  \\
&=& \beta_{t}^{-1} \,( \varphi_{t-}\cdot\Delta S_t + S_{t-}\cdot\Delta \varphi_t + \Delta S_t\cdot\Delta \varphi_t) + (\varphi_{t-}\cdot S_{t-})\,\Delta (\beta_t^{-1}) \\
&=& \beta_{t}^{-1} \,( \varphi_{t-}\cdot\Delta S_t + S_{t}\cdot\Delta \varphi_t) + (\varphi_{t-}\cdot S_{t-})\,\Delta (\beta_t^{-1})   
\ees
and equating the last two equations we can observe that condition \eqref{eq:JumpDividendAndSelffinPort}  is respected.
\end{proof}

\begin{lemma}\label{le:BuyAndHoldGeneralEMM}
Define the buy-and-hold trading strategy $\mathcal{Y}^i_t:=\Psi_t(\varphi^{\beta h(i)})=\varphi^b_t \beta_t + S^i_t$ with $\varphi^{\beta h(i)}_t:=(0,0,1,\ldots, \varphi^b_t,0,\ldots)$ where the unitary long position is for the $i$-th asset and $i\neq b$. The discounted wealth $\widetilde{\mathcal{Y}}^{i}_t:=\frac{\mathcal{Y}^i_t}{\beta_t}$ is self-financing if it satisfies 
\[
\widetilde{\mathcal{Y}}_t^{i} = \varphi^b_0 + \widetilde{G}_t^{\beta,i}
\] 
where
\beq\label{eq:h0tbeta}
\varphi^b_t = \varphi^b_0 + \int_0^t \left\{\,\frac{\diff D_u^i}{\beta_{u-}}+\diff\left[ D^i, \frac{1}{\beta}\right]_u\,\right\}.
\eeq
Moreover, $\mathcal{Y}^i_t = \theta^b_t \beta_t + S^i_t +\Delta D^i_t$ where
\beqs
\theta^b_t =\theta^b_0+ \int_0^{t-}\left\{\,\frac{\diff D_u^i}{\beta_{u-}}+\diff\left[ D^i, \frac{1}{\beta}\right]_u\,\right\}.
\eeqs
\end{lemma}
\begin{proof}
We proved in the previous proposition the self-financing condition \eqref{eq:dynOfPsiBeta} for any self-financing trading strategy $\varphi$. In particular,
\beq\label{eq:TildeYBeta}
\diff \widetilde{\mathcal{Y}}_u^{i} = \diff \widetilde{\Psi}_u^\beta(\varphi^{\beta h(i)}) = \diff \widetilde{G}_u^{\beta,i} = \diff \left(\frac{S^i}{\beta}\right)_u + \frac{\diff D_u^i}{\beta_{u-}}+\diff\left[ D^i, \frac{1}{\beta}\right]_u
\eeq
recalling that since $\widetilde{G}_t^{\beta,b}=1$, then $\diff \widetilde{G}_t^{\beta,b}=0$. Integrating in $(0,t]$ on both sides we obtain 
\[
\widetilde{\mathcal{Y}}_t^{i} = \widetilde{\mathcal{Y}}_0^{i} + \widetilde{G}_t^{\beta,i}- \widetilde{G}_0^{\beta,i} = \widetilde{G}_t^{\beta,i} + \frac{\varphi^b_0 \beta_0 + S^i_0}{\beta_0} - \frac{S^i_0}{\beta_0} = \widetilde{G}_t^{\beta,i} + \varphi^b_0
\]
which is the thesis. On the other hand,
\bes
\diff \widetilde{\mathcal{Y}}_t^{i} &:=& \diff \left(\frac{\varphi^b \beta + S^i}{\beta}\right)_t = \diff \varphi_t^b + \diff \left(\frac{S^i}{\beta}\right)_t 
\ees
Comparing this with the self-financing condition \eqref{eq:TildeYBeta} gives
\[
 \diff \varphi_t^b =   \frac{\diff D_t^i}{\beta_{t-}}+\diff\left[ D^i, \frac{1}{\beta}\right]_t 
\]
so by integrating $\varphi^b_t = \varphi^b_0 + \int_0^t \diff \varphi^b_u$ we obtain \eqref{eq:h0tbeta}. The last statement is, recalling \eqref{eq:hToThetaMap}, since $\theta_t^{\beta h(i)}=\varphi_{t-}^{\beta h(i)}$. Also note that, using
\eqref{eq:JumpsOfSemiMgIntegral}-\eqref{eq:JumpsOfQuadraticCov},
\bes
\Delta \varphi_t^{\beta h(i)}\cdot S_{t} &=& \Delta \varphi^b_t\,\beta_t = \left\{\frac{\Delta D_t^i}{\beta_{t-}}+\Delta D^i_t \left(\frac{1}{\beta_t}-\frac{1}{\beta_{t-}}\right)\right\}\beta_{t}\\
&=&\beta_t\,\frac{\Delta D_t^i}{\beta_{t-}}+\beta_t\,\Delta D^i_t \left(\frac{\beta_{t-}-\beta_{t}}{\beta_t\,\beta_{t-}}\right) = \Delta D^i_t.
\ees
According to \eqref{eq:hToThetaMap}, the above quantity must be equal to $\theta^{\beta h(i)}_t\cdot \Delta D_t = \Delta D_t^i$ which is confirmed (recalling that $\beta$ does not pay dividends by construction).
\end{proof}

\begin{corollary} In case $D^i$ has finite variation, we have
\[
\widetilde{G}_t^{\beta,i} =  \frac{S_t^i}{\beta_t}+\int_0^t  \frac{\diff D^i_u}{\beta_{u-}} + \sum_{0<u\leq t} \Delta D^i_u\,\Delta (\beta^{-1})_u.
\] 
If, in addition, either $D$ or $\beta$ are continuous
\[
\widetilde{G}_t^{\beta,i} =  \frac{S_t^i}{\beta_t}+\int_0^t  \frac{\diff D^i_u}{\beta_{u-}}
\]
which is exactly the same formulation as that used for the Risk Neutral measure with a different deflator $\beta$.
\end{corollary}
\begin{proof}
Straightforward application of \eqref{eq:QCovOfFVProcess}.
\end{proof}

We can now extend the result of \cite{Bjork} pp.~244-245 -- which only tackles the continuous processes case -- to the general semimartingale case.

\begin{theorem}[$\mathbf{1}^{\text{\textbf{st}}}$-FTAP with Dividends under a General EMM]\label{th:fftapWithDivsGeneralEMM}
Denote by $\q^\beta$ the measure with numéraire $\beta$. Then, the deflated gain process $\widetilde{G}_t^{\beta,i}$, defined in (\ref{eq:GainProcessGeneralMeasure}), is a $\q^\beta$ local martingale.
\end{theorem}
\begin{proof}
The portfolio $\widetilde{\mathcal{Y}}_u^{i}$ of the non-dividend-paying-self-financing trading strategy $\varphi^{\beta h(i)}$ must be a $\q^\beta$-local martingale due to the previous considerations. By setting $\varphi^{b}_0\setto 0$, one obtains the thesis.
\end{proof}

\begin{corollary}
$\widetilde{\Psi}^\beta_t$ is a local martingale under $\q^\beta$.
\end{corollary}
\begin{proof}
We proved \eqref{eq:dynOfPsiBeta}. Therefore, $\widetilde{\Psi}_t^\beta$ is a $\q^\beta$-local martingale being a sum of $\q^\beta$-local martingales 
$\widetilde{G}_t^{\beta,i}$ (by the previous theorem).
\end{proof}

\subsubsection{Spot Price}

From now on until the end of Section \ref{sec:GenTheory}, in order to simplify the notation, we perform a small abuse of notation and we no longer interpret $S,\widetilde{G},D,q,\Phi...$ as vector processes but as generic scalar components $S^i,\widetilde{G}^i, D^i,q^i,\Phi^i...$ of these vectors. We detail in this section some results on assets spot price, straightforwardly obtained using the $1^{\text{st}}$-FTAP.

\begin{corollary}\label{cor:exDivCumDivPrices}
Under the Risk Neutral measure $\q$, under the hypothesis that $\widetilde{G}$ is a strict martingale (not only local), one obtains:
\be\label{eq:SmartDivs}
S_t = B_t \,\e_t\left[\frac{S_T}{B_T}+\int_t^T \frac{\diff D_u}{B_u}\right]
\ee
which says that today's spot price is equal to the expected value of all future discounted earnings arising from holding the stock: the sum of the discounted final value of the asset and of the discounted future flow of dividends.
More generally,
\be\label{eq:SexpectedValueWithGenMeas}
S_t = \beta_t \,\e_t^\beta\left[\frac{S_T}{\beta_T}+\int_t^T \left\{\,\frac{\diff D_u}{\beta_{u-}}
+\diff\left[ D,  \frac{1}{\beta}\right]_u\,\right\}\right]
\ee
\end{corollary}
\begin{proof}
All results are straightforwardly derived from the definitions of $\widetilde{G},\,\widetilde{G}^{\beta}$ and from the fact that $\widetilde{G}_t=\e_t[\widetilde{G}_T]$
and $\widetilde{G}_t^{\beta}=\e_t^\beta[\widetilde{G}_T^{\beta}]$.
\end{proof}

\begin{proposition}\label{prop:propDivsDiscount}
Let us assume (\ref{eq:RiskNeutralDynWithDivs})-(\ref{eq:RiskNeutraldriftWithDivs}) with
\[
\diff D_t \setto q_t\,S_t \diff t +\diff \Phi_t
\]
where $q$ is the stochastic proportional (to the asset value) dividend rate.
Moreover, $\Phi$ is any residual (with respect to the total dividend $D$) locally bounded semimartingale (but very often a pure jump process representing the cash dividends) with $\Phi_0=\Phi_{0-}=\vec{0}$. Defining $\mu := r-q$, we have that (not only $\widetilde{G}$ but also)
\[
\TX_t := \frac{S_t}{B_t^{\mu}} +\int_0^t \frac{\diff \Phi_u}{B^\mu_u} 
\]
is a $\q$-local martingale and, in case it is a strict martingale, for any $T>t$
\beq\label{eq:PriceFormulaWithChangeOfDiscountingRate}
S_t = B_t^\mu \,\e_t\left[\frac{S_T}{B_T^{\mu}}+\int_t^T \frac{\diff \Phi_u}{B_u^{\mu}}\right]
\eeq
\end{proposition}

\begin{proof}
Defining $\widetilde{S}:= S (B^\mu)^{-1}$ and using \eqref{eq:processDeflatedByBankAccount}-\eqref{eq:RiskNeutralDynWithDivs}-\eqref{eq:RiskNeutraldriftWithDivs}, we have
\bes
\diff \widetilde{S}_u &=& \frac{\diff S_u}{B^\mu_u} - \mu_u \frac{S_{u-}}{B^\mu_u} \diff u \\
&=& \frac{1}{B^\mu_u}\Big[ (r_u-q_u) S_{t-}\diff t +  \diff M_u -\diff \Phi_u- \mu_u  S_u \diff u\Big]\\
&=& \frac{1}{B^\mu_u}\Big[ \diff M_u -\diff \Phi_u \Big]
\ees
Integrating in $(0,t]$ on both sides
\[
\TX_t := \widetilde{S}_t +\int_0^t \frac{1}{B^\mu_u} \diff \Phi_u= \widetilde{S}_0 +\int_0^t \frac{1}{B^\mu_u} \diff M_u = \TX_0 +\int_0^t \frac{1}{B^\mu_u} \diff M_u
\]
where the last equality is since $\Phi_0=\Phi_{0-} =0$ by hypothesis. Moreover, taking the conditional expectation we obtain 
\[
\e_t[\TX_T] :=  \TX_0 + \e_t\left[ \int_0^T \frac{1}{B^\mu_u} \diff M_u \right] =  \TX_0 +   \int_0^t \frac{1}{B^\mu_u} \diff M_u  = \TX_t
\]
where we exploited Property 10 of Proposition \ref{prop:SemiMgIntegralProperties}. Hence, $\TX$ is a $\q$-local martingale. Substituting on both sides the definition of $\TX$ we obtain the second result.
Alternatively, in case the model is Markovian and all processes are continuous (so that all partial derivatives are well defined) and driven by a Brownian vector, this result can be derived applying the Feynman-Kac theorem
to expected value \eqref{eq:SmartDivs} to obtain the corresponding PDE. In this PDE perform a change of discounting rate by a change of the continuous payoff and then apply again the Feynman-Kac theorem to this modified PDE to obtain the expected value \eqref{eq:PriceFormulaWithChangeOfDiscountingRate}.
\end{proof}

\begin{remark}\label{rk:PropDivsToChangeOfDiscountingRate}
By (\ref{eq:SmartDivs}) and the above proposition
\[
S_t = B_t \,\e_t\left[\frac{S_T}{B_T}+\int_t^T \frac{\diff \Phi_u}{B_u} + \int_t^T \frac{q_u S_u \diff u}{B_u}\right] = B_t^\mu \,\e_t\left[\frac{S_T}{B_T^{\mu}}+\int_t^T \frac{\diff \Phi_u}{B_u^{\mu}}\right]
\]
which is a useful trick for asset pricing (especially when interpreting $S$ as a financial derivative, see Section \ref{sec:assetPricing}): it transforms a continuous dividend 
cash flow into a change of discounting rate.
\end{remark}

\begin{remark}
It is important to underline here that $B^\mu$ (as any other domestic bank account different from $B$) \emph{must be} only a mathematical tool and not a proper tradable asset, since its dynamics is written in any measure as $\diff B^\mu_t = \mu_t B^\mu_t \diff t$ and therefore, by absence of arbitrage, if it were a tradable asset it would drift with the risk-free interest rate $r$.\\
In particular, there is no way to change the measure to a measure, say $\q^\mu$, with numéraire $B^\mu$ since
\[
\frac{\diff \q^\mu}{\diff \q} \bigg|_{\ft}= \frac{B^\mu_t}{B_t} \frac{B_0}{B^\mu_0}
\]  
has no martingale part, therefore it could be a $\q$-martingale (and hence a \vvirg{$\q$-price process}) only if it were a constant, i.e.~in case $r\equiv\mu$ so that $\q\equiv \q^\mu$.
\end{remark}

\subsubsection{Forward Price}\label{sec:ForwardPrice}

We specify that the interest rates in this section are considered as stochastic.

\begin{definition}
The Forward price $F_t(T)$ is the par rate of a forward contract, i.e.~for $t\leq T$ the strike $K$ such that
\[
0=\e_t \left[\frac{B_t}{B_T}(S_T-K)\right]
\]
\end{definition}

\begin{proposition}
The Forward price can also be written as $F_t(T)=\e^T_t [S_T]$.
\end{proposition}
\begin{proof}
A straightforward change of measure,
\[
\e_t \left[\frac{B_t}{B_T}(S_T-K)\right]=\e_t^T\left[\frac{P_t(T)}{P_T(T)}(S_T-K)\right]=P_t(T)\,\Big\{\e_t^T \left[S_T\right]-K\Big\}
\]
from which we have the thesis.
\end{proof}

\begin{proposition}
The forward price with stochastic dividends and interest rates writes
\be\label{eq:FwdWithDivs}
F_t(T)=\frac{1}{P_t(T) } \left\{S_t - B_t\,\e_t\left[  \int_t^T \frac{\diff D_u}{B_u}\right]\right\}
\ee 
where one should note that inside the curly brackets we have only $\fst_t$-measurable quantities (the asset value minus the expected value of the discounted future dividends, not earned by the long 
forward contract holder) that are capitalized (through the $T$-zero coupon bond) to the date $T$ where the forward contract transaction takes place. 
\end{proposition}
\begin{proof}
Under the usual change of numéraire techniques, defining
\[
L_t = \frac{\diff \q^T}{\diff \q\phantom{a}}\Big |_{\fst_t} = \frac{P_t(T)}{B_t}\,\frac{B_0}{P_0(T)}
\]
we have
\bes
F_t(T)&=&\e^T_t[S_T]=\frac{\e_t\left[ L_T\,S_T\right]}{L_t} =  \frac{B_t}{P_t(T)} \,\e_t\left[ \frac{S_T}{B_T}\right]\\
&=& \frac{B_t}{P_t(T)} \,\e_t\left[ \widetilde{G}_T -\int_0^T \frac{\diff D_u}{B_u} \right]\\
&=& \frac{B_t}{P_t(T)} \left\{\, \widetilde{G}_t -\e_t\left[\int_0^T \frac{\diff D_u}{B_u} \right]\right\}\\
&=& \frac{B_t}{P_t(T)} \left\{\, \frac{S_t}{B_t} -\e_t\left[\int_t^T \frac{\diff D_u}{B_u} \right]\right\}
\ees
which is the thesis.
\end{proof}

The above formula can also be justified by replication arguments. We recall the par long forward contract transactions:
\begin{itemize}
		\item At $t$:
				\begin{itemize}
				\item Both parties agree to enter into the contract and set the par strike $K$, no cash transaction takes place (since $K$ is the par strike);
				\end{itemize}
		\item In the interval $(t,T)$:
				\begin{itemize}
				\item No transaction;
				\end{itemize}
		\item At $T$:
				\begin{itemize}
				\item The long forward contract holder obtains the asset (of value $S_T$) and pays $K$.				
				\end{itemize}
\end{itemize}
A static replication strategy can be built as follows:
\begin{itemize}
		\item At $t$:
				\begin{itemize}
				\item One buys at $t$ the asset for a value of $S_t$ in order to deliver it at maturity $T$: hence one holds the asset but has an outflow of cash of $S_t$;
				\item One sells in the market to a third counterparty (say $\mathcal{M}^D$) the future dividends of the asset in the interval $(t,T]$ via a dividend swap%
				\footnote{In general, in dividend swaps, a leg of future dividends of an asset is exchanged with a fixed rate or floating rate leg. In any case one can immediately  sell  this last leg and obtain $V_t^D(T)$ of cash.}, 
				obtaining as cash the value of expected future discounted (net) dividends, i.e.
					\be\label{eq:valueOfFutureDivs}
					V_t^D(T) := B_t\,\e_t\left[  \int_t^T \frac{\diff D_u}{B_u}\right] 
					\ee						
					Note that selling the dividend swap is necessary since, in general, future dividends are stochastic.  In case the dividends were deterministic one can simplify the replication strategy,  
					\[
					V_t^D(T) =  \int_t^T P_t(u) \diff D_u 
					\]
					and hence one can sell at inception, instead of a dividend swap, a strip of $u$-Zero Coupon bonds, with $u\in (t,T]$, and notional $\diff D_u$ obtaining at inception an amount $V_t^D(T)$ of cash (at any expiry $u$ the money obtained by the dividend $\diff D_u$
					is given as notional redemption to the $u$-Zero Coupon bond counterparty): see in particular formula \eqref{eq:DetermAbsDivsFwd} for the deterministic cash dividend case. 
				\item Since the forward contract has a zero cash transaction at $t$, one has to finance the above cash debt, this is implemented with selling a notional $N_t$ of a $T$-zero-coupon bond  for which
					\[
					S_t - V_t^D(T) = N_t \times P_t(T)
					\]
					where on the left-hand side we have the cash amount needed and on the right-hand side the asset value of the zero coupon bond (hence the cash obtained at $t$ from the short position on the bond).
				\end{itemize}
		\item In the interval $(t,T)$:
				\begin{itemize}
				\item For any $u$ in the interval, one receives the (net stochastic realized) dividend $\diff D_u$ from the asset and gives it to the counterparty $\mathcal{M}^D$ (hence all cash transactions offset);
				\end{itemize}
		\item At $T$:
				\begin{itemize}
				\item One still holds the asset (of value $S_T$);
				\item One has to deliver the notional $N_t$ to the $T$-Zero Coupon buyer.
				\end{itemize}
\end{itemize}

Since we have built a replication strategy, by arbitrage $K=N_t$, hence the forward par $F_t(T):= K = N_t$ and the zero coupon notional coincides with the right-hand side of formula \eqref{eq:FwdWithDivs}. 

\begin{proposition}
Fix the same setting/notation as in Proposition \ref{prop:propDivsDiscount} and add the hypothesis that the dividend rate $q$ is deterministic. We have
\beq\label{eq:fwdWithDeterminiticRates}
F_t(T)=\frac{S_t - \int_t^T P_t^\mu(u)\,\e_t^T\left[\diff \Phi_u\right]}{P_t^\mu(T)}
\eeq
where for any $u\geq t$ we define $P_t^\mu(u):=\e_t[B_t^\mu/B_u^\mu]=P_t(u)\,\esp^{\int_t^u q_s\diff s}$.
\end{proposition}
\begin{proof}
From \eqref{eq:PriceFormulaWithChangeOfDiscountingRate},
\bes
(S\, B^{q})_t &=& \,B_t \,\e_t\left[  \frac{(S\,B^{q})_T }{B_T }  +  \int_t^T  \frac{ B^{q}_u\diff \Phi_u }{B_u} \right]\\
 &=& \,P_t(T) \,B^{q}_T\, F_t(T)  +  \int_t^T B^{q}_u\,P_t(u)\,\e_t^T\left[\diff \Phi_u\right]
\ees
and the result is obtained with some algebra.
\end{proof}
To be more concrete, the interested reader can refer to Appendix \ref{app:BasicDividendModels} to explore two basic deterministic dividend models: the continuous proportional dividend and the cash dividend case.


\subsubsection{Dynamics\label{sec:dynamics}}


We now study the Risk-Neutral dynamics of a generic asset $S$.

\begin{proposition}\label{prop:RiskNeutralDrift}
Define $M$ as $\q$-local martingale process, and $A$ as a finite variation process with $A_0=M_0=0$. 
Let $D$ be a locally bounded semimartingale cumulative dividend process with $D_0=D_{0-}=0$. Assume the following general dynamics under $\q$:
\be\label{eq:RiskNeutralDynWithDivs}
\diff S_t = \diff A_t +  \diff M_t -\diff D_u
\ee
Observe that with this dynamics $S$ is indeed a semimartingale (being the sum of semimartingales) with dividend dropping $\diff D_u$ explained at Remark \ref{rk:DividendDropping}. Then this dynamics is arbitrage free (more precisely, it respects the NFLVR condition) if we set
\be\label{eq:RiskNeutraldriftWithDivs}
\diff A_t \setto r_t \,S_{t-}\diff t 
\ee
so that $\diff \widetilde{G}_t = \frac{ \diff M_t}{B_t}$.
\end{proposition}
\begin{proof}
For any scalar semimartingale $X$ we have from \eqref{eq:ItoWithJumpsOnProduct} and the fact that $B$ is continuous with finite variation (recalling \eqref{eq:QCovOfFVProcess}):
\be
\diff (X B^{-1})_t &=& B^{-1}_{t-} \diff X_t +X_{t-} \diff (B^{-1})_t + \diff [X, B^{-1}]_t \0\\
&=& B^{-1}_{t-} \diff X_t - X_{t-} \frac{1}{B_t^2}\diff B_t\0\\
&=& \frac{\diff X_t}{B_t} - r_{t} \frac{X_{t-}}{B_t}\diff t\label{eq:processDeflatedByBankAccount}
\ee
that can be used to obtain
\[
\diff \widetilde{G}_t := \diff \left(\frac{S}{B}\right)_t +\frac{\diff D_t}{B_t} = 
\frac{\diff A_t +\diff M_t -\diff D_t - r_t S_{t-}\diff t+\diff D_t}{B_t} = \frac{\diff A_t +\diff M_t  - r_t S_{t-}\diff t}{B_t}
\]
Now, $\widetilde{G}$ must be a $\q$-local martingale by Theorem \ref{th:fftapWithDivs}, hence, with $M$ already being a local martingale, we have the thesis.
\end{proof}


\begin{remark}\label{rk:DividendDropping}
The dynamics (\ref{eq:RiskNeutralDynWithDivs}) is coherent with the concept that, whenever the dividend is paid to the asset holder, the asset value experiences the opposite absolute shock: precisely
\[
\diff S_t = (\ldots) - \diff D_t
\]
and in particular, in case $(\ldots)$ is continuous, we have $\Delta S_t=-\Delta D_t$.  
However, one would expect that the drop is represented by the gross dividend and not by the net dividend: in fact one could argue that the company pays a part of the dividend to the shareholders and a part to the Government via taxes; the share value should drop by the total (i.e.~gross) dividend amount. 
First of all, we should check if the above result is influenced by guessing (\ref{eq:RiskNeutralDynWithDivs}): let us define the gross dividend process with $\widehat{D}$ and assume 
\bes
\diff S_t = \diff \widehat{A}_t +  \diff M_t -\diff \widehat{D}_u
\ees
then we obtain, following the same passages of the proof here above,
\bes
\diff \widehat{A}_t \setto r_t \,S_{t-}\diff t + (\diff \widehat{D}_u-\diff D_u) 
\ees
which is the same result of previous proposition. In fact, the drop with the gross dividend is allowed in the physical measure $\p$ but not in the Risk Neutral measure $\q$: the dividend taxes are (a negative) part of the physical drift that goes away in the passage to the Risk Neutral measure exactly as the expected return of the asset in the physical measure becomes the risk-free rate.
\end{remark}


\subsubsection{Non-Negative and Total Return Assets\label{sec:nonnegasset}}


We now concentrate on assets that bear no liability to the holder, such as a stock or an equity index. In this case, the spot price cannot go negative, and in particular neither can the forward price, because otherwise, as noted in \cite{bu10}, in a forward contract we would exchange a non-negative asset for a negative cash amount, leading to arbitrage opportunities. If we look back to dynamics (\ref{eq:RiskNeutralDynWithDivs})-(\ref{eq:RiskNeutraldriftWithDivs}), it is clear that the non-negativity of \(S\) depends on how \(M\) and \(D\) are modeled. In the following proposition, instead, we adopt a modeling representation originally designed in \cite{ehd2007}, then expanded in \cite{bu10} and finally generalized in \cite{bu18}, for which assets and forward prices are non-negative \emph{by construction}.
Let\footnote{We avoid the complication of managing the case $\widehat{T}\to +\infty$ since of no practical utility.} $t\in [0,\widehat{T}]$ and, recalling \eqref{eq:valueOfFutureDivs}, define
\be\label{eq:defS0star}
S_0^\star := S_0 - V_0^D(\widehat{T}) \; \stackrel{\text{if } \widetilde{G} \text{ is a strict mg., from } \eqref{eq:SmartDivs} }{=}\; B_0\,\e\left[\frac{S_{\widehat{T}}}{B_{\widehat{T}}}\right]
\ee
as the \vvirg{inner spot} of $S$, which represents the value of $S_0$ which is not due to future expected discounted dividends. We are now ready for the proposition, the interested reader can also refer to Appendix \ref{app:BasicDividendModels} for two simple concrete examples of the following framework. We remark that, in contrast to \cite{bu18}, we had to add the hypothesis of $D_t$ being $\fst_{t-}$-measurable (e.g.,~$D$ being predictable or with jumps driven by a Poisson process\footnote{See, for example, the first comment+answer of \href{https://almostsuremath.com/2016/11/22/predictable-processes/}{\texttt{this}} post on the \vvirg{Almost Sure} blog.}) which is an important restriction. On this topic, \cite{bu18} is not explicit, hence we cannot determine whether the author considered this problem, and if so, how it was addressed.

\begin{proposition}\label{prop:bu18}
Let $t\in [0,\widehat{T}]$. Let $D_t\geq 0$ and $D_t$ being $\fst_{t-}$-measurable for all $t$. Recall definition (\ref{eq:defS0star}), and assume $S_0^\star>0$ (by hypothesis or by expected dividends quotation at time zero). We have that:
\begin{romanlist}
\item For any non-negative asset price process $S$ which pays $D$, there exists a non-negative $\q$-(local) martingale $M^{S}$, with $M^{S}_0=1$, such that
\be\label{eq:hb2S}
S_t = S_0^\star \,B_t\, M^{S}_t + V_t^D(\widehat{T}) 
\ee
\item For any non-negative $\q$-(local) martingale $M^{S}$ with $M^{S}_0=1$, the asset price $S$ defined
by (\ref{eq:hb2S}) is non-negative and pays dividends $D$.
\item For any $t\in [0,\widehat{T}]$, $S_t\geq V_t^D(\widehat{T})$.
\item In case $M^S$ is a strict martingale, for any $t,T\in [0,\widehat{T}]$ with $t\leq T$,  the spot price $S_t$ follows (\ref{eq:SmartDivs}), while the forward price $F_t(T)$ follows (\ref{eq:FwdWithDivs}).
\item For any $t,T\in [0,\widehat{T}]$ with $t\leq T$, the forward price $F_t(T)$ is non-negative.
\item The SDE of $S$ is (\ref{eq:RiskNeutralDynWithDivs})-(\ref{eq:RiskNeutraldriftWithDivs}), where
\bes
\diff M_t &:=& B_t \diff \widetilde{G}_t = B_t \,S_0^\star\diff M^{S}_t + B_t \diff M^{D}_t\\
M^{D}_t &:=&  \frac{V_t^D(\widehat{T})}{B_t} + \int_0^t \frac{\diff D_u}{B_u}
\ees
and in particular $M^{D}$ is a non-negative $\q$-martingale with $M^{D}_0=V_0^D(\widehat{T})$.
\end{romanlist}
\end{proposition}

\begin{proof}
Define $\widetilde{S}:=S/B$, $\widetilde{V}^D(\widehat{T}):=V^D(\widehat{T})/B$. For item i) we follow quite closely the proof of \cite[Th.~5.1]{bu18}. We know from the $1^{\text{st}}$-FTAP that there exists a $\q$-(local) martingale $\widetilde{G}$, defined in \eqref{eq:GTilde}. Define
\[
M_t^S:=\frac{\widetilde{S}_t - \widetilde{V}_t^D(\widehat{T})}{S_0^\star} := \frac{\widetilde{G}_t-M^D_t}{S_0^\star}
\]
and observe that, 
\[
M_t^D :=  \widetilde{V}_t^D(\widehat{T})  + \int_0^t \frac{\diff D_u}{B_u} = \e_t \left [ \int_0^{\widehat{T}} \frac{\diff D_u}{B_u} \right]
\]
only depends on $t$ by expectation conditioning, hence it is a strict martingale by the tower property of conditional expectation. Therefore $M^S$, being the sum of two local martingales, is a local martingale. In particular, $M^S$ is a strict martingale if $\widetilde{G}$ is a strict martingale. It remains to show that $M^S$ is non-negative: to this end, assume that $(\tau_n)_{n\in \n}$ is an increasing sequence of finite stopping times, having values in $[0,\widehat{T}]$ and such that $\lim_{n\to \infty}\tau_n = \widehat{T}$, that localizes $\widetilde{G}$ (i.e.~such that all its $\tau_n$-stopped processes are strict martingales).
\[
0\leq \e_t[\widetilde{S}_{\tau_n\vee t}] = \e_t\left[\widetilde{G}_{\tau_n\vee t}-\int_0^{\tau_n\vee t}\frac{\diff D_u}{B_u}\right] = \widetilde{G}_{t}-\e_t\left[\int_0^{\tau_n\vee t}\frac{\diff D_u}{B_u}\right] = \widetilde{S}_{t}-\widetilde{V}_t^D(\tau_n\vee t)=: S_0^\star \,M^{n}_t,
\]
i.e.~$M^{n}_t \downarrow M^{S}_t$, which implies $M^{S}_t\geq 0$ by monotone convergence.\\  
\noindent ii) follows from the fact that, being $M^{S}\geq 0$, since by assumption $S^\star_0>0$ and $D\geq 0$, then $S\geq 0$. 
\noindent iii) follows from the fact that $S_0^\star>0$ by hypothesis and that $M^S$ is non-negative. iv) first part:
\bes
S_t = B_t \,\e_t\left[ S_0^\star\, M^{S}_T + \e_T\left[  \int_T^{\widehat{T}} \frac{\diff D_u}{B_u}\right] +\int_t^T \frac{\diff D_u }{B_u}\right] = 
B_t \,S_0^\star\,M^{S}_t + V_t(\widehat{T})
\ees
\noindent iv) second part:
\bes
F_t(T)&:=&\e_t^T[S_T] = \frac{B_t}{P_t(T)} \,\e_t\left[ \frac{S_T}{B_T}\right] = \frac{B_t}{P_t(T)} \,\e_t\left[ S_0^\star \, M^{S}_T + \e_T\left[  \int_T^{\widehat{T}} \frac{\diff D_u}{B_u}\right] \right]\\
&=& \frac{1}{P_t(T)} \,\left\{ B_t\,S_0^\star \, M^{S}_t + B_t\,\e_t\left[  \int_T^{\widehat{T}} \frac{\diff D_u}{B_u}\right] \right\}\\
&=& \frac{1}{P_t(T)} \,\left\{ S_t - V_t^D(\widehat{T}) + B_t\,\e_t\left[  \int_T^{\widehat{T}} \frac{\diff D_u}{B_u}\right] \right\} = \frac{1}{P_t(T)} \,\left\{ S_t - V_t^D(T)\right\}.
\ees
v) follows from iii) and iv). We finally prove vi):  we saw at item i) that $M^D$ is a strict martingale, moreover,
\[
\diff M_u^D = \diff \widetilde{V}_u^D(\widehat{T}) + \frac{\diff D_u}{B_u}
\]
which can be easily seen integrating both lhs and rhs in $[0,t]$. Now,
\[
\diff \widetilde{S}_t = S_0^\star \diff M^{S}_t  + \diff \widetilde{V}_t^D(\widehat{T}) = S_0^\star \diff M^{S}_t + \diff M_t^D -\frac{\diff D_t}{B_t}
\]
and the result is easily obtained observing that, by Itô formula,
\[
\diff \widetilde{S}_t = \frac{\diff S_t}{B_t} - r_{t} \frac{S_{t-}}{B_t} \diff t. 
\]
We finally check that the dividend detached by $S$ is effectively $D$, i.e.~that $\Delta S_t \stackrel{\Delta M^S_t=0}{=} \Delta V_t^D(\widehat{T}) = - \Delta D_t$:
\bes
\Delta V_t^D(\widehat{T}) &:=& V_t^D(\widehat{T}) - V_{t-}^D(\widehat{T}) := B_t\left\{\e_t\left[  \int_{(t,\widehat{T}]} \frac{\diff D_u}{B_u}\right] -  \e_{t-}\left[  \int_{(t-,\widehat{T}]} \frac{\diff D_u}{B_u}\right]\right\}\\
&=& B_t\left\{\e_t\left[  \int_{(t,\widehat{T}]} \frac{\diff D_u}{B_u}\right] -  \e_{t-}\left[ \int_{(t-,t]} \frac{\diff D_u}{B_u} + \int_{(t,\widehat{T}]} \frac{\diff D_u}{B_u}\right]\right\}\\
&=& B_t\left\{\e_t\left[  \int_{(t,\widehat{T}]} \frac{\diff D_u}{B_u}\right] -  \e_{t-}\left[ \frac{\Delta D_t}{B_t} + \int_{(t,\widehat{T}]} \frac{\diff D_u}{B_u}\right]\right\} = -\Delta D_t
\ees
where we used \eqref{eq:JumpsOfSemiMgIntegral}, and the last equality is due to the hypothesis that $D_t$ is $\fst_{t-}$-measurable. In this case, $\Delta M^{D}_t=0$ which gives $\Delta S_t\stackrel{\Delta M^S_t=0}{=}-\Delta D_t$ also starting from the SDE of $S$.
\end{proof}

We finally study the Total Return $S^{\text{TR}}$ of an asset $S$, being the self-financing portfolio with all null quantities except for the asset $S$ itself, obtained by reinvesting all dividends of $S$ in the portfolio. We recall that, in order to simplify the notation, with respect to Section \ref{sec:FirstThAssetPricing} we are performing a small abuse of notation and we no longer interpret $S,\widetilde{G},D,q,\Phi,...$ as vector processes but as generic scalar components $S^i,\widetilde{G}^i, D^i,q^i,\Phi^i,...$ of these vectors. The following preposition explains the derivation in Section 2.3 of \cite{bu10} (where the author sets $S^{\text{TR}}_0=1$) and extends it to any dividend model (not only cash dividends). 

\begin{proposition}[Total Return]
Recall the notation and results of Section \ref{sec:FirstThAssetPricing}. Let $S^{\text{TR}}_t=\varphi_t S_t = \theta_t (S_t+\Delta D_t)$ be a self-financing portfolio with $\varphi_0=\theta_0=1$ and $\theta_t=\varphi_{t-}$, representing the total return of the (scalar) asset $S$. Under $\q$, the no-arbitrage SDE of $S^{\text{TR}}$ coherent with Proposition \ref{prop:RiskNeutralDrift} is:
\[
\begin{cases}
\diff S^{\text{TR}}_t  = r_t \,S^{\text{TR}}_{t-}\diff t + S^{\text{TR}}_{t-}\,\frac{\diff M_t}{S_{t-}} &\qquad t>0\\
S^{\text{TR}}_0=S_0 
\end{cases}
\]
If moreover $M$ can be written in a such way that $\diff M_t = S_{t-}\diff M_t^1$ for some (local) martingale $M^1$, then the SDE of $S^{\text{TR}}$ no longer depends on $S$. In case of lump dividends $\diff D_t=\Delta D_t$ arriving at an increasing sequence of stopping times $(\tau_i)_{i\in \n}$, the solution of the SDE is
\[
S^{\text{TR}}_t = \varphi_{\tau_{N_t}}\,S_t= S_t\,\prod_{i=1}^{N_t}\left\{\frac{S_{\tau_i}+\Delta D_{\tau_i}}{S_{\tau_i}} \right\}
\]
where $N$ is the counting process associated with dividend times.
\label{prop:totalReturn}
\end{proposition}
\begin{proof}
Define $\widetilde{S}^{\text{TR}}:=S^{\text{TR}}/B$. The self-financing condition, recalling Proposition \ref{prop:selffinancingPortfWithDivs}, can be written as:
\bes
\diff \widetilde{S}^{\text{TR}}_t &=& \varphi_{t-}\diff \widetilde{G}_t = \varphi_{t-}\left[ \diff \left(\frac S B\right)_t +\frac{\diff D_t}{B_t}\right]\\
&=& \varphi_{t-}\left[ \frac{\diff S_t}{B_t} -r_t S_t\diff t +\frac{\diff D_t}{B_t}\right]=\varphi_{t-}\, \frac{\diff M_t}{B_t}
\ees
where, in the last equality, we used the SDE of Proposition \ref{prop:RiskNeutralDrift}. On the other hand, by Itô formula,
\[
\diff \widetilde{S}^{\text{TR}}_t = \frac{\diff S^{\text{TR}}_{t}}{B_t} - r_t \frac{S^{\text{TR}}_{t-}}{B_t}\diff t
\]
and comparing the two equalities above and using $\varphi_{t}=\frac{S^{\text{TR}}_{t}}{S_t}$, the desired result can be obtained straightforwardly. The last statement follows from the fact that the portfolio must be rebalanced only at dividend times to take into account lump dividends: the self-financing rebalancing at dividend time $\tau_i$ for $i>0$ writes
\[
\varphi_{i-1} ( S_i +\Delta D_i) = \varphi_{i}\,S_i 
\]
where the lhs corresponds to the portfolio before rebalancing and the rhs after it. Therefore we have a recursive equation for the portfolio quantity
\[
\varphi_{i} = \varphi_{i-1}\,\frac{ S_i +\Delta D_i}{S_i}
\]
One can proceed forward for $i=1,2,\ldots, N_t$ to obtain $\varphi_{t}=\varphi_{\tau_{N_t}}$.
\end{proof}


\subsubsection{Pricing\label{sec:assetPricing}}


Define $V$ as the price process of a financial derivative with expiry $T$, written on underlyings $S^1,\ldots, S^n$ where we denote with $\Pi$ the cumulative \vvirg{dividend} (read intermediate cash flows) process of the financial derivative and with 
$\phi_T= \phi(S^1_T, \ldots, S^n_T)$ the final payoff at expiry $T$, where $\phi: \re^n\mapsto \re$.
The financial derivative, being itself a tradable asset, should be priced in a way for which the extended market $B, S^1,\ldots, S^n, V$ is consistent with the $1^{\text{st}}$-FTAP, i.e.~such that 
\be\label{eq:gainProcOfDerivative}
\widetilde{G}^V_t := \frac{V_t}{B_t} + \int_0^t \frac{\diff \Pi_u}{B_u}
\ee
is a $\q$-local martingale, where $\widetilde{G}^V$ represents the deflated gain process of the financial derivative. 

In case $\widetilde{G}^V$ is a strict martingale, one can proceed as in \eqref{eq:SmartDivs} to obtain 
\be
V_t &=& B_t \,\e_t\left[\frac{V_T}{B_T}+\int_t^T \frac{\diff \Pi_u}{B_u}\right] \nonumber\\
&=& B_t \,\e_t\left[\frac{\phi_T}{B_T}+\int_t^T \frac{\diff \Pi_u}{B_u}\right]. \label{eq:EuropeanAssetPricing}
\ee
This is the main result one can use to price a financial derivative. Without loss of generality, one could set $\phi_T\setto 0$ and inflate the dividend process $\Pi$ with the expiration cash flow, obtaining a simpler formula
\[
V_t = B_t \,\e_t\left[ \int_t^T \frac{\diff \Pi_u}{B_u}\right]
\]
Specifically, as an example, one could write the dividend process of a strip of vanilla Call options at deterministic fixing times $0<T_1<T_2 < \ldots $ and payment time lag $\delta\geq 0$:
\[
\diff \Pi_u = \sum_j (S_{T_j}-K_j)^+ \diff \Theta_{T_j+\delta}(u)
\]
and observe that this process has finite variation (it depends on time only through the Heaviside functions). In particular, we have from Stieltjes integral properties
\bes
\Pi_t &=& \int_0^t \diff \Pi_u =  \sum_j (S_{T_j}-K_j)^+ \,\int_0^t \diff \Theta_{T_j+\delta}(u) = \sum_j (S_{T_j}-K_j)^+ \,\sum_{0<u\leq t}  \Delta \Theta_{T_j+\delta}(u) \\
&=& \sum_{\substack{j\\ 0<T_j+\delta\leq t}} (S_{T_j}-K_j)^+ = \sum_{\substack{j\\ 0<T_j\leq t-\delta}} (S_{T_j}-K_j)^+
\ees
from which it is clear that even if the process $S$ is optional, $\Pi$ is predictable for $\delta>0$. Moreover,
\bes
B_t\int_t^T \frac{\diff \Pi_u}{B_u} &=& B_t \sum_j (S_{T_j}-K_j)^+ \int_t^T \frac{\diff \Theta_{T_j+\delta}(u)}{B_u} = B_t\sum_j (S_{T_j}-K_j)^+ \sum_{t<u\leq T} \frac{\Delta \Theta_{T_j+\delta}(u)}{B_u}\\
&=&B_t \sum_{\substack{j\\ t<T_j+\delta\leq T}} \frac{(S_{T_j}-K_j)^+ }{B_{T_j+\delta}} 
\ees

More generally, under measure $\q^\beta$ with numéraire $\beta$, using \eqref{eq:SexpectedValueWithGenMeas},
\be
V_t = \beta_t \,\e_t^\beta\left[\frac{\phi_T}{\beta_T}+\int_t^T \left\{\frac{\diff \Pi_u}{\beta_{u-}}
+\diff\left[ \Pi,  \frac{1}{\beta}\right]_u\right\}\right]
\ee

\section{Collateralized Derivatives\label{sec:CollaterDerivatives}}


\subsection{Introduction to Collateralization Modeling}

In this section, we detail the approach of \cite{MP2017} for collateralization modeling in the easiest case where all cash flows are in domestic currency. We will remove this hypothesis and be more technical in the next section. 
We deal with a financial derivative with price $V_u$ at time $u$.

In order to mitigate the counterparty risk, from contract inception to the time in which the contract is closed, the deal is collateralized (recall Section \ref{sec:Intro} for a brief non-technical presentation of this topic).
For this purpose, the most common agreement (we will see other examples in the next section) is a bilateral agreement documented by the International Swaps and Derivatives Association (ISDA), known as Credit Support Annex (CSA). In particular, the agreement also regulates the possibility of re-hypothecation of the collateral assets, namely to use them for funding purposes, as opposed to segregation.

We define with $C_u$ the collateral value at time $u$, with the convention (from the investor point of view) that $C_u>0$ means that the collateral is held by the investor, otherwise the collateral is held by the counterparty.  $C_u>0$ means that the financial derivative has a positive value for the investor and for this reason the investor holds (with or without re-hypothecation) some cash (or liquid securities) as a loan from the counterparty: in case of default of the counterparty this loan will mitigate the possible losses of the investor on the financial derivative position. In the opposite case, $C_u<0$ means that the financial derivative has a negative value for the investor, and for this reason the counterparty borrows some collateral from the investor. In the same spirit, $\diff C_u>0$ means that the counterparty posts collateral at time $u$, while $\diff C_u<0$ means that the investor posts collateral.

We define also the target collateral value with: 
\be\label{eq:CollateralWithHaircut}
C_{u}^\star:=(1+\alpha_{u}) V_{u-}
\ee
where $\alpha_u\geq 0$ is the (predictable) \emph{haircut} or \emph{proportional margin} process: to mitigate the risk of loss, borrowers are required to post up collateral in excess of the 
market value of the primary transaction $V$. The left limit in this formula is to ensure the predictability of the collateral process (the collateral posting party cannot be aware of unpredictable jumps of the financial derivative).

The over-collateralization (i.e.~the fact that the collateral value is greater than the financial derivative price) is meant to avoid losses in case of default of one counterparty in the period between two collateral posting times, losses due to changes in market value of $V$ itself and/or of any 
collateral security different than domestic cash. In particular, a defaulted counterparty stops posting collateral: we denote with $\tau$ the first-to-default time among all counterparties and with $\tau+\delta$ the time of the bankruptcy procedure closeout cash flow payments (generally $\delta$ is about two weeks) and with $t_N$ the last margin call date (see below for details), where $t_N \leq \tau \leq \tau+\delta$. 
Following \cite{ccb2020}, one calls the \emph{margin period of risk} the time lag between $t_N$ and $\tau+\delta$: the \emph{cure period} constitutes the second part of the margin period of risk (i.e.~the time between $\tau$ and $\tau+\delta$), the first part being the time lag between $t_N$ and $\tau$. 
Moreover, as anticipated above, when securities are used as collateral, their market value can fall in the margin period of risk and therefore they are  generally devalued to have a cushion for this phenomenon:
in this section we will  deal only with cash collateralization, see Section \ref{sec:repo} for a concrete example of stock collateralization. 

We have the following specifications:
\begin{itemize}
\item \textbf{Margin Calls}: at any time $u$ the collateral value $C_u$ can move away from its target value $C_u^\star$ because of financial derivative and/or collateral market moves: 
then with a contractual frequency (generally every business day of contract life) the counterparties 
check that the absolute value of the difference between the collateral value and its target value is inside of an agreed threshold or percentage: this is in order to reduce the administrative burden. 
When this limit is broken, the counterparty who is exposed to this breach performs a margin call, meaning that $u=t_i$ for some $i\in \{1, 2, \ldots, N\}$ where we define with $t_0<t_1 < \ldots < t_N$ the \emph{margin (call) times}. 
At $t_i$ the collateral value 
has a movement $\diff C_{t_i}$ that makes $C_{t_i}\setto C_{t_i}^\star$. By construction, the contract inception date coincides with the first margin date $t_0$. 
We refer to Section 3.3 of \cite{ccb2020} for more precise details on collateralization schemes. 
\item \textbf{Variation Margin (}$\boldsymbol{\mathfrak{M}}$\textbf{)}: it is represented by the first addend of \eqref{eq:CollateralWithHaircut} meaning that the target value of this margin account is such that $\mathfrak{M}_u^\star=V_{u-}$  which is the \vvirg{fair} collateral value (since the collateral is covering the financial derivative counterparty risk). Very often the counterparties agree to make the Variation Margin subject to re-hypothecation.
\item \textbf{Initial Margin (}$\boldsymbol{\mathfrak{I}}$\textbf{) versus Haircut}: it is the second addend of \eqref{eq:CollateralWithHaircut} representing a buffer of over-collateralization for all the above mentioned reasons.
The over-collateralization target value can be rough and represented by a second addend $\alpha_{t_0} V_{u-}$ where $\alpha_{t_0}$ is a fixed percentage decided by the counterparties at contract inception: 
this is a haircut the a strict sense (see e.g.~the Repurchase Agreement contract case at Section \ref{sec:repo}). In other cases, the second addend of \eqref{eq:CollateralWithHaircut} represents another 
(very often segregated, i.e.~not subject to re-hypothecation) account
called Initial Margin  (also known as the Independent Amount) and $\mathfrak{I}^\star_u=\alpha_u\,V_{u-}$ represents a quite complicated calculation, e.g.~in case of ISDA SIMM methodology%
\footnote{%
As we read in \href{https://www.risk.net/definition/isda-simm}{\texttt{Risk.net}}, the Standard Initial Margin Model (SIMM) is a common methodology to help market participants calculate initial margin on non-cleared derivatives under the framework developed by the Basel Committee on Banking Supervision and the International Organization of Securities Commissions.\\ 
The SIMM methodology was developed by ISDA, and is intended to reduce the potential for disputes and create efficiency through netting of exposures. The model applies a sensitivity-based calculation across four product groups: interest rates and foreign exchange, credit, equity and commodities.\\
The SIMM was officially launched in September 2016 and an updated version, ISDA SIMM 2.0, became effective in December 2017 to include a range of clarifications, enhancements and additional risk factors. 
We refer to the ISDA website and official documents therein for further details.
}. We do not enter in such technical details and we will consider both the Variation Margin and the Initial Margin (and this is a major simplification) as subject to re-hypothecation (see in particular Remark \ref{rk:Rehypothecation}).
\end{itemize}


Let us see a simple example to understand the collateral  mechanics: imagine the investor buys a Call option (of price process $V$) at $t_0$ from the counterparty. The investor pays to the counterparty $V_{t_0}$ of cash and writes in her book the long position on the Call for the same value. At the same time, the counterparty lends an amount $C_{t_0}$ to the investor to mitigate his counterparty risk. 
Thus, at $t_0$, the investor books the Call option, receives $C_{t_0}-V_{t_0}$ of cash from the counterparty and writes a debt of $C_{t_0}$ for the same reason. At $t_1$ the Call value is $V_{t_1}$ hence the investor/counterparty posts the required collateral in order to set $C_{t_1}\setto C_{t_1}^\star$, moreover, the investor must pay some interests on the debt, linked to the interest rate  $c$. Consequently, the cash to post at $t_1$ (if positive by the counterparty, if negative by the investor) is
\[
C_{t_{1}} - C_{t_{0}}(1+c_{t_{0}}(t_1-t_0))
\]
At $t_2$ the Call value is $V_{t_2}$ and the amount to post to have $C_{t_2}\setto C_{t_2}^\star$ is
\[
C_{t_{2}} - C_{t_{1}}(1+c_{t_{1}}(t_2-t_1))
\]
Imagine at $t_2$ the investor closes the Call position selling back the Call option to the counterparty (in case $t_2$ is the expiry of the Call option the option value is $V_{t_2}=(S_{t_2}-K)^+$ meaning that the investor obtains the payoff of the Call): the investor erases the Call from her portfolio, she should receive $V_{t_2}$ of cash from the counterparty but, at the same time, she should reimburse the debt of $C_{t_2}$ with the counterparty. The net cash flow to the investor when closing the financial derivative and the collateral position is thus $V_{t_2}-C_{t_2}$ of cash (equal to zero in case $\alpha=0$). Generalizing, for $i>0$,
\[
C_{t_{i}} - C_{t_{i-1}}(1+c_{t_{i-1}}(t_i-t_{i-1}))= C_{t_{i}} - C_{t_{i-1}} - c_{t_{i-1}} C_{t_{i-1}}(t_i-t_{i-1}) \qquad \longrightarrow \qquad \diff C_u - c_u \,C_{u-} \diff u.
\]
The implications of the above arrow are discussed in the next remark. 

\begin{remark}\label{rk:DiscreteCollateralization}
More precisely, as we will see in the next section, in this paper there is no \vvirg{continuous time approximation} on collateral posting: this is an advantage with respect to the setting of \cite{MP2017}.
In fact, in following sections $C$ will be any predictable semimartingale: this is coherent with a discontinuous LCRL process or also with a pure jump RCLL process $C_t = \sum_{0<u\leq t} \Delta C_u$ with predictable jumps. 
In the latter case, $\diff C_u=\Delta C_u$ and we also set $\Delta C_u=0$ except at collateral posting times $\Delta C_{t_i}=C_{t_{i}} - C_{t_{i}-}$ for $i>0$ where the $t_1, t_2,\ldots$ are the (stochastic) jump times of process $C$. In this case, setting also the collateral interest rate process $c$ as a (stochastic) step process with steps at times $t_i$, we obtain
\bes
\int_0^t \Big\{\diff C_u - c_u\,C_{u-}\diff u\Big\} &=& \sum_i \left\{\int_{(t_{i-1}, t_i]} \diff C_u - c_u\,C_{u-}\diff u \right\}\\ 
&=& \sum_i\left\{ \sum_{t_{i-1} <u\leq t_i} \Delta C_u - c_{t_{i-1}}\,C_{t_{i-1}} \int_{t_{i-1}}^{t_i} \diff u\right\}\\ 
&=& \sum_i C_{t_{i}} - C_{t_{i-1}}\Big(1 + c_{t_{i-1}}\,(t_{i}-t_{i-1})\Big)
\ees
so we have perfect discrete collateral posting. With the same setting, the discounted version 
\bes
\int_0^t \frac{ \diff C_u - c_u\,C_{u-}\diff u }{B_u} &=& \sum_i\left\{ \sum_{t_{i-1} <u\leq t_i} \frac{\Delta C_u}{B_u} - c_{t_{i-1}}\,C_{t_{i-1}} \int_{t_{i-1}}^{t_i}
\frac{1}{B_u}\diff u\right\}\\
&=& \sum_i\left\{  \frac{C_{t_i}-C_{t_{i-1}}}{B_{t_{i}}} - c_{t_{i-1}}\,C_{t_{i-1}} \int_{t_{i-1}}^{t_i}
\frac{1}{B_u}\diff u\right\}\\
\ees
which is different from the desired result of
\[
\sum_i   \frac{C_{t_i}-C_{t_{i-1}}- c_{t_{i-1}}\,C_{t_{i-1}}(t_{i}-t_{i-1})}{B_{t_{i}}} 
\]
and this is the only approximation applied in case of discrete collateralization.
\end{remark}

We define 
\[
\diff D_u^{\textit{V-C}} :=\diff \Pi_u + \diff C_u - c_u\, C_{u-} \diff u
\]
where $\diff \Pi_u$ represents the financial derivative dividend/cash flow at time $u$, one should note that:
\begin{arabiclist}
\item The above formulation of dividend process implies that the collateral is subject to re-hypothecation: see in particular Remark \ref{rk:Rehypothecation}.
\item The first addend has generally finite variation (see next section) but the second addend could have non-finite variation (e.g.~in case $C_u=C_u^\star$ for any $u$);
\item As outlined in the previous remark, the collateral process is intrinsically discontinuous since the equality $C_u\setto C_u^\star$ 
is implemented only on margin call times $t_i$'s.
\end{arabiclist}

We now write the Profit and Loss ($P\&L$) discounted at inception date $t$ of all cash transactions consisting in buying the collateralized derivative at time $t$ and holding it until time $T$ with 
$T>t$:
\[
P\&L_t = - (V_t-C_t)   +B_t \int_t^T \frac{\diff D^{\textit{V-C}}_u}{B_u} +\frac{B_t}{B_T} \,(V_T-C_T) 
\]
where the first added is the inception transaction (pay the financial derivative price and obtain the collateral value), 
the second added represents the discounted intermediate cash flows of the position and the third represents the closing of all positions (sell the financial derivative obtaining its price and give back the collateral). In order to avoid arbitrage, we must have that $\e_t[P\&L_t]=0$ meaning that the inception transaction is the equilibrium transaction corresponding to the $t$-Risk-Neutral discounted expectation of all other future transactions. 
In particular, $\e_t[P\&L_t]=0$ if and only if 
\[
(V-C)_t = B_t \,\e_t\left[\frac{(V-C)_T}{B_T}  + \int_t^T \frac{\diff D^{\textit{V-C}}_u}{B_u}   \right]
\]
which in turn corresponds to require that the deflated gain process defined as follows
\[
\widetilde{G}_u^{\textit{V-C}} :=  \frac{(V-C)_u}{B_u} + \int_0^u \frac{\diff D_s^{\textit{V-C}}}{B_s} 
\] 
must be a Risk-Neutral martingale. This is the usual condition of Theorem \ref{th:fftapWithDivs} for asset $X:=(V-C)$ with dividend process $D^{X}$ (and this explains the weird notation of the dividend process).
The gain process of the collateralized derivative is hence $G^X := X+ D^{X}$.

Now define the buy-and-hold portfolio $\Psi_u:=X_u+\varphi^0_u B_u$ in the interval $u\in [t,T]$. As usual, all the dividends of the position are reinvested in the bank account:
if $\diff D_u^{X}>0$ the investor uses this cash in order to decrease the bank account position, i.e.~she buys a quantity $\diff \varphi^0_u =\diff D^{X}_u/B_u>0$ of bank account (paying $B_u \diff \varphi^0_u $ of cash). 
If $\diff D_u^{X}<0$ she funds this transaction with selling a quantity $\diff \varphi^0_u=\diff D^{X}_u/B_u<0$ of bank account (receiving $B_u \diff \varphi^0_u $ of cash).
In this light,  we thus obtain the self-financing condition of the buy-and-hold strategy on the collateralized derivative $X$:
\[ 
\diff \Psi_u \stackrel{\textit{Itô}}{=} \diff X_u + \varphi^0_{u-} \diff B_u + B_u \diff \varphi^0_u \stackrel{\textit{self-fin.}}{=} \diff X_u + \varphi^0_{u-} \diff B_u +\diff D^{X}_u=:\diff G^X_u + \varphi^0_{u-} \diff B_u
\]
where the first equality is by the Itô Formula and the fact that $B$ is continuous with finite variation, the second is by substitution of the aforementioned dividend reinvestment strategy, 
the third equality is by definition of the deflated gain process here above. Moreover, for $u\geq t$
the bank account quantity writes
\[
\varphi^0_u = \varphi^0_t + \int_t^u \diff \varphi^0_s = \varphi^0_t + \int_t^u \frac{\diff D^{X}_s}{B_s} 
\]
and the bank account position is of value $\varphi^0_u B_u$: be aware of the analogies with Lemma \ref{le:BDiscountedWealth}. 

As anticipated in Section \ref{sec:Intro}, the common thread of the paper is to identify, thanks to no-arbitrage conditions, Risk-Neutral martingales in progressively more challenging contexts where our intuition could be increasingly lost: 
as a first example we saw the martingale  corresponding to a non-dividend paying asset, then we have attempted to recognize the martingale corresponding to a dividend-paying asset and  
here we move a step forward for detecting a martingale linked to a collateralized derivative. In the next section at \eqref{eq:DeflatedGainOftheDerivative}, we will prove (in a more general multi-currency environment) that, not only   
$\widetilde{G}_u^{\textit{V-C}}$ but also
\[
\widetilde{G}_u := \frac{V_u}{B_u}  + \int_0^u \frac{\diff D_s}{B_s}, \qquad \diff D_u := \diff \Pi_u -  C_{u-}  (c_u - r_u) \diff u
\]
is a Risk-Neutral Martingale. We will continue using this demonstration strategy and recognize new martingales in Section \ref{sec:Applications}.

All these efforts in searching for martingales are motivated from the fact that 
martingales have some nice properties and, primarily, since they have well defined dynamical features and specific connections with expected values: among all expected values we are particularly interested in pricing ones. 
For example, from $\widetilde{G}_t=\e_t[\widetilde{G}_T]$, we have
\[
V_t   = B_t \,\e_t\left[  \frac{V_T }{B_T} + \int_t^T \frac{\diff D_u}{B_u} \right]
\]
which corresponds to \eqref{eq:DerValueWithFXandCollateral} in the present single currency environment.

\subsection{Multi-Currency Collateralized Derivatives\label{sec:multiccyCollDeriv}}

We can generalize the previous section results to a multi-currency setup: denote with $t\geq 0$ today's time and with $V^f_t$ today's value of the financial derivative with start date 0\footnote{In order to avoid burdening the notation, we analyze the case of a spot/past start deal: the forward start case would bring only slight modifications of the following results. In fact, if one is willing to tackle all cases, then one should let the start date be $t_0<T$ and change all integration domains from $(t,T]$ to $(\max\{t,t_0\},T]$.} and expiry $T$ (for $t\leq T$) denominated in the $f$-currency as the cumulative dividend account
\[
\Pi_t^f := \int_0^t \psi^f_u\diff u + \Phi_t^f
\]
with continuous dividend rate process $\psi^f$ and (RCLL/finite variation process) cash dividends/cash-flow account
\[
\Phi_t^f:=\sum_i \phi_{T_i}^f \1_{ 0 < T_i \leq t}
\]
both denominated in the $f$-currency. Note that, thanks to the arguments of Section \ref{sec:assetPricing}, even if we describe $V^f$ as the price of a financial derivative,
all the results of this section will be valid even if $V^f$ were the value of a collateralized underlying of the model where $T$ is not the expiry of the underlying but a generic future time.
This is why we keep the final condition $V^f_T=\phi^f_T$ as a placeholder of the value of the underlying at a future time, even if we could have set $\phi^f_T=0$ and charge the last cash-flow on the cash dividend flow. 
In particular, the dividend paid at time $u$ writes
\[
\diff \Pi_u^f =  \psi^f_u \diff u + \diff \Phi_t^f = \psi^f_u \diff u + \sum_i \phi_{T_i}^f \diff \Theta_{T_i}(u)
\] 


The collateral account $C^g$ is instead in the $g$-currency. Denote with $d$ the domestic currency and, for currencies $x,y\in \{d,f,g,\ldots\}$ denote with $X^{xy}$
the FX rate to convert an amount denoted in currency $x$ into currency $y$, and with $X^{x}:=X^{xd}$  the FX rate to convert an amount denoted in foreign currency $x$ into the domestic currency $d$: we adopt the convention of
\cite{MP2017} for which the currency $d$ is \vvirg{silent}, hence 
$B\equiv B^d$ is the domestic bank account and $\q\equiv \q^d$ the domestic measure. Clearly $X^{xx}\equiv 1$ and $X^{xy}\equiv 1/X^{yx}$.
Moreover, for any foreign currency $x$, $X^x$ is a semimartingale with dynamics under $\q$
\be\label{eq:DynXg}
\diff X_t^x = \mu^x_t X_{t-}^x \diff t+\diff M_t^x 
\ee
where for all $x$ we have that $M^x$ is a RCLL square-integrable $\q$-martingale with $M^x_0=0$. Clearly $X^d_t=1$ for all $t$, so $\mu^d_t=M_t^d=0$ for all $t$.

Again from \cite{MP2017}, for currency $x$, recall the definition of the $x$-currency basis measure $\q^{xb}$ as the measure with the numéraire
corresponding to the $x$-currency basis bank account $B^{xb}:=B^{r^{xb}}$ where we define the $x$-basis spot risk-free interest rate
\beq
r^{xb}:= r-\mu^x
\eeq
and by direct substitution $B:= B^r = B^{db}$ so that clearly $\q\equiv\q^{d}\equiv\q^{db}$. Define also for any $t\leq U$ the $x$-currency basis zero coupon bond 
\[
P^{xb}_t(U):=\e_t^{xb}\left[\frac{B^{xb}_t}{B^{xb}_U}\right]
\]
where we denote with $\e_t^{xb}[\cdot]$ the expectation under measure $\q^{xb}$. In case $x\neq d$, the curve $U\mapsto P^{xb}_t(U)$ is often called \vvirg{Forex (Discount) Curve}: see in particular \eqref{eq:fxforward}. 
Observe that without FX market dislocations (see the second part of \cite{MP2017} for an explanation of these frictions on the FX market), we obtain the classical FX drift rate $\mu^x=r-r^x$ (where the right-hand side corresponds to the difference of spot risk-free interest rates of corresponding currencies) and hence $r^{xb}=r^x$ , then also $\q^{xb}=\q^{x}$ (the Risk Neutral measure of currency $x$): in practice  all $b$-superscripts can be canceled out. 

For any currency couple $x,y$ define with 
\beq\label{eq:defBetaBankAccount}
\beta^{yx}:=  B^{yb} X^{yx}  =  \frac{ B^{yb}}{ X^{xy} }
\eeq
the $y$-basis bank account converted in currency $x$, and the Radon–Nikodym derivative
\[
\e^{xb}_t\left[\frac{\diff \q^{yb}}{\diff \q^{xb}} \right] = \frac{\beta^{yx}_t}{B^{xb}_t}\,\frac{B^{xb}_0}{\beta^{yx}_0} 
\]
i.e.~$\beta^{yx}$ is the deflator to \vvirg{export} an $x$-denoted security in the $\q^{yb}$-economy: the passage $\q^{xb}\mapsto \q^{yb}$ is equivalent to the deflator passage $B^{xb}\mapsto \beta^{yx}$. In particular, recalling Definition \ref{def:QPriceProcess}, the fact that $\beta^{yd}=B^{yb} X^{y}$ is a \vvirg{$\q$-price process} (hence eligible to be a numéraire) can be directly proved using \eqref{eq:DynXg} and \eqref{eq:ItoWithJumpsOnProduct}: 
we will confirm the more general result that $\beta^{yx}$ is a \vvirg{$\q^{xb}$-price process} (and hence that the above change of numéraire is valid) ex-post at Corollary \ref{cor:betaxfIsaQfPriceProcess}.

We  are now ready to extend the heuristical reasoning of the previous section to the current setup with a more precise methodology: the following proposition corresponds to  Proposition 2.1 of \cite{MP2017} in a more general semimartingale setting.

\begin{proposition}\label{prop:MuticurrencyCollateralPricing}
With the above notation, defining also
\be
V_t:=X^{f}_t V_t^f,  \quad \widetilde{V}_t  :=\frac{V_t}{B_t}, \quad C_t:=X^{g}_t C_t^g, \quad \widetilde{C}_t  :=\frac{C_t}{B_t} 
\ee
and
\be\label{eq:diffDVMinusC}
\diff D_u^{\textit{V-C}} :=  X^f_{u} \diff \Pi_u^f + X^g_{u-} \left\{\diff C_u^g - c_u^g \,C_{u-}^g \diff u +   \left[ \diff C^g_u,   \frac{\diff X^g_u}{X^g_{u-} } \right ]\right\}
\ee
with $D_0^{\textit{V-C}}=D_{0-}^{\textit{V-C}}=0$, then the gain process of the collateralized financial derivative $(V-C)$ under collateral re-hypothecation is defined as $G_t^{\textit{V-C}}:=(V-C)_t+D_t^{\textit{V-C}}$ and its deflated version
\be\label{eq:defGTilde}
\widetilde{G}_t^{\textit{V-C}} := \widetilde{V}_t -\widetilde{C}_t + \int_0^t \frac{\diff D_u^{\textit{V-C}} }{B_u}
\ee
is a martingale under the domestic Risk-Neutral measure $\q$.
We obtain the following pricing formula:
\be\label{eq:DerValueWithFXandCollateral}
V_t^f   = \frac{B_t}{X^{f}_t}\,\e_t\left[  \frac{\phi_T^f \,X_T^f}{B_T} + \int_t^T \frac{\diff D_u}{B_u} \right]
\ee
where
\beq\label{eq:DerivativeDividendsDef}
\diff D_u = X^{f}_u \diff D^f_u,\qquad \diff D^f_u := \diff \Pi_u^f -  X^{gf}_{u-} C^g_{u-}  (c_u^g - r_u^{gb}) \diff u
\eeq
with $D^f_0=D^f_{0-}=0$. In particular, the gain process of the financial derivative is $G_t:=V_t+D_t$ and its deflated version
\beq\label{eq:DeflatedGainOftheDerivative}
\widetilde{G}_t := \widetilde{V}_t + \int_0^t \frac{\diff D_u}{B_u}
\eeq
is a $\q$-martingale.
\end{proposition}
\begin{proof}
See Appendix \ref{App:MuticurrencyCollateralPricing}.
\end{proof}

\begin{remark}[Re-hypothecation]\label{rk:Rehypothecation}
In dividends  formulation (\ref{eq:diffDVMinusC}) we have an implicit hypothesis of collateral re-hypothecation: in fact all positive and negative collateral cash flows enter in the collateralized financial derivative gain process.
Instead, in case of no re-hypothecation, meaning that the collateral is segregated, the negative collateral cash flows should always be there but the positive cash flows should not: these last cash flows should only let 
the collateral credit grow, but this benefit will be on the disposal of the investor only at collateral closing date. See also Section \ref{sec:Futures} for an example of a segregated collateral account.
\end{remark}

\begin{remark}[Non-linearity]\label{rk:NonLinearity}
In pricing formula (\ref{eq:DerValueWithFXandCollateral}) the price of the financial derivative depends 
(via the dividends (\ref{eq:DerivativeDividendsDef})) on future realizations of the collateral value process $C^g$ which in turn will depend (in a potentially very complicated way) on future realizations of the price of the financial derivative itself: the pricing formula is in some sense recursive. This phenomenon is known in the literature as 
non-linearity of the pricing, where this terminology has a precise meaning arising from the theory of BSDEs (i.e.~the non-linearity of the generator of the BSDE describing the price process of the financial derivative): see \cite{bcr2018} and references therein (where the authors achieve broader and more ambitious results). In some cases, see Propositions \ref{prop:PriceWithSimpleCollateral}-\ref{prop:PriceWithSimpleCollateralDomesticMeas} and corresponding corollaries, we will add some simplifying hypotheses that will let us tackling the non-linearity issue.  
\end{remark}

\begin{remark}
As anticipated before, without FX market dislocations one has the usual drift of the FX rate $X^g$, i.e.~$\mu^g_t = r_t-r^g_t$ where $r^g$ is the spot risk-free interest rate of currency $g$.  Therefore, $r_u^{gb}:=r_u-\mu^g_u=r_u^{g}$ . In this case, the contribution of $r$ is canceled out in the last addend of formula (\ref{eq:DerivativeDividendsDef}). Moreover, if we remunerate the collateral at risk-free rate $c^g\setto r^g$ the last addend of the above formula disappears ending up with the usual pricing formula.
\end{remark}

\noindent We now derive the dynamics of the principal objects of our interest.  

\begin{proposition}\label{prop:DynamicsOfVMinusCandVunderQ}
The no-arbitrage dynamics of $(V-C)$ under $\q$ is
\[
\diff (V-C)_t = r_t (V-C)_{t-} \diff t + \diff \mathcal{M}^{\textit{V-C}}_t - \diff D^{\textit{V-C}}_t
\]
where $\mathcal{M}^{\textit{V-C}}$ is a RCLL $\q$-martingale with $\mathcal{M}^{\textit{V-C}}_0=0$. Then, under $\q$
\beq\begin{split}\label{eq:DynamicsOfDerivativePrice}
\diff V_t &= r_t V_{t-} \diff t+ \diff \mathcal{M}_t - \diff D_t\\
\diff \mathcal{M}_t &:= \diff \mathcal{M}^{\textit{V-C}}_t + C^g_{t-} \diff M^g_t
\end{split}\eeq
where $\mathcal{M}$ is a RCLL $\q$-martingale with $\mathcal{M}_0=0$.
\end{proposition}

\begin{proof}
The dynamics of $(V-C)$ is a direct consequence of the fact that $\widetilde{G}_t^{\textit{V-C}}$ is a $\q$-martingale, see Proposition \ref{prop:RiskNeutralDrift} from which one also obtains that
$\diff \widetilde{G}_t^{\textit{V-C}} = B_t^{-1} \diff \mathcal{M}^{\textit{V-C}}_t$. Moreover, by \eqref{eq:defGTilde},
\[
\diff \widetilde{G}_t^{\textit{V-C}} = \diff \widetilde{V}_t -\diff \widetilde{C}_t + \frac{\diff D_u^{\textit{V-C}} }{B_u}
\]
so, recalling \eqref{eq:DerivativeDividendsDef},
\bes
\diff \widetilde{V}_t &=& \diff \widetilde{G}_t^{\textit{V-C}}  - \left(\frac{\diff D_u^{\textit{V-C}} }{B_u} - \diff \widetilde{C}_t\right)\\
&=& \frac{1}{B_t}\left\{ \diff \mathcal{M}^{\textit{V-C}}_t -   X^f_{u} \diff \Pi_u^f  -  X_{u-}^g C^g_{u-} \left\{ (r_u-c^g_u)   \diff u- \frac{\diff X_u^g}{X_{u-}^g} \right\}  \right\} \\
&=& \frac{1}{B_t}\left\{ \diff \mathcal{M}^{\textit{V-C}}_t -   X^f_{u} \diff \Pi_u^f  +  X_{u-}^g C^g_{u-} \left\{ (c^g_u - r_u^{gb}) \diff u + \frac{\diff M_u^g}{X_{u-}^g} \right\}  \right\} \\
&=& \frac{1}{B_t}\left\{ \diff \mathcal{M}^{\textit{V-C}}_t + C^g_{u-}\diff M_u^g  -  \diff D_u   \right\} 
\ees
where we used \eqref{eq:diffAOverB} at the second line and Proposition \ref{prop:LebesgueIntegralDiscontinuities} in the last equation. Now,
\bes
\diff V_t &=& \diff ( B \widetilde{V})_t = B_t \diff \widetilde{V}_t  + \widetilde{V}_{t-}\diff B_t\\
&=&  \diff \mathcal{M}^{\textit{V-C}}_t + C^g_{u-}\diff M_u^g  -  \diff D_u + r_t  V_{t-}\diff t
\ees
which is the dynamics of $V$.
\end{proof}

\begin{remark}
One could deduce the dynamics of $C$ from $\diff V_u - \diff (V-C)_u$:
\bes
\diff C_u &=& r_u\, C_{u-}\diff u + (\diff \mathcal{M}_u-\diff \mathcal{M}^{\textit{V-C}}_u) - (\diff D_u-\diff D^{\textit{V-C}}_u)
\ees
however, this is not an explicit dynamics since in the term $\diff D^{\textit{V-C}}$ there is indeed a $\diff C_u$ term: in fact, substituting
\bes
\diff C_u &=& r_u \,C_{u-}\diff u + C^g_{u-} \diff M^g_u + X^g_{u-} \diff C_u^g  +   \diff [  C^g, X^g ]_u - r_u^{gb} \,X^g_{u-} C^g_{u-} \diff u \\
&=&  \mu_u^{g} \,C_{u-}\diff u + C^g_{u-} \diff M^g_u + \diff(X^g C^g)_u - C^g_{u-} \diff X_u^g
\ees
so, canceling $\diff C_u$ on both sides,
\bes
0 &=& \mu_u^{g} \,C_{u-}\diff u + C^g_{u-} \diff M^g_u  - C^g_{u-} \diff X_u^g\\
&=& \mu_u^{g} \,C_{u-}^g\, X^g_{u-} \diff u + C^g_{u-} (\diff X^g_u -\mu_u^{g}X^g_{u-}\diff u ) - C^g_{u-} \diff X_u^g
\ees
and also the right-hand side is equal to zero. 
\end{remark}


\begin{proposition}\label{prop:FXSDEInQf}
Recalling \eqref{eq:DynXg}, under $\q^{fb}$, we have
\be\label{eq:DynOfFXUnderQfb}
\diff X^x_t =  (r_t-r^{xb}_t) X^x_{t-}\diff t + \frac{\diff [M^x, M^f]_t}{X^f_t }+  \diff M^{x,fb}_t 
\ee
where $M^{x,fb}$ is $\q^{fb}$-martingale corresponding to $M^x$ after measure change to $\q^{fb}$. Moreover, under $\q^{fb}$,
\begin{align}
\diff X^{xf}_t &=   X^{xf}_{t-} \left\{ (r^{fb}_t-r^{xb}_t) \diff t  - X^{xf}_{t-} \,\diff M^{f,fb}_t \right\} \qquad &&x=d\label{eq:diffFXdf}\\
\diff X^{xf}_t &= X^{xf}_{t-} \left\{ (r^{fb}_t -r^{xb}_t)   \diff t  -\frac{\diff M^{f,fb}_t}{X^{f}_{t-}}  + \frac{\diff M^{x,fb}}{X^x_{t-}} \right\}\qquad &&x\neq d\label{eq:diffFXxf}
\end{align}
where $M^{f,fb}$ is a $\q^{fb}$-martingale corresponding to $M^f$ after measure change to $\q^{fb}$.
\end{proposition}
\begin{proof}
See Appendix \ref{App:proofOfFXSDEInQf}.
\end{proof}

\begin{corollary}\label{cor:betaxfIsaQfPriceProcess}
Recalling definition \eqref{eq:defBetaBankAccount}, we confirm that $\beta^{xf}:=B^{xb} X^{xf}$ is a \vvirg{$\q^{fb}$-price process}.
\end{corollary}
\begin{proof}
We have to prove that $\beta^{xf}/B^{fb}$ is a $\q^{fb}$-martingale. We have, using
\eqref{eq:ItoWithJumpsOnProduct} and that $B^{fb}, B^{xb}$ are continuous with finite variation:
\bes
\diff \left( \beta^{xf} (B^{fb})^{-1}\right)_t &=& (B^{fb})^{-1}_t \diff \beta^{xf}_t + \beta^{xf}_{t-} \diff ((B^{fb})^{-1})_t \\
&=& (B^{fb})^{-1}_t \left\{B^{xb}_t \diff X^{xf}_t + X^{xf}_{t-} \diff B^{xb}_t - r^{fb}_t\,\beta^{xf}_{t-}\diff t \right\} \\
&=& \frac{B^{xb}_t}{B^{fb}_t} \left\{ \diff X^{xf}_t + r^{xb}_t\,X^{xf}_{t-} \diff t - r^{fb}_t\,X^{xf}_{t-}\diff t \right\} \\
&=& \frac{B^{xb}_t}{B^{fb}_t} \left\{ \diff X^{xf}_t - (\mu^{x}_t-\mu^{f}_t) X^{xf}_{t-} \diff t \right\}
\ees
and the result is obtained substituting \eqref{eq:diffFXdf} or \eqref{eq:diffFXxf} and recalling that $\mu^d=0$.
\end{proof}

\begin{proposition}\label{prop:fxforward}
The FX forward (i.e.~the par rate of an FX forward contract) is written as
\beq\label{eq:fxforward}
X^{xf}_t(U):=\e_t^{U;fb}\left[X^{xf}_U\right] = X^{xf}_t \,\frac{P_t^{xb}(U)}{P_t^{fb}(U)}
\eeq
\end{proposition}
\begin{proof}
We have, recalling definition \eqref{eq:defBetaBankAccount} and that the passage $\q^{fb}\mapsto \q^{xb}$ is equivalent to the deflator passage $B^{fb}\mapsto \beta^{xf} :=  B^{xb} X^{xf}$,
\bes
X^{xf}_t(U)&:=&\e_t^{U;fb}\left[X^{xf}_U\right] = \frac{1}{P^{fb}_t(U)}\, \e_t^{U;fb}\left[ \frac{P^{fb}_t(U)}{P^{fb}_U(U)} \,X^{xf}_U \right]\\
&=&  \frac{1}{P^{fb}_t(U)} \,\e_t^{fb}\left[ \frac{B^{fb}_t }{B_U^{fb}}\, X^{xf}_U \right]\\
&=&  \frac{1}{P^{fb}_t(U)} \,\e_t^{xb}\left[ \frac{(B^{xb}\cdot X^{xf})_t}{(B^{xb}\cdot X^{xf})_U}\,X^{xf}_U\right]\\
&=&  \frac{X^{xf}_t}{P^{fb}_t(U)} \,\e_t^{xb}\left[ \frac{B^{xb}_t \,}{B^{xb}_U }\right]
\ees
from which we have the thesis.
\end{proof}

\begin{proposition}\label{prop:VfCollateralPricing}
Defining
\be
\widetilde{V}^f_t := \frac{V^{f}_t}{B^{fb}_t}, \qquad C^f_t :=  C_t^{g}\,X^{gf}_t,  \qquad \widetilde{C}^f_t := \frac{C^f_t}{B_t^{fb}}
\ee
and
\be\label{eq:DivsOfVfMinusCf}
\diff D^{\textit{f,V-C}}_u := \diff \Pi_u^f + X^{gf}_{u-} \left\{ \diff C_u^g - c_u^g \,C_{u-}^g \diff u +   \left[\diff  C^g_u,   \frac{\diff  X^{gf}_u }{X^{gf}_{u-} }\right]\right\},
\ee
with $D^{\textit{f,V-C}}_0=D^{\textit{f,V-C}}_{0-}=0$, then the gain process of the collateralized financial derivative $(V^f-C^f)_t$ in currency $f$ is $G^{\textit{f,V-C}}:=(V^{f}-C^{f})_t+D^{\textit{f,V-C}}_t$ and the deflated gain process
\be\label{eq:GTildef}
\widetilde{G}^{\textit{f,V-C}}_t := \widetilde{V}^f_t - \widetilde{C}^f_t + \int_0^t \frac{ \diff D^\textit{f,V-C}_u }{B_u^{fb}}
\ee
is a martingale under $\q^{fb}$. 
We also have the pricing formula:
\be\label{eq:DerValueWithFXandCollateralInBasisMeasure}
V_t^f  = B^{fb}_t  \,\e_t^{fb}\left[  \frac{\phi_T^f  }{B^{fb}_T } + \int_t^T \frac{ \diff D^f_u }{B_u^{fb}}  \right]
\ee
recalling definition (\ref{eq:DerivativeDividendsDef}). In particular, the gain process of the financial derivative is $G_t^f:=V_t^f+D_t^f$ and its deflated version
\[
\widetilde{G}_t^f := \widetilde{V}_t^f + \int_0^t \frac{ \diff D^f_u}{B_u^{fb}}
\]
is a $\q^{fb}$-martingale.
\end{proposition}
\begin{proof}
See Appendix \ref{App:VfCollateralPricing}.
\end{proof}

\noindent We now derive the Risk-Neutral drift of $V^f$, expanding the results of Section 2 in \cite{GPS2019}.

\begin{proposition}\label{prop:DynamicsOfVf}
The no-arbitrage dynamics of $(V^f-C^f)$ under $\q^{fb}$ is
\[
\diff (V^f-C^f)_t = r_t^{fb} (V^f-C^f)_{t-} \diff t + \diff \mathcal{M}^{f,\textit{V-C}}_t - \diff D^{f,\textit{V-C}}_t
\]
where $M^{f,\textit{V-C}}$ is a RCLL $\q^{fb}$-martingale with $\mathcal{M}^{f,\textit{V-C}}_0=0$. Then, under $\q^{fb}$,
recalling definition (\ref{eq:DerivativeDividendsDef}),
\be
\diff V^f_t &=& r^{fb}_t\,V_{t-}^f \diff t + \diff \mathcal{M}_t^f - \diff D^f_t \label{eq:diffVf}\\
\diff \mathcal{M}_t^f &:=& \diff \mathcal{M}^{\textit{f,V-C}}_t  +   X_{t-}^{gf}\,C^g_{t-}  \left\{\frac{\diff M^{g,fb}}{X^g_{t-}} -  \frac{\diff M^{f,fb}_t}{X^{f}_{t-}} \right\}\nonumber
\ee
where $\mathcal{M}^f$ is a RCLL $\q^{fb}$-martingale. Moreover, under $\q$, the quanto-$d$ dynamics of $V^f$ writes:
\beq\label{eq:QuantodDynOfVf}
\diff V^f_t = V_{t-}^f\left\{ r^{fb}_t \diff t- \frac{\diff[\mathcal{M}^f, X^f]_t}{X_{t-}^f\,V_{t-}^f} \right\}  + \diff \mathcal{M}_t^{\textit{fqd}}   - \diff D^f_t
\eeq
where $\mathcal{M}^{\textit{fqd}}$ is a $\q$-martingale corresponding to $\mathcal{M}^f$ after a change of measure to $\q$. 
\end{proposition}
\begin{proof}
See Appendix \ref{app:ProofOfDynamicsOfVf}.
\end{proof}

We extend here the results of Corollary 2.1 of \cite{MP2017}, in different directions: firstly, in this reference the authors do not consider over-collateralization, secondly they work in a setting where all market risks are
described by a vector of continuous Itô processes driven by a Brownian vector. Coherently with \eqref{eq:CollateralWithHaircut}, define the target value of collateral account 
as:
\[
C_{u}^{g\star}:=(1+\alpha_{u}) V_{u-}^f\,X^{fg}_{u-}
\]
where the haircut $\alpha_u\geq 0$ is a predictable process. As promised in Remark \ref{rk:NonLinearity}, the following Proposition tackles the issue of a recursive pricing formula.

\begin{proposition}\label{prop:PriceWithSimpleCollateral}
We add the hypotheses that the continuous dividend is proportional, i.e.~$\psi^f_u\setto\ell^f_u\,V^f_{u-}$ for some predictable process $\ell^f$ and also that we have continuous margin calls i.e.~$C^g_u\setto C_{u}^{g\star}$ for any $u\in [0,T]$. Then
\be\label{eq:PriceWithSimpleCollateral}
V_t^f  = B^{z^f}_t  \,\e_t^{fb}\left[  \frac{\phi_T^f  }{B^{z^f}_T }  +  \int_t^T  \frac{\diff \Phi_u^f}{B^{z^f}_u} \right]
\ee
where we define the blended interest rate
\be\label{eq:discRateWithSimpleCollateral}
z_u^f := (1+\alpha_u) ( c_u^g-r^{gb}_u + r^{fb}_u) - \big\{ \alpha_u \,r^{fb}_u  + \ell^f_u\big\}
\ee
\end{proposition}

\begin{proof}
Substituting $\psi^f_u\setto\ell^f_u\,V^f_{u-}$ and $C^g_u\setto C_{u}^{g\star}$, and defining
\be\label{eq:DeltaRateDefinition}
q_u :=\ell^f_u - (1+\alpha_u) (c_u^g-r^{gb}_u),
\ee
the financial derivative dividend process \eqref{eq:DerivativeDividendsDef} becomes
\[
\diff D^f_u = \diff \Phi_u^f + q_u V^f_{u-} \diff u 
\]
and using the result of Proposition \ref{prop:propDivsDiscount}, we obtain the thesis since
\bes
z_u^f &:=& r^{fb}_u - q_u\\
&=&(1+\alpha_u-\alpha_u)\,r^{fb}_u  +(1+\alpha_u) ( c_u^g-r^{gb}_u) -\ell^f_u\\
&=& -\alpha_u \,r^{fb}_u  +(1+\alpha_u) ( c_u^g-r^{gb}_u + r^{fb}_u) -\ell^f_u
\ees
Alternatively, as in Corollary 2.1 of \cite{MP2017}, we can use the Feynman-Kac Theorem in case all processes are continuous (so that all PDEs are well defined) and driven by a Brownian vector.
%
\end{proof}

\begin{remark}
Formula (\ref{eq:PriceWithSimpleCollateral}) is sometimes referred among practitioners to as \vvirg{CSA discounting}. This formula is in the most general form: 
all rates are stochastic and cash flows and collateral are under different foreign currencies. 
One could simplify the first hypothesis: see Corollary \ref{cor:multiccyDiscountingExplain} also in order to better understand the formula. Moreover, since $d, f$ and $g$ are only placeholders for currencies, one could also simplify the multi-currency framework: e.g.~setting $g\setto f$ (cash flows and collateral both in currency $f$) or $f\setto d$ (cash flows in domestic currency, collateral in currency $g$), or the other way round $g\setto d$ (cash flows in currency $f$ and collateral in domestic currency), etc. 
\end{remark}

\noindent We try to gain intuition on the crucial pricing formula \eqref{eq:PriceWithSimpleCollateral} and in particular on the first addend of $z^f$, the second being a payoff increase (in case both $r^{fb},\,\ell^f\geq 0$, remark the minus sign before the curly brackets) for over-collateralization and proportional dividends.

\begin{corollary}\label{cor:multiccyDiscountingExplain}
In case $\alpha=\ell^f=0$ and all rates are deterministic, (\ref{eq:PriceWithSimpleCollateral}) becomes
\bes
V_t^f  =    P^{z^f}_t(T)\,\e_t^{fb}\left[  \phi_T^f\right] +  \int_t^T  P^{z^f}_t(u)\, \e_t^{fb}\left[\diff \Phi_u^f \right]
\ees
where, for $U\geq t$
\[
P^{z^f}_t(U) = \frac{P^{c^g}_t(U)\,X^{fg}_t(U)}{X^{fg}_t}
\]
which has a clear financial meaning: the forward FX converts the time-$U$ cash flow from currency $f$ to currency $g$, then the cash flow is discounted to $t$ with rate $c^g$ and finally converted back in currency $f$ with $1/X^{fg}_t=X^{gf}_t$.
\end{corollary}
\begin{proof}
Straightforward from hypotheses, \eqref{eq:fxforward}-\eqref{eq:PriceWithSimpleCollateral}.
\end{proof}


\begin{proposition}\label{prop:PriceWithSimpleCollateralDomesticMeas}
In the same setting as in Proposition \ref{prop:PriceWithSimpleCollateral} we have
\be\label{eq:PriceWithSimpleCollateralDomesticMeas}
V_t^f  = \frac{B^{z}_t}{X^f_t}  \,\e_t\left[  \frac{\phi_T^f \,X_T^f  }{B^{z}_T }  +  \int_t^T  \frac{X_{u}^f\,\diff \Phi_u^f}{B^{z}_u} \right]
\ee
where
\[
z_u := (1+\alpha_u) ( c_u^g-r^{gb}_u + r_u) - \big\{ \alpha_u \,r_u  + \ell^f_u\big\}
\]
\end{proposition}
\begin{proof}
Recalling \eqref{eq:DeltaRateDefinition} and that $z^f := r^{fb} - q$, we have
\[
B^{z^f}_t := \esp^{\int_0^t (r^{fb}_u - q_u )\diff u} = \frac{B^{fb}_t}{B^{q}_t}
\]
So using \eqref{eq:SexpectedValueWithGenMeas}, and defining $z:=r-q$, the \eqref{eq:PriceWithSimpleCollateral} can be rewritten as
\bes
V_t^f  &=& \frac{B^{fb}_t}{B^{q}_t}  \,\e_t^{fb}\left[  \frac{\phi_T^f \,B^{q}_T }{B^{fb}_T }  +  \int_t^T  \frac{B^{q}_u\,\diff \Phi_u^f}{B^{fb}_u} \right]\\
&=& \frac{ \beta^{df}_t}{B^{q}_t}  \,\e_t \left[  \frac{\phi_T^f \,B^{q}_T }{\beta^{df}_T}  +  \int_t^T  \frac{B^{q}_u\,\diff \Phi_u^f}{\beta^{df}_{u-}} + B^{q}_u \diff \left[\Phi^f, \frac{1}{\beta^{df}}\right]_u\right]\\
&=& \frac{ B^{z}_t }{X^f_t}  \,\e_t \left[  \frac{\phi_T^f \,X^{f}_T }{B^{z}_T}  +  \int_t^T  \frac{X^{f}_{u-}\,\diff \Phi_u^f}{B^{z}_{u}} + \sum_{t<u\leq T } \frac{  \Delta \Phi^f_u \Delta X^f_u}{B^{z}_u}\right]\\
&=& \frac{ B^{z}_t }{X^f_t}  \,\e_t \left[  \frac{\phi_T^f \,X^{f}_T }{B^{z}_T} + \sum_{t<u\leq T }   \frac{X^{f}_{u-}\,\Delta \Phi_u^f +  \Delta \Phi^f_u \Delta X^f_u}{B^{z}_u}\right]
\ees
recalling that $\Phi^f$ is a pure jump process. Moreover,
\bes
z_u &:=& r_u - q_u\\
&=&(1+\alpha_u-\alpha_u)\,r_u  +(1+\alpha_u) ( c_u^g-r^{gb}_u) -\ell^f_u\\
&=& -\alpha_u \,r_u  +(1+\alpha_u) ( c_u^g-r^{gb}_u + r_u) -\ell^f_u 
\ees
which gives the thesis.
\end{proof}

\noindent Under the same hypotheses of Corollary \ref{cor:multiccyDiscountingExplain}, the payoff multiplier becomes 
\[
P^{z}_t(U) X^f_U = \frac{P^{c^g}_t(U)\,X^{dg}_t(U)}{X^{dg}_t} \,X^f_U
\]
i.e., counterclockwise from the right: convert the payoff from currency $f$ to the domestic currency with the corresponding stochastic future FX rate, then convert it to currency $g$ with the corresponding forward FX rate, discount to reference time $t$ with rate $c^g$ and then convert again the flow in the domestic currency with the corresponding FX spot.



\begin{corollary}[Domestic Collateral and Cash Flows]
In case of continuous collateralization and when both collateral and cash-flows are denoted in domestic currency (i.e.~$g=f=d$), we have
\[
V_t  = B^{z}_t \,\e_t\left[  \frac{\phi_T  }{B^{z}_T }  +  \int_t^T  \frac{\diff \Phi_u}{B^{z}_u} \right]
\]
where in this case $z = (1+\alpha)c -  \{\alpha r  + \ell\}$ and $c$ is the domestic collateral rate. In case $\ell=\alpha=0$ we have obtained the usual pricing formula with a new discounting rate passing from $r$ to the collateral discounting rate $c$: 
the price of the financial derivative no longer depends on the domestic spot risk-free rate.
\label{cor:PerfectCollSingleCcy}\end{corollary}
\begin{proof}
In this case the FX rate is equal to one for all $t$ and its drift is equal to zero for all $t$.
\end{proof}

\noindent The following proposition extends the results of \cite{GPS2019} to the case of stochastic rates. Let $\q^{U;fb}$ correspond to the $U$-forward basis measure 
of currency $f$ for domestic collateral rate $c=r$:
\[
\e_t^{fb}\left[\frac{\diff \q^{U;fb} }{\diff \q^{fb}\phantom{X}}\right] = \frac{P_t^{fb}(U)}{B^{fb}_t} \,\frac{B^{fb}_0}{P_0^{fb}(U)}
\]
i.e.~$Z^f_t(T;r)$ in the notation of \cite{MP2017} (see their equation (2.23)).

\begin{proposition}
For $t\leq U\leq T$, the forward of $V^f$ writes
\be\label{eq:forwardOfAssetsWithCollateral}
F^{f}_t(U):=\e^{U;fb}_t[V_U^f] = \frac{1}{P^{fb}_t(U)} \,\left\{ V_t^f  - B^{fb}_t\,\e_t^{fb}\left[  \int_t^U \frac{\diff D_u^f}{B_u^{fb}} \right]\right\}
\ee
where the dividends $D^f$ are defined in (\ref{eq:DerivativeDividendsDef}). The $d$-forward quanto is instead written:
\beq
F^{\textit{fqd}}_t(U):=\e^{U}_t[V_U^f] = \frac{1}{P_t(U)} \,\left\{ V_t^f  - B_t\,\e_t\left[  \int_t^U \frac{\diff D_u^{\textit{fqd}} }{B_u} \right]\right\}
\eeq
where
\[
\diff D_u^{\textit{fqd}}  := \diff D_u^f + V_{u-}^f\left\{ \mu^{f}_u \diff u + \frac{\diff[\mathcal{M}^f, X^f]_u}{X_{u-}^f\,V_{u-}^f} \right\}
\]
\end{proposition}
\begin{proof}
Recalling definition \eqref{eq:GTildef}, we have
\bes
F^f_t(U)&:=&\e^{U;fb}_t[V_U^f] \\ 
&=&\frac{1}{P^{fb}_t(U)}\, \e_t^{U;fb}\left[ \frac{P^{fb}_t(U)}{P^{fb}_U(U)} \,V_U^f  \right]\\
&=&  \frac{B^{fb}_t}{P^{fb}_t(U)} \,\e_t^{fb}\left[ \frac{V_U^f}{B_U^{fb}}\right]\\
&=&  \frac{B^{fb}_t}{P^{fb}_t(U)} \,\e_t^{fb}\left[ \widetilde{G}^f_U - \int_0^U \frac{\diff D_u^f}{B_u^{fb}} \right]\\
&=&  \frac{B^{fb}_t}{P^{fb}_t(U)} \,\left\{ \widetilde{G}^f_t   - \e_t^{fb}\left[ \int_0^U \frac{\diff D_u^f}{B_u^{fb}} \right]\right\}\\
&=&  \frac{1}{P^{fb}_t(U)} \,\left\{ V_t^f  - B^{fb}_t\,\e_t^{fb}\left[  \int_t^U \frac{\diff D_u^f}{B_u^{fb}} \right]\right\}\\
\ees
where we exploited the martingale property of $\widetilde{G}^f$. For the forward quanto: rewrite \eqref{eq:QuantodDynOfVf} with
\beq\label{eq:VfQuantoDynamicsAdjusted}
\diff V^f_t = r_t \,V_{t-}^f  \diff t + \diff \mathcal{M}_t^{\textit{fqd}}   - \diff D^{\textit{fqd}}_t
\eeq
so that
\[
\widetilde{G}^{\textit{fqd}}_t := \frac{V^f_t}{B_t} + \int_0^t\frac{\diff D_u^{\textit{fqd}}}{B_u} 
\]
is a $\q$-martingale. Then one can proceed analogously as the proof of $F^f_t(U)$-formula in the domestic currency case.
\end{proof}

\begin{remark}
Formula (\ref{eq:forwardOfAssetsWithCollateral}) has the usual appearance of the capitalization from $t$ to date $T$ thanks to $P^{fb}_t(U)$ of the asset value $V_t^f$ plus the expected discounted future dividend flow: it extends the concept of forward price to the new context with multi-currency collateralization and FX market dislocations. 
\end{remark}

\begin{corollary}
In case 
\[
\frac{\diff[\mathcal{M}^f, X^f]_u}{X_{u-}^f\,V_{u-}^f} = \rho_u \diff u
\]
and all interest rates are deterministic, the forward quanto writes
\beq
F^{\textit{fqd}}_t(U)=\frac{1}{P_t^\eta(U)} \,\left\{ V_t^f  - B_t^\eta\,\e_t\left[  \int_t^U \frac{\diff D_u^{f} }{B_u^\eta} \right]\right\}
\eeq
where $\eta:=r^{fb}-\rho$ and $P_t^\eta(U):=\e_t[B_t^\eta/B_U^\eta]=B_t^\eta/B_U^\eta.$
\end{corollary}
\begin{proof}
From dynamics \eqref{eq:VfQuantoDynamicsAdjusted} and Remark \ref{rk:PropDivsToChangeOfDiscountingRate}, defining $q:=\mu^f+\rho$,
\beqs
V^f_t = B_t \,\e_t\left[\frac{\phi_T^f}{B_T}+\int_t^T \frac{\diff D_u^f}{B_u} + \int_t^T \frac{q_u V^f_u \diff u}{B_u}\right] = B_t^{r-q} \,\e_t\left[\frac{\phi_T^f}{B_T^{r-q}}+\int_t^T \frac{\diff D_u^f}{B_u^{r-q}}\right]
\eeqs
and $r-q:=r-(\mu^f+\rho)=r^{fb}-\rho=:\eta$. Using this result, the tower property of conditional expectation and the fact that we have deterministic rates:
\bes
F^{\textit{fqd}}_t(U)&:=&\e^U_t[V^f_U]= \e_t[V^f_U] = \e_t\left[B_U^{\eta} \,\frac{\phi_T^f}{B_T^{\eta}}+B_U^{\eta}\int_U^T \frac{\diff D_u^f}{B_u^{\eta}}\right]\\
&=& \frac{B_U^{\eta}}{B_t^\eta} \,\e_t\left[ B_t^\eta\,\frac{\phi_T^f}{B_T^{\eta}}+B_t^\eta\int_U^T \frac{\diff D_u^f}{B_u^{\eta}}\right]\\
&=& \frac{B_U^{\eta}}{B_t^\eta} \,\left\{ V^f_t-B_t^\eta\,\e_t\left[ \int_t^U \frac{\diff D_u^f}{B_u^{\eta}}\right]\right\}\\
\ees
which gives the thesis.
\end{proof}

In the following proposition we add some simplifying hypotheses in order to have a sharper formula for the forward price. The result is the same as in \cite{GPS2019} in the new, more general, semimartingale framework.
\begin{proposition}\label{prop:ForwardWithDeterministicRates}
With the same hypotheses as in Proposition \ref{prop:PriceWithSimpleCollateral}, in case also all interest rates are deterministic, we have
\[
F^f_t(U) = \frac{1}{P_t^{z^f}(U)} \left[V^f_t -  \sum_i P_t^{z^f}(T_i)\,\e_t^{fb}\left[\phi_{T_i}^f \right]\,\1_{t<T_i\leq U}\right]
\]
where we recall definition of rate $z^f$ at (\ref{eq:discRateWithSimpleCollateral}) and we defined $P_t^{z^f}(u)=\e_t^{fb}[B_t^{z^f}/B_u^{z^f}]=B_t^{z^f}/B_u^{z^f}$.
\end{proposition}
\begin{proof}
From \eqref{eq:forwardOfAssetsWithCollateral}, recalling the definition of rate $q$ at \eqref{eq:DeltaRateDefinition},
\bes
F^f_t(U)\,P^{fb}_t(U) &=&  V_t^f  - B_t^{fb}\,\e_t^{fb}\left[  \int_t^U \frac{\diff D_u^f}{B_u^{fb}} \right]\\
&=&V_t^f  - B_t^{fb}\,\int_t^U \e_t^{fb}\left[\frac{\diff \Phi_u^f +	q_u\,V^f_u   \diff u }{B_u^{fb}} \right]\\
&=&V_t^f  - \,\int_t^U P_t^{fb}(u)\,\e_t^{u;fb}\left[  \diff \Phi_u^f + q_u\,V^f_u   \diff u  \right]\\
&=&V_t^f  - \,\int_t^U P_t^{fb}(u)\,\e_t^{u;fb}\left[ \sum_i \phi_{T_i}^f \,\delta(u-T_i) + q_u\,V^f_u \right]\diff u \\
&=&V_t^f  - \,\int_t^U P_t^{fb}(u) \sum_i \e_t^{u;fb}\left[ \phi_{T_i}^f \right]\,\delta(u-T_i) + \e_t^{u;fb}\left[q_u\,V^f_u \right]\diff u \\
\ees
where we recall that $\delta(u-x_0)$ is the Dirac mass centered at $x_0$. With deterministic rates $\q^{fb}=\q^{u;fb}$ for any $u$, we have
\[
\partial_U \Big(F^f_t(U)\,P^{fb}_t(U)\Big) = - P_t^{fb}(U) \left( \sum_i \e_t^{fb}\left[\phi_{T_i}^f \right]\,\delta(U-T_i)+ q_U\,F^f_t(U) \right)
\]
while, on the other hand, 
\[
\partial_U \Big(F^f_t(U)\,P^{fb}_t(U)\Big) =P^{fb}_t(U)\,\partial_U F^f_t(U) + F^f_t(U)\,\partial_U P^{fb}_t(U) = P^{fb}_t(U)\left(\partial_U F^f_t(U) -r_U^{fb}\right)
\]
so substituting we end up with the following ODE, for $ U \in [t, T]$
\[
\partial_U  F^f_t(U) = z^f_U\,F^f_t(U) + b(U)
\]
where 
\[
b(U):=- \sum_i \e_t^{fb}\left[\phi_{T_i}^f \right]\,\delta(U-T_i)
\]
with initial condition $F^f_t(t)=V^f_t$. The linear ODE can be solved with:
\bes
F^f_t(U) &=& \frac{1}{B_{U}^{z^f}} \left[ \frac{F^f_t(t)}{B_{t}^{z^f}} + \int_{t}^U  \frac{b(u)}{B_{u}^{z^f}} \diff u \right]\\
&=& \left(\frac{B_{t}^{z^f}}{B_{U}^{z^f}}\right)^{-1}  \left[ V^f_t + \,\int_{t}^U   b(u) \left(\frac{B_{t}^{z^f}}{B_{u}^{z^f}}\right) \diff u \right]
\ees
which can be rearranged to obtain the thesis.
\end{proof}

\section{Applications\label{sec:Applications}}



\subsection{Repurchase Agreement\label{sec:repo}}


We follow \cite{mc2010} for the deal description. The so called \emph{classic Repurchase Agreement} (classic repo for short) is a contract between two parties, let us say $R^S$ (the \vvirg{repo seller}: the seller of securities $S$) and $R^B$ (the \vvirg{repo buyer}: the buyer of securities $S$). In essence, a repo agreement is a secured loan (or collateralized loan): 
on the trade date the two counterparties sign an agreement whereby on a future date (the value or settlement date)
$R^S$ will sell to $R^B$ a nominal amount of securities in exchange for
cash. The price received for the securities is the market price of the stocks 
on the value date. The agreement also demands that on the termination date
$R^B$ will sell identical stock back to $R^S$ at a previously agreed price: the cash exchanged at value date plus interests calculated at an agreed
repo rate. 
%
Note that although legal title to the collateral passes to the repo
buyer, economic costs and benefits of the collateral (e.g.~dividends, coupons, capital gains, etc.) remain with the repo seller.

The transaction described above is a \emph{specific} classic repo, that is, one in which the collateral supplied is specified as a particular stock, as opposed to a \emph{general collateral} (GC) trade in which a basket
of collateral can be supplied, of any particular issue, as long as it is of the required type and credit quality.  
If the collateral asset is particularly on demand -- that is, it is \emph{special} -- the repo rate may be significantly below the GC rate. Special status
will push the repo rate downwards: zero rates and even negative rates are possible when dealing in specials. In fact, we distinguish in the repo market between the cash-driven players (players that demand/invest cash)
and security-driven players (players that are interested in securities loans e.g.~for selling them).



To counter all risks of repo transactions (credit risk, market risk of collateral, issuer risk of collateral, etc.), the repo loan is not only secured but also subject to a margination process. In particular, we denote with $t_i$ all margin call dates, where $t_0\equiv t$ and $t$ is the repo value date. Moreover, we denote with $T$ the repo termination date and with $\kappa_t^T$ the repo rate of the transaction, that fixes at time $t$ and is valid for the
repo contract of interval $[t,T]$. Moreover, we denote with $C_u$ for any $u\in[t,T]$ the value of the collateral (stock and cash) hold by the repo buyer.
We define also the collateral target value
\[
C_{u}^\star := (1+\alpha_t) \Big[H_t(1+\kappa_t^T\times \textit{yf}(t, u))\Big]
\]
where $\textit{yf}$ is the year fraction according to the termsheet accrual rule, $\alpha_t\geq 0$ is the haircut and $H_t:=N S_t$ represents the cash loaned by the seller for a reference number of securities $N$ and a repo collateral stock denoted by $S$. Hence the collateral target value not only encompasses the cash loan but also the repo interests already due at date $u$. We have the following items:
\begin{itemize}
\item $u=t$: the start cash proceeds of a repo can be less than the market value of the collateral by an agreed amount or percentage known as the haircut 
 so that at value date $t$ we have the following equalities:
\[
Q_t S_t = C_t \setto (1+\alpha_{t}) H_t =:C_{t}^\star
\]
meaning that the seller receives $H_t$ of cash and at the same time sells to the buyer an effective number $Q_t$ of securities $S$ with $Q_t = N(1+\alpha_t)$. In practice, the collateral value is a multiple (greater than one) of the cash: the deal is over-collateralized. We precise that the repo interests are calculated with respect to the cash loan, so the total repo interests writes 
$H_t\times \kappa_t^T \times \textit{yf}(t,T)$.

\item $u\in(t,T]$: The market value of the collateral is maintained through the life of repo contract. So if the market value of the collateral falls, the buyer calls for extra cash or collateral. If the market value of the collateral rises, the seller calls for extra cash or for obtaining back part of the collateral stocks. In order to reduce the administrative burden, margin calls can be limited to changes in the market value of the collateral in excess of an
agreed amount, or percentage, which is called a \emph{margin maintenance limit}. Imagine that at date $u$ the collateral value $C_{u}$ falls below the original cash loan $H_t$, then the buyer, under a margin arrangement, 
can call margin from the counterparty in the form of securities or cash, setting $u=t_i$ for some $i>0$. We have two cases:
\begin{romanlist}
\item The original repo contract is maintained. The margin adjustment at time $t_i$ is such that the collateral value reaches its target value, i.e.~$C_{t_i}\setto C_{t_i}^\star$. 
In a cash-driven trade, additional shares are delivered to restore
the haircut level. In this case, the terms of the original trade remain
unchanged, but the transfer of shares becomes a separate repo of stock against 
zero cash, and is unwound at the maturity date of the original trade. Where a repo is stock-driven, cash may be supplied as collateral. The market convention is for this cash to earn interest 
for the collateral supplier at the trade repo rate, or at a rate agreed between the two parties (that we will denote $c$). The interest
accrues at ACT/365 or on a basis agreed between parties.
\item The original repo is closed-out, and then re-opened under
new terms. The changed terms of the trade reflect the movement in collateral required to restore margin. See \cite{mc2010} p.~382 for further details on this case.
\end{romanlist}
\end{itemize}

We now formalize case i) above. For any margin call date $t_i$, we define with $F_{t_i}$ the residual quantity (with respects to collateral securities)  of cash collateral (equal to zero at inception date $F_t=0$): $F_u$ has values in $\re$ and represents the collateral cash held 
by the repo buyer at date $u$, (the sign convention is due to the buyer being the stock collateral taker), $F_u>0$ means that the repo buyer holds cash collateral, while $F_u<0$ means that the repo seller holds cash collateral.
For any $i\geq 0$ 
\beqs
C_{t_i} = Q_{t_i} S_{t_i} + F_{t_i} \setto C_{t_i}^\star
\eeqs
where either $Q$ or $F$ are set such that the last equality is respected. For any $i>0$ then the margin call received by the repo buyer is therefore
\begin{gather}
C_{t_{i}} - \Big\{Q_{t_{i-1}} S_{t_{i}}  + F_{t_{i-1}}\Big\} - c_{t_{i-1}}\, F_{t_{i-1}}(t_i-t_{i-1})\0\\
=  (Q_{t_{i}}-Q_{t_{i-1}}) S_{t_{i}} + F_{t_{i}} - F_{t_{i-1}}\Big(1+c_{t_{i-1}}\,(t_i-t_{i-1})\Big)\label{eq:repoBuyerMarginCall}   
\end{gather}
where the addend in curly bracket is the value of collateral pre-rebalancing (the natural evolution of $C_{t_{i-1}}$ in the interval $(t_{i-1},t_{i}]$: note that the cash remains flat since there is no other cash inflow/outflow, but the value of the stock evolves with its market movement), the last addend of the first equation ensures that the collateral cash be remunerated via interest rate $c$.

We now summarize all transactions of the repo contract.

\begin{itemize}
\item At time $t$:
\begin{itemize}
\item $R^S$ obtains $H_t=N S_t$ of cash from $R^B$.
\item $R^S$ repo sells to $R^B$ a number $Q_t=N(1+\alpha_t)$ of asset $S$ (of value $C_t = Q_t S_t=(1+\alpha_t)H_t$).
\item $R^S$ and $R^B$ agree on the repurchase price of
\[
H_t \Big(1+\kappa_t^T\times \textit{yf}(t,T)\Big)
\] 
which is fixed at $t$ but will be paid at $T$. 
\end{itemize}

\item  In the interval $(t,T]$:
\begin{itemize}
\item All the dividends/coupons paid by the collateral assets $S$ are collected from $R^B$ and transferred to $R^S$.
\item In case of margin maintenance breach, the value of the collateral account is restored as explained previously.
\end{itemize}

\item  At time $T$:
\begin{itemize}
\item $R^S$ pays to $R^B$ the repurchase price agreed at $t$.
\item $R^B$ closes the collateral account -- giving back to $R^S$ all the assets $S$ and the cash (if any) she detains as collateral (in case $R^B$ detains negative cash collateral, she receives it back from the seller). In case of no margin calls, 
$R^B$ gives back to $R^S$
the  quantity $Q_t$ of asset $S$ (of value $N(1+\alpha_t)\times S_{T}$).
\end{itemize}
\end{itemize}

Now we focus on a stock-driven repo for a fixed number of stocks $Q_t=N(1+\alpha_t)\setto 1$, i.e.
\[
N=\frac{1}{1+\alpha_t}
\]
and assume that all margin calls are performed via (domestic currency) cash: hence the non-cash part of the collateral is represented by a single stock $S$ for the entire contract period, i.e.~$Q_u=1$ for all $u\in[t,T]$. We have that, at inception,
the collateral is of value
\[
C_t = Q_t S_t = S_t
\]
which corresponds to the single stock $S$ repo-sold by $R^S$. 
For $u\in [t,T]$ we define the cash loan value process:
\[
X_u := S_t\left\{\frac{1}{1+\alpha_t}\right\} \Big(1+\kappa_t^T(u-t) \Big) = S_t\left\{\frac{1}{1+\alpha_t}\right\} \frac{B_u^\kappa}{B_t^\kappa}
\]
where the last equality is by a bootstrap of a deterministic term structure of repo rate $\kappa_u$: observe that since at the first order $1+\kappa_t^T(u-t) \approx \esp^{\kappa_t^T(u-t)}$ then $\kappa_u\approx \kappa_t^T$ for any $u\in[t,T]$. Observe also that $X_t=H_t$ represents the cash loan received by the repo seller at value date, and 
$X_T$ is the cash that the seller should pay at maturity to close the loan position. By Itô formula,
\beq\label{eq:dynamicsOfCashLoanValue}
\diff X_u = \frac{S_t}{B_t^{\kappa}}  \left\{\frac{1}{1+\alpha_t}\right\}\,\kappa_u \,B_u^{\kappa}\diff u = \kappa_u \,X_u \diff u
\eeq

Consistently with previous simplifying hypothesis, we assume that the collateralization eliminates the bilateral counterparty risk. Under this assumption, we write the P\&L discounted at date $t$ of the repo seller:
\bes
P\&L_t &=& S_t\left\{\frac{1}{1+\alpha_t}\right\} - C_t + B_t \left\{\int_t^T \frac{\diff \Phi_u - (\diff F_u - c_u F_{u-} \diff u)}{B_u} +  \frac{- S_t\left\{\frac{1}{1+\alpha_t}\right\} \Big(1+\kappa_t^T(T-t) \Big) +C_T}{B_T}\right\}\\
&=& - \left(S_t +F_t - X_t\right) + B_t \left\{\int_t^T \frac{\diff \Phi_u - (\diff F_u - c_u F_{u-} \diff u)}{B_u} +  \frac{S_T +F_T - X_T }{B_T}\right\}\\
\ees
where  $\diff \Phi_u$ is the dividend (e.g.~in case of $S$ being an equity stock) or coupon (e.g.~bond) of the asset $S$ paid at time $u$ (recall that the repo seller receives the dividend from the repo buyer); 
the margin calls are by \eqref{eq:repoBuyerMarginCall}, recalling that this equation is written from the repo buyer point of view, hence the negative sign. Recall that all margin calls are by domestic cash and that $F_t=0$ by construction.

In order to avoid arbitrage, we must have $\e_t[P\&L_t]=0$ meaning that the inception transaction is the equilibrium transaction. Defining $\boldsymbol{\mathcal{S}}:=S-X$ and $\boldsymbol{\mathcal{C}}:=-F$, this condition is equivalent to the fact that 
\[
\widetilde{G}^{\boldsymbol{\mathcal{S}}\textit{-}\boldsymbol{\mathcal{C}}}_u := \frac{\boldsymbol{\mathcal{S}}_u  -\boldsymbol{\mathcal{C}}_u }{B_u} +\int_0^u \frac{\diff D^{\boldsymbol{\mathcal{S}}\textit{-}\boldsymbol{\mathcal{C}}}_u}{B_u} 
\]
is a $\q$-martingale where
\[
\diff D^{\boldsymbol{\mathcal{S}}\textit{-}\boldsymbol{\mathcal{C}}}_u := \diff \Phi_u + \diff \boldsymbol{\mathcal{C}}_u - c_u \,\boldsymbol{\mathcal{C}}_{u-} \diff u.
\]
Therefore, we can exploit the results of Proposition \ref{prop:MuticurrencyCollateralPricing} for $f=g=d$, and in particular \eqref{eq:DerValueWithFXandCollateral},  to obtain that
\bes
\boldsymbol{\mathcal{S}}_t =  B_t \,\e_t\left[  \frac{\boldsymbol{\mathcal{S}}_T}{B_T} + \int_t^T \frac{\diff D_u}{B_u} \right]
\ees
where
\beqs
\diff D_u := \diff \Phi_u - \boldsymbol{\mathcal{C}}_{u-} (c_u - r_u) \diff u=\diff \Phi_u +F_{u-} (c_u - r_u) \diff u
\eeqs
with $D_0=D_{0-}=0$.  In case of continuous margin calls for any $u\in [t,T]$ we have $C_u = S_u+F_u \setto C_{u}^\star$, 
where
\[
C_{u}^\star = S_t(1+\kappa_t^T(u-t)) = (1+\alpha_t) X_u,
\]  
so
\[
F_u \setto C_{u}^\star - S_u = (1+\alpha_t) X_u - S_u  = - \boldsymbol{\mathcal{S}}_u + \alpha_t X_u
\]
and
\be
\diff D_u &=& \diff \Phi_u + (- \boldsymbol{\mathcal{S}}_{u-} + \alpha_t X_{u-} ) (c_u - r_u) \diff u = \diff \Phi_u +  \alpha_t X_{u-} (c_u - r_u) \diff u  +  \boldsymbol{\mathcal{S}}_{u-} (r_u - c_u ) \diff u\0\\
&=:& q_u\,\boldsymbol{\mathcal{S}}_{u-} \diff u + \diff \Pi_u \label{eq:DivsOfRepo}
\ee
where $q_u:=r_u-c_u$ and $\diff \Pi_u:=\diff \Phi_u +  \alpha_t X_{u-} (c_u - r_u) \diff u$. We can now apply Proposition \ref{prop:propDivsDiscount} with obvious substitutions 
$(S^i,q^i,\Phi^i)\mapsto (\boldsymbol{\mathcal{S}}, q,\Pi)$ and, since $\mu :=  r - q = c$,  we have:
\bes
S_t &=& X_t + B^{c}_t \,\e_t\left[  \frac{S_T -X_T }{B^{c}_T }  +  \int_t^T  \frac{\diff \Pi_u}{B^{c}_u} \right]\\
&=& (X_t-P_t^c(T)X_T) + B^{c}_t \,\e_t\left[  \frac{S_T }{B^{c}_T }  +  \int_t^T  \frac{\diff \Pi_u}{B^{c}_u} \right]\\
&=& S_t\left\{\frac{1}{1+\alpha_t}\right\} \left[1- \frac{P_t^c(T)}{P_t^\kappa(T)}\right] + B^{c}_t \,\e_t\left[  \frac{S_T }{B^{c}_T }  +  \int_t^T  \frac{\diff \Pi_u}{B^{c}_u} \right]
\ees
since $X_T$ is $\fst_t$-measurable. With some algebra we also obtain
\beq\label{es:stockValueUnderRepo}
 S_t = \frac{1+\alpha_t}{\frac{P_t^c(T)}{P_t^\kappa(T)}+\alpha_t}  \,B^{c}_t \,\e_t\left[  \frac{S_T }{B^{c}_T }  +  \int_t^T  \frac{\diff \Pi_u}{B^{c}_u} \right]
\eeq
If we also add the hypothesis of deterministic rates, we obtain a simplified forward price formula of asset $S$, for any $U\in(t,T]$:
\bes
F_t(U) &:=& \e_t^U[S_U] = \e_t[S_U] = X_U + \e_t[\boldsymbol{\mathcal{S}}_U] \\
&=& X_U + \frac{1}{P_t^{c}(U)} \left\{\boldsymbol{\mathcal{S}}_t -  \int_t^U P_t^{c}(u) \,\e_t \left[\diff \Pi_{u} \right]  \right\}\\
&=&  \frac{S_t}{1+\alpha_t} \left[\frac{1}{P_t^{\kappa}(U)}+\frac{\alpha_t}{P_t^{c}(U)}\right] - \frac{1}{P_t^{c}(U)} \left\{   \int_t^U P_t^{c}(u) \Big\{\e_t\left[\diff \Phi_{u} \right] +\alpha_t \,X_{u-}  
(c_u - r_u) \diff u \Big\}\right\}
\ees
where the second equality is since interest rates are deterministic, the third since $X_U$ is $\fst_t$-measurable, the second line is by \eqref{eq:fwdWithDeterminiticRates} with the same mapping $(S^i,q^i,\Phi^i)\mapsto (\boldsymbol{\mathcal{S}}, q,\Pi)$. Setting  $c_u\setto \kappa_u$ for any $u\in[t,T]$, 
\bes
F_t(U) &=& \frac{1}{P_t^{\kappa}(U)} \left\{S_t  -  \int_t^U P_t^{\kappa}(u) \,\Big\{\e_t\left[\diff \Phi_{u} \right] +\alpha_t \,X_{u-}  
(\kappa_u - r_u) \diff u\Big\}\right\}
\ees
and if $\alpha_t\setto 0$ we finally obtain:
\bes
F_t(U) &=& \frac{1}{P_t^{\kappa}(U)} \left\{S_t  -  \int_t^U P_t^{\kappa}(u) \,\e_t\left[\diff \Phi_{u} \right] \right\}
\ees
which is the formula of the dirty forward price of a bond (see e.g.~\cite[ch.~13]{TuckmanSerrat}) under repo rate continuous compounding, interpreting $S_t$ as the dirty price of the bond at time $t$ and $\diff \Phi_{u}$ as the (stochastic) bond coupon at time $u$. 
Under the same assumptions, using \eqref{eq:DynamicsOfDerivativePrice}-\eqref{eq:DivsOfRepo}-\eqref{eq:dynamicsOfCashLoanValue} and the continuity of process $X$, we have that for any $u\in[t,T]$:
\bes
\diff S_u &=& \diff \boldsymbol{\mathcal{S}}_u + \diff X_u =  r_u\, \boldsymbol{\mathcal{S}}_{u-} \diff u + \diff M_u - \diff D_u +  \diff X_u\\
&=& r_u\, \boldsymbol{\mathcal{S}}_{u-} \diff u + \diff M_u - \Big(\diff \Phi_u + \boldsymbol{\mathcal{S}}_{u-} (r_u-\kappa_u) \diff u \Big) +  \diff X_u\\ 
&=& \kappa_u\, S_{u-} \diff u + \diff M_u -  \diff \Phi_u   +  (\diff X_u-\kappa_u\, X_u\diff u)\\ 
&=& \kappa_u \,S_{u-} \diff u + \diff M_u - \diff \Phi_{u}
\ees 
for some RCLL $\q$-martingale driver $M$ with $M_0=0$. This dynamics is also coherent with the last forward price formula and with \eqref{es:stockValueUnderRepo} in case $c=\kappa$: as a result of the repo transaction, in all these cases the repo rate substitutes the spot risk-free rate (a similar result is obtained also in \cite{vp2010} with different arguments). Remark however that this dynamics is the result of 
all the simplifications above.


\subsection{Securities Lending}


We follow \cite{mc2010} for the transaction description. The securities lending contracts are generally set under the master agreement of the International Securities Lending Association (ISLA): see in particular \cite{isla2010}  
for further contract details.

Securities lending or stock lending is defined as a temporary transfer of securities in exchange for collateral. It is not a repo in the normal sense; there is no
sale or repurchase of the securities. The temporary use of the desired asset (the stock that is being borrowed) is reflected in a fixed fee payable by the
party temporarily taking the desired asset, usually accruing daily as a basis point charge on the market value of the stock
being lent, and generally payable in arrears on a monthly basis. The most common type of collateral is cash; however, it can happen that the collateral is represented by other securities to be given as collateral. In case of cash collateralization, the stock lender must pay the interest rates for the cash loan: these interest rates accrue and are paid with the same conventions as the lending fee, so that these payments can offset. 

Most stock loans are on an \vvirg{open} basis, meaning that they are confirmed (or terminated) each morning,
although term loans also occur. As in a classic repo transaction, coupon or dividend payments that become
payable on a security or bond during the term of the loan will be transferred from the stock borrower to the  the stock lender. 
Analogously, any coupon or dividend payments that become
payable on the collateral (in case of collateral assets) during the term of the loan will be transferred from the stock lender to the stock borrower.

At the maturity of the deal (or in case of default by one of the counterparties), the stock lender receives back the stock and returns back the collateral to the stock borrower. The counterparty risk is reduced not only by the presence of the collateral, but also due to the margination process.

We write below the P\&L discounted at date $t$ of the stock lender in cases of cash collateralization for a number of stocks $N=1$, in cases where both the asset and the collateral are quoted in domestic currency and that the
collateralization process eliminates the counterparty risk:
\bes
P\&L_t &=& -(S_t - C_t) + B_t \int_t^T \frac{\diff \Phi_u +\ell_u\, S_{u-}\diff u + \diff C_u - c_u\,C_{u-}\diff u}{B_u} + B_t \frac{(S_T- C_T)}{B_T}
\ees
where $\ell$ is the stock loan fee, $c$ is the interest rate of the cash collateralization, $\diff \Phi_u$ is the dividend/coupon of the asset $S$ paid at time $u$ (recall that the stock lender receives the dividend from the stock borrower), $C$ is the cash collateral value process (specifically, $C_u$ is the cash loan of the stock seller at time $u$) and $\diff C_u$ (with values in $\re$) represents the possible margin call of the collateral at time $u$. We have three terms in the above formula: the first represents all transactions at 
inception date $t$, the second represents all intermediate flows (discounted at $t$) in the interval $(t,T]$  and the third is the closing of all cash/borrowing positions at maturity $T$ (discounted at $t$). See also Remark \ref{rk:DiscreteCollateralization} for encompassing the case of discrete collateralization: we do not take into account the fact that collateral interests and lending fees accumulates in a period (e.g.~one month) and then are paid at the end of this period.

As usual, in order to avoid arbitrage, we must have $\e_t[P\&L_t]=0$, which is verified if and only if
\[
\widetilde{G}^{\textit{S-C}}_t := \frac{(S_t - C_t)}{B_t} +\int_0^t \frac{\diff D^{\textit{S-C}}_u}{B_u}
\]
is a Risk-Neutral martingale, where we defined
\[
\diff D^{\textit{S-C}}_u := \diff \Phi_u +\ell_u\, S_{u-}\diff u + \diff C_u - c_u\,C_{u-}\diff u
\]
and $D^{\textit{S-C}}_0=D^{\textit{S-C}}_{0-}=0$. We can directly apply the results of Proposition \ref{prop:MuticurrencyCollateralPricing} obtaining
\[
S_t  = B_t \,\e_t\left[  \frac{S_T  }{B_T }  +  \int_t^T  \frac{\diff D_u}{B_u} \right]
\]
where, recalling \eqref{eq:DerivativeDividendsDef},
\be\label{eq:diffDivsUnderStockLending}
\diff D_u := \diff \Phi_u + \ell_u\, S_{u-}\diff u - C_{u-}(c_u - r_u) \diff u
\ee
with $D_0=D_{0-}=0$ and also
\[
\widetilde{G}_t := \frac{S_t}{B_t} +\int_0^t \frac{\diff D_u}{B_u}
\]
is a Risk-Neutral martingale. Moreover, due to Proposition \ref{prop:DynamicsOfVMinusCandVunderQ}, the dynamics of $S$ under the Risk-Neutral measure must be
\[
\diff S_u = r_u S_{u-}\diff u + \diff M_u - \diff D_u
\]
where $M$ is a $\q$-RCLL martingale with $M_0=0$. 

Moreover, due to \eqref{eq:forwardOfAssetsWithCollateral}, for $t\leq U\leq T$, the forward price of $S$ is:
\bes
F_t(U):=\e^{U}_t[S_U] = \frac{1}{P_t(U)} \,\left\{ S_t  - B_t\,\e_t\left[  \int_t^U \frac{\diff D_u}{B_u} \right]\right\}.
\ees
In case of continuous margin calls $C_u\setto C^\star_u:=(1+\alpha_u)S_{u-}$ for any $u$, then we can also apply Proposition \ref{prop:PriceWithSimpleCollateralDomesticMeas} obtaining
\[
S_t = B^{z}_t \,\e_t\left[  \frac{S_T  }{B^{z}_T }  +  \int_t^T  \frac{\diff \Phi_u}{B^{z}_u} \right]
\]
where
\[
z_u =  -\alpha_u \,r_u  +(1+\alpha_u) c_u  - \ell_u
\]
With deterministic rates, we obtain from \eqref{eq:fwdWithDeterminiticRates} a simplified forward price formula
\[
F_t(U) = \frac{1}{P_t^{z}(U)} \left\{S_t -  \int_t^U P_t^{z}(u) \,\e_t\left[\diff \Phi_{u} \right]\right\}
\]

\begin{proposition}
In case of continuous proportional dividends $\diff \Phi_u = q_u S_{u-}\diff u$ and in case $r=c$ (or with no collateralization) one obtains:
\[
\diff S_u = (r_u-q_u-\ell_u) S_{u-}\diff u + \diff M_u
\]
Assuming also that $q,\ell$ are deterministic (while $r$ remains stochastic), we have 
\[
F_t(U)=\frac{S_t \,\esp^{-\int_t^U (q_u+\ell_u)\diff u}}{P_t(U)}
\]
\end{proposition}
\begin{proof}
The SDE is by direct substitution.  The forward price formula is derived as \eqref{eq:fwdWithDeterminiticRates} with dividend rate $\widetilde{q}:=q+\ell$ and $\Phi=0$.
\end{proof}
\begin{remark}
As a result, we obtained  the well known dynamics of the asset with proportional dividends and lending fee.
In market jargon, the rate $\ell$ is generally called the \vvirg{repo rate}, even if this terminology is quite misleading. 
\end{remark}



\subsection{Futures\label{sec:Futures}}

We follow \cite{HK} for the deal description. We work on a Futures contract on underlying $S$. We define once for all a set of times $0 \leq t_0 \leq t \leq T$, where
\begin{itemize}
\item $0$ is the inception date of the Futures contract;
\item $t_0$ is the date at which the investor enters into the Futures contract;
\item $t$ is today, a generic time between $t_0$ and $T$;
\item $T$ is the Futures contract expiry.
\end{itemize}
The Futures contract has the following properties:
\begin{itemize}
\item At every point in time $t\in[0,\,T]$, a Futures price process $f_{t}^T$ is quoted on the market;
\item At reset date $f_{T}^T  = S_T = F_{T}(T)$;
\item As in the case of forward contracts, since the Futures contract holder does not own the underlying $S$,  she does not receive any dividend eventually paid by the underlying.
\item The deal is centralized: i.e.~the counterparty of any Futures contract is the Clearing House (or central counterparty) and the deal is subject to margination in order to mitigate the investor's default risk for the House. In fact, the investor opens at inception $t_0$ an account at the Clearing House paying to her some cash $C_{t_0}$ and opening a credit $C$ position with her of the same amount. We distinguish between:
\begin{itemize}
\item \emph{Initial Margin account} $(\mathfrak{I})$: the movements of the investor to the credit $C$, in particular $C_{t_0}=\mathfrak{I}_{t_0}$ since at $t_0$ the investor pays $\mathfrak{I}_{t_0}$ and opens the account writing the credit $C_{t_0}$ in her book. Whenever the balance drops below the \emph{maintenance margin level}, the investor must bring the balance back up to the initial margin level again: hence at time $u$ she enhances her credit of $\diff \mathfrak{I}_u>0$
receiving $-\diff \mathfrak{I}_u<0$ (so paying) the same amount of cash. On the other hand, whenever the balance exceeds the initial margin level, the investor has the right to (and therefore should) withdraw any excess credit amount above the initial margin: hence at time $u$ the investor credit decreases of $\diff \mathfrak{I}_u<0$ and the investor receives an amount of cash of $-\diff \mathfrak{I}_u>0$.
\item A \emph{Variation Margin account} $(\mathfrak{M})$: the movements of the Clearing House to $C$, in particular $\mathfrak{M}_{t_0}=0$. 
Over the time the Futures price will change from the initial value $f_{t_0}^T$, sometimes going up $f_{u+ \delta u}^T-f_u^T>0$, sometimes going down $f_{u+ \delta u}^T-f_u^T<0$: at the end of every day the House checks the closing Futures price. If the price has gone up, the exchange credits to the investor this increase; if it has gone down it debits this decrease. 
\item Interest rate accrual: the Clearing must remunerate its debt at a rate of $c$ (possibly zero); 
\item Position closing: when the investor closes the contract at $u$, she  obtains back all remaining funds $C_u$ of her credit with respect to the House. Futures contacts are generally cash settled, but in the non-interest rates case there could be an underlying delivery: in this case the investor should hold the contract until expiry $u=T$, receiving back as usual $C_T$, but also paying $f_{T}^T$ of cash in order to obtain the underlying asset, whose value is $f_{T}^T=F_{T}(T)=S_T$. We will see in a moment that $C_T$  can be calculated from \eqref{eq:COfFutures}, but forgetting here the initial margin account and the interest rate accrual, so setting $C_T=\mathfrak{M}_T$, the investor's net cash position is $C_T-f_{T}^T=\mathfrak{M}_T-f_{T}^T=f_{T}^T-f_{t_0}^T-f_{T}^T=-f_{t_0}^T$ versus the stock obtained (similar to the long forward contract case where the holder pays $F_{t_0}^T$ in order to obtain the stock).
\end{itemize}		

\end{itemize}

To our knowledge, the best references that deals with Futures pricing are \cite{Bjork}-\cite{DD2001} who do not take into account the initial margin and represents the margin account flows as continuous cash flows (dividend) of the financial derivative: instead, we  adapt the framework of \cite{MP2017} to this contract case, obtaining different results (in particular the price of the Futures contract  being different from zero).

Denoting with $t_i$ the margin call times, the collateral account value can be written as:
\beq\label{eq:futuresMargination}\begin{cases}
C_{t_{i}}=C_{t_{i-1}}(1+c_{t_{i-1}}(t_{i}-t_{i-1})) + (\mathfrak{I}_{t_{i}}-\mathfrak{I}_{t_{i-1}}) + (f_{t_{i}}^T-f_{t_{i-1}}^T)  & i>0\\
C_{t_{0}}=\mathfrak{I}_{t_{0}}
\end{cases}\eeq
The first line describes the dynamics of the collateral account: it accrues with interest rate $c$ and grows with cash inflows/outflows of the investor ($\diff \mathfrak{I}_t$) and of the Clearing ($\diff \mathfrak{M}_t$).
Note that, setting $c=0$, we can prove recursively that we obtain
\beqs\begin{cases}
C_{t_{i}}=\mathfrak{I}_{t_{i}}+ (f_{t_{i}}^T-f_{t_{0}}^T)  & i>0\\
C_{t_{0}}=\mathfrak{I}_{t_{0}}
\end{cases}\eeqs
which says that the collateral account value at margin date is exactly equal to the initial margin value plus the difference between the current
and the inception quoted Futures price.

First, one can describe the dynamics of $C$ thinking that all processes $\mathfrak{I},c,C,\mathfrak{M}$ are RCLL pure jump predictable processes with predictable jumps at margin call times $t_i$: in particular 
\beq\label{eq:MarginAccountValue}
\mathfrak{M}_t = \sum_{t_0<u\leq t} \Delta \mathfrak{M}_u = f_{\gamma(t)-}^T -f_{t_0}^T 
\eeq
where $\Delta \mathfrak{M}_u=\mathfrak{M}_{t_0}=0$ except on jump times $t_i$ where $\Delta \mathfrak{M}_{t_i} = f_{t_i-}^T-f_{t_{i-1}}^T$ and we denoted with $\gamma(u)$ the greater margin call date $t_i$ such that $t_i\leq u$. 
The left limit is in order to guarantee that $\mathfrak{M}$ is predictable. Therefore, $\diff \mathfrak{M}_u$ is different from $\diff f_u^T$ which is a market quantity quoted in continuous time and not a simple pure jump process. In particular, 
we have from	 \eqref{eq:futuresMargination}:
\beqs\begin{split}
C_{t_0}&=\mathfrak{I}_{t_0}, \quad\mathfrak{M}_{t_0}=0\\
\diff C_u &=  c_u \,C_{u-} \diff u +\diff \mathfrak{I}_u+\diff \mathfrak{M}_u
\end{split}\eeqs
Integrating in the interval $[t_0,t]$,
\be
C_t &=& C_{t_0} + \int_{t_0}^t c_u C_{u-}\diff u + \mathfrak{I}_t-\mathfrak{I}_{t_0} + \mathfrak{M}_{t}-\mathfrak{M}_{t_0} \nonumber\\
&=&  \int_{t_0}^t c_u C_{u-}\diff u + \mathfrak{I}_{t}+ f_{\gamma(t)-}^T-f_{t_0}^T\label{eq:COfFutures}
\ee
Now, defining $\widetilde{C}_t:=C_t/B_t$ and using \eqref{eq:processDeflatedByBankAccount}, we have
\be
\diff \widetilde{C}_u &=& \frac{ \diff \mathfrak{I}_u + \diff \mathfrak{M}_u + C_{u-}(c_u-r_u) \diff u}{B_u} \label{eq:diffCtildeFutures}
\ee
and integrating in $[t_0,\tau]$
\[
\widetilde{C}_\tau = \widetilde{C}_{t_0} + \int_{t_0}^\tau \frac{ \diff \mathfrak{I}_u + \diff \mathfrak{M}_u + C_{u-}(c_u-r_u)\diff u}{B_u}
\]
The P\&L discounted at inception time $t_0$ of all transactions of the investor until target expiry $\tau$ is:
\bes
\e_{t_0}[P\&L_{t_0}] &=& -C_{t_0} + B_{t_0}\,\e_{t_0}\left[\frac{C_\tau}{B_\tau}+ \int_{t_0}^\tau \frac{-\diff \mathfrak{I}_u}{B_u}\right]\\
&=& B_{t_0}\,\e_{t_0}\left[\int_{t_0}^\tau \frac{\diff \mathfrak{M}_u + C_{u-}(c_u-r_u)\diff u}{B_u}\right]
\ees
where the first addend at the first line represents the opening of the credit at $t_0$, the second addend is the closing of the credit (the investor withdraws the cash in the account) and the third represents the \emph{initial margin account} (possible) calls: remember that $\diff \mathfrak{M}_u$ is not a cash flow disposable for the investor (it is operated by the Clearing House on a segregated account). In the second line we substituted the expression of $\widetilde{C}_\tau$.

For a generic time $t$ with $t_0\leq t\leq \tau$, the value of the contract is represented by the expected value of all discounted future transactions:
\be
V_t &=& B_t\,\e_t\left[\frac{C_\tau}{B_\tau}+ \int_t^\tau \frac{-\diff \mathfrak{I}_u}{B_u}\right]\0\\ 
&=& C_t+ B_t\,\e_t\left[ \int_t^\tau \frac{  \diff \mathfrak{M}_u+ C_{u-}(c_u-r_u)\diff u}{B_u}\right]\label{eq:futuresPricing}
\ee
where the reference to $t_0$ is hidden both in the value of $C_t$ (see \eqref{eq:COfFutures}) and in the value of $\mathfrak{M}_u$, recalling \eqref{eq:MarginAccountValue}.  

As in the previous section, we work under the hypothesis that the margination eliminates the counterparty risk for  the Clearing House and that the House is default free. Since we assume continuous margin call, we assume also that
$f^T$ is continuous so that $\mathfrak{M}$ stays predictable.

\begin{proposition}
In the above setting, we add the hypothesis of continuous margin calls so that $\diff \mathfrak{M}_u=\diff f^T_u$ and $\gamma(u)=u$. We assume that $(f_u^T)_{t\in [0,T]}$ is a continuous semimartingale with decomposition $f_u^T = f_{0}^T + A_u+M_u$ where $A$ is a continuous process with finite variation and $M$ is a continuous strict $\q$-martingale with $A_{0}=M_{0}=0$. Then absence of arbitrage implies:
\begin{romanlist}
\item Perfect collateralization, i.e.~$V_u=C_u$ for all $u\in[t_0,T]$;
\item Collateral rate equal to the risk-free interest rate, i.e.~$c\equiv r$;
\item The Futures price $f_u^T$ must be a $\q$-martingale (i.e.~$A_u=0$ for all $u\in [0,T]$).
\end{romanlist}
\end{proposition}
\begin{proof}
At inception, by construction $V_{t_0}=C_{t_0}=\mathfrak{I}_{t_0}$ which is the cash to enter in the Futures contact. At any time $\tau\in[t_0,T]$ the investor can exit the contract obtaining $C_\tau$, so $V_\tau=C_\tau$. Since $\tau$ is an arbitrary time by absence of arbitrage we have $V_t=C_t$ for any $t$. Hence, from this and from \eqref{eq:futuresPricing}
we must have that $\e_t[\widetilde{G}_\tau]=\widetilde{G}_t$, i.e.~that $\widetilde{G}$ is a $\q$-martingale, where we define
\[
\widetilde{G}_s:=\int_0^s \frac{  \diff f_u^T+ C_{u-}(c_u-r_u)\diff u}{B_u}
\]
In case $\widetilde{G}$ is a $\q$-martingale, note also that $\e_{t_0}[P\&L_{t_0}]=0$.
In order  for $\widetilde{G}$ to be a $\q$-martingale we must have  $\diff A_u = -C_{u-}(c_u-r_u)\diff u$. In this case
\bes
A_t &=& \int_0^t \diff A_u = -\int_0^t C_{u-}(c_u-r_u)\diff u.
\ees
Now, in the market we have many Futures contracts on the same underlying and expiry but with different inception dates $t_0^1, t_0^2, \ldots$ since all market participants can enter in the Futures contract at any time before the expiry. 
For all of them, the above condition should be valid. However recalling \eqref{eq:COfFutures}, for the same time $u$ the value of $C_u$ with inception time $t_0^i$ is different from the value of $C_u$ with inception time $t_0^j$ for $i\neq i$, while $f_u^T$ is the unique market quote among all inception dates. Hence we should set $c_u=r_u$ in order to make the second addend of $\widetilde{G}$ disappear. 
This implies that $A_t=0$ for all $t$, and so $f_t^T=f_{0}^T+M_t$ is a $\q$-martingale. 
\end{proof}

\begin{remark}
We have an example of perfect collateralization: if a counterparty wants to enter exactly in the same position as the investor, she must pay $C_t$. 
In particular, under all hypotheses/constraints of the above proposition we have
\bes
V_t &=& C_t+ B_t\,\e_t\left[ \int_t^\tau \frac{  \diff f_u^T}{B_u} \right] = C_t = \int_{t_0}^t r_u C_u\diff u + \mathfrak{I}_t + (f_{t}^T-f_{t_0}^T)
\ees
where the last equality is by (\ref{eq:COfFutures}). In particular: the first addend is the interest carry accrued from collateral, the second addend is the current level of the initial margin account and the third addend is the level of the margin account due to market moves. This result can also be compared with the price of an uncollateralized forward contract (with no counterparty risk)
struck at par at issue date $t_0$ that has value at time $t>t_0$:
$$V_{t} = P_t(T)\bigl[F_t(T)-F_{t_0}(T)\bigr]$$  
while the long uncollateralized Futures contract value is $V_t=f_{t}^T-f_{t_0}^T$.
\end{remark}

Finally, the result does not change if $\tau$ becomes a (stochastic) stopping time of an American option:
\bes
V_t &=& \sup_{\tau\in \mathcal{T}} B_t\,\e_t\left[\frac{C_\tau}{B_\tau}+ \int_t^\tau \frac{-\diff \mathfrak{I}_u}{B_u}\right]\\
&=& C_t + \sup_{\tau\in \mathcal{T}} B_t\,\e_t\left[\int_t^\tau \frac{  \diff f_u^T}{B_u}\right]\\
&=& C_t
\ees
where $\mathcal{T}$ is the set of all exercise strategies with values in $[t_0,T]$ and we used \eqref{eq:futuresPricing} for $r=c$: since $f_u^T$ is a $\q$-martingale, all exercise strategies must give the same price $V_t=C_t$.

Moreover, since $f_u^T$ is a $\q$-martingale,
\be\label{eq:FuturesPrice}
f_t^T=\e_t\left[f_{T}^T\right]=\e_t\left[S_T\right]= \e_t\left[F_{T}(T)\right]
\ee
and therefore the Futures price can be seen as a convexity adjusted Forward price: in fact the last expected value is taken under the \vvirg{unnatural} measure $\q$ instead of the natural $\q^{T}$ one (the measure under which $F$ is a martingale).

\section{Conclusion}

We refer to Section \ref{sec:Intro} for a list of the achievements of this paper which is mainly targeted for practitioners in the financial industry: we tried to make it the most linear, rigorous, self-contained and (hopefully) didactic as possible. 
Moreover, we tried to never take anything for granted and to make a (perhaps small) step further in understanding the topics of the paper. We hope that the efforts of this approach were visible for the reader especially for Section \ref{sec:GenTheory} (a brief but operationally comprehensive tour of the No-Arbitrage theory for dividend-paying assets) and for Section \ref{sec:Applications} (where we worked on applications studying from scratch the technicalities of the termsheets). A further development would be to extend the same approach to all other valuation adjustments.

\section*{Acknowledgments}

We are grateful to Andrea Pallavicini for his incredibly interesting and very didactic workshop \emph{\vvirg{A Pricing Framework for Valuation Adjustments}} held in the University of Verona (October 2019): it 
was our very first contact with the subject of this paper and with the general topic of valuation adjustments. We also thank Andrea for his comments and suggestions on the first draft of this paper. We thank Alessandro Gnoatto for organizing the above mentioned workshop and to Banca Akros for funding our participation in it. 
We are deeply grateful to Simone Pasquini for his precious reading, comments and suggestions, and for pointing out the references of Propositions \ref{prop:leftLimitVersionOfCadlagFct}-\ref{prop:LebesgueIntegralDiscontinuities}. We finally acknowledge
our colleague Delia Silvino for her constant support and encouragement.


\addcontentsline{toc}{section}{References}

\newpage
\appendix

\section*{Appendix}
\addcontentsline{toc}{section}{Appendix}

\subappendix{Semimartingales\label{sec:SemiMartingales}}

We briefly enumerate some results of semimartingales theory in order to let the paper be almost self-contained: only few propositions with a detailed proof and no indication of a source are, to the best of our knowledge, original.

A function $f$ is RCLL (\vvirg{\emph{right continuous with left limits}}) or càdlàg (from French acronym \vvirg{\emph{continue à droite, limitée à gauche}}) in a domain $D$ if it is right continuous and has finite left limits in $D$.
A function $f$ is LCRL (\vvirg{\emph{left continuous with right limits}}) or càglàd (\vvirg{\emph{continue à gauche, limitée à droite}}) in a domain $D$ if it is left continuous and has finite right limits in $D$.

For a RCLL function, we define also, for $h>0$, the \emph{jump} of $f$ at time $t$ as
\[
\Delta f(t):=f(t) - f(t-) \qquad \textit{where} \quad f(t-):=\lim_{h\to 0} f(t-h),
\]
Remark also that the jump operator is linear (straightforward by definition): for $a,b\in \re$ and RCLL functions $f,g$:
\beq\label{eq:linearityOfJumpOperator}
\Delta\Big(a\,f(t)+b\,g(t)\Big)=a\,\Delta f(t) + b \,\Delta g(t)
\eeq
and (again straightforward by definition)
\beq\label{eq:jumpOperatorProductRule}
\Delta\Big( f(t) g(t)\Big)= f(t-) \,\Delta g(t) + g(t-) \,\Delta f(t)+ \Delta f(t)\Delta g(t)
\eeq
A jump processes is a process with RCLL paths. We then enounce the following theorem, which is a cornerstone of jump process theory, stating that 
the number of big jumps of RCLL functions is finite and that their number of jumps is finite or countable.

\begin{theorem}[From \cite{da2009}]
If $f:D\mapsto \re$ with $D\subseteq \re$ is a RCLL function then
\begin{romanlist}
\item For each $k > 0$, the set $S_k = \{t\in D, \Delta f (t) > k\}$ is finite.
\item The set $S = \{t\in D, \Delta f (t) \neq 0\}$ is at most countable.
\end{romanlist}
\end{theorem}

\begin{proposition}[From  \href{https://math.stackexchange.com/questions/2373293/showing-that-the-left-hand-limit-function-is-left-continuous}{\texttt{math.stackexchange.com}}]\label{prop:leftLimitVersionOfCadlagFct}
Let $f:[a,b]\mapsto \re$. Define $g:(a,b]\mapsto \re$ with $g(t):=f(t-)$, then $g$ is left continuous for any $t\in (a,b]$.
\end{proposition} 
\begin{proof}
$g(t)=g(t-)$ can be proved using the sequential definition of left continuity, i.e.~we should prove that (from \cite{da2009}) for all sequences 
$(s_n ,n\in \n)$ in $(a,b)$ with each $s_n < t$ and $\lim_{n\to \infty} s_n = t$ we have that $\lim_{n\to \infty} g(s_n) = g(t)$.\\
Using the definition of $g(t):=f(t-)$, given $\epsilon>0$ there is $\delta>0$ such that
\[
|f(s)-g(t)|\leq \epsilon
\]
for all $t-\delta\leq s<t$. Since $s_n\to t-$, there is $n^\star$ such that $t-\delta\leq s_n<t$ for all $n\geq n^\star$. Fix $n\geq n^\star$. Then for $t-\delta< s <s_n$ we have
\[
|f(s)-g(t)|\leq \epsilon
\]
or equivalently,
\[
g(t) - \epsilon\leq f(s)\leq g(t) + \epsilon
\]
Letting $s\to s_n-$ in the previous inequality, we get
\[
g(t) - \epsilon\leq g(s_n)\leq g(t) + \epsilon
\]
so we have $|g(s_n)-g(t)|\leq \epsilon$ for all $n\geq n^\star$.
\end{proof}

\begin{remark}\label{rk:JumpsAtTMinus}
From the previous proposition, we have, for $h>0$:
\[
f(t--):=\lim_{h\to 0} g(t-h) := g(t-) = g(t) =: f(t-) 
\]
where used the left continuity of $g$. Then also:
\[
\Delta f(t-) := f(t-) - f(t--) = f(t-) - f(t-)=0.
\]
\end{remark}

\begin{proposition}\label{prop:LebesgueIntegralDiscontinuities}
Let $f$ be a RCLL function such that $f:D\mapsto \re$ with $D\subseteq\re$. Then, for $t,T \in D$
\[
\int_t^T f(u-) \diff u = \int_t^T f(u) \diff u
\]
\end{proposition}
\begin{proof}
See Remark 11.25 of \cite{WR76}: this is a more general result saying that the Lebesgue integral does not \vvirg{see} discontinuities of null Lebesgue-measure.  
\end{proof} 
	
%


\begin{definition}[From \cite{jyc2009}]
Let $\mathcal{F}$ be a given filtration.
\begin{romanlist}
\item The \emph{optional} $\sigma$-algebra $\mathcal{O}$ is the smallest $\sigma$-algebra on $\re^+ \times \Omega$ generated by RCLL $\mathcal{F}$-adapted processes (considered as mappings on $\re^+ \times \Omega$).
\item The predictable $\sigma$-algebra $\mathcal{P}$ is the smallest $\sigma$-algebra generated by the $\mathcal{F}$-adapted 
LCRL processes. The inclusion $\mathcal{P} \subset \mathcal{O}$ holds.
\end{romanlist}
A process is said to be predictable (resp. optional) if it is measurable with respect to the predictable (resp. optional) $\sigma$-field.
\label{def:PredictableProcesses}
\end{definition}

If a process is left continuous (specifically LCRL) then it is predictable but the opposite is not always true: as a very simple example of a predictable 
right continuous (specifically RCLL) process, think to process $X:=(X_t)_{t\geq 0}$ with $X_t = Z \1_{t\geq \tau}$ where $\tau\in [0,+\infty)$ and $Z$ is a random variable $\fst_{\tau-}$-measurable (the example is valid even in cases where $\tau$ is a predictable stopping time).

\begin{proposition}[From \cite{js2003}] 
If  $X$ is  a RCLL adapted process, then  $X_-:=(X_{u-})_{u\geq 0}$ is  a predictable
process;  moreover, if  $X$  is  predictable, then  $\Delta X_u$  is  predictable for any $u\geq 0$.
\label{prop:leftLimitOFXisPredictable}
\end{proposition}

\begin{definition}[From \cite{jyc2009}]
An $\fst$-semimartingale is a RCLL process $X$ which can
be written as $X_t = X_0+M_t + A_t$ where $M$ is an $\fst$-local martingale and where $A$ is
an $\fst$-adapted RCLL process with finite variation, and $M_0=A_0=0$.
\end{definition}


In general, the decomposition of a semimartingale is not unique; we shall speak about
decompositions of semimartingales. It is necessary to add some conditions
on the finite variation process to get the uniqueness.

\begin{definition}[From \cite{jyc2009}]
A \textbf{special semimartingale} is a semimartingale where $A$ (the finite variation part) is predictable. Such a decomposition $X = M + A$ with $A$
predictable, is unique. We call it the \textbf{canonical decomposition} of $X$, if it exists.
\end{definition}



\begin{definition}[From \cite{jyc2009}]\label{def:ContMartPartOfSemiMart}
Let $X$ be a semimartingale such that $\forall t \geq 0$, 
$\sum_{0<s\leq t} |\Delta X_s | < \infty$. Then process $\TX_t:=X_t -\sum_{0<s\leq t} \Delta X_s$ is a continuous semimartingale with unique decomposition $\TX= M + A$ where $M$ is a continuous local martingale and $A$ is a continuous process with bounded variation. The continuous martingale $M$ is called the \textbf{continuous local martingale part} of $X$ and it is denoted by $X^c$.
\end{definition}


\begin{definition}[From \cite{jyc2009}]
A process $H$ is \textbf{locally bounded} if there exists a sequence of stopping times $(\tau_n)_{n\geq 1}$ with $\tau_n>0$ increasing to $\infty$ a.s.~such that for each $n\geq 1,$ $(H_{t\wedge \tau_n})_{t\geq 0}$ is bounded. 
\end{definition}

As we read in \cite{jyc2009}, if $X$ is a semimartingale with decomposition $M+A$, then for any predictable locally bounded process $H$, we can define the process
\[
(H\star X)_t := \int_{(0,t]} H_u \diff X_u = \int_0^t H_u \diff M_u +\int_0^t H_u \diff A_u 
\]
where the second addend is a Stieltjes integral and the integral is not dependent on the decomposition of $X$ (see \cite{fp2004} for a more technical and detailed construction of this integral). 
\begin{proposition}[From \cite{fp2004,jyc2009}]
\label{prop:SemiMgIntegralProperties}
Let $H,K$ be locally bounded predictable processes and $X,Y$ semimartingales. Note that:
\begin{arabiclist}
\item The process $(H\star X)_t$ is a semimartingale, in particular it is a RCLL adapted process.
\item The process $(H\star X)_t$ does not depend on the decomposition of the semimartingale $X$.
\item Bilinearity:
\[
(H+K) \star X = H\star X + K\star X
\]  
and
\[
H \star (X+Y) = H\star X + H\star Y
\]
\item Associativity: 
\[
 H \star ( K  \star X)  =  ( H K) \star X
\]
which is, defining $Z_t:=\int_0^t K_u\diff X_u$, then
\[
\int_0^t H_u\diff Z_u = \int_0^t H_u K_u  \diff X_u
\]
\item The jump process of the stochastic integral of $H$ with respect to $X$ is equal to $H$ times the
jump process of $X$: 
\beq\label{eq:JumpsOfSemiMgIntegral}
\Delta(H\star X)_t = H_t \,\Delta X_t.
\eeq 
\item We have $H\star X^c = (H\star X)^c$. 
\item We also have that
\[
\int_{[0,t]} H_u \diff X_u = H_0 \Delta X_0 + \int_{(0,t]} H_u \diff X_u = H_0  X_0 + \int_{(0,t]} H_u \diff X_u
\]
since $X_{0-}=0$ by convention.
\item Let $\tau$ be a stopping time
\[
(H\star X)_{t\wedge \tau} = (H \1_{[0,\tau]} )\star X = H\star (X^\tau)
\]
where we denoted the stopped process $X^\tau_t:= X_{t\wedge \tau}$.
\item Let $M$ be a local martingale, then $H\star M$ is a local martingale (Theorem 29 of \cite{fp2004}).
\item Let  $M$  be a square integrable martingale, and  $Z$ be a bounded predictable process. Then  $Z\star M$ is a square integrable martingale (Theorem 11 of \cite{fp2004}).
\item Let  $M$  be a square integrable martingale, and  $Z$ be predictable such that $\int_0^t Z_u^2 \diff [M,M]_u<\infty$ a.s.~for each $t$. Then  $Z\star M$ is a square integrable martingale (Lemma p.~171 of \cite{fp2004}).
\end{arabiclist}
\end{proposition}


\begin{definition}[Quadratic Covariation, from \cite{fp2004}]
Given two semimartingales $X, Y$, the Quadratic Covariation process $[X, Y]$ is the semimartingale defined
by
\beq
[X, Y]_t := X_t Y_t - X_0 Y_0 - \int_0^t Y_{u-} \diff X_u - \int_0^t X_{u-} \diff Y_u.
\eeq
Since by definition $[X, Y]_0=0$, integrating both $[X, Y]_t$ and $(XY)_t$ allows one to have the differential form
\beq\label{eq:diffQuadraticCovDef}
\diff [X, Y]_t = \diff (XY)_t - Y_{t-} \diff X_u - X_{t-} \diff Y_t
\eeq
\noindent We will also denote the Quadratic Variation $[X,X]\equiv[X]$.
\end{definition}

\begin{corollary}[Product differentiation rule]
If $X, Y$ are semimartingales then
\be\label{eq:ItoWithJumpsOnProduct}
\diff(XY)_t = Y_{t_-} \diff X_t + X_{t_-} \diff Y_t   + \diff [X,Y]_t
\ee
\end{corollary}
\begin{proof}
Restatement of \eqref{eq:diffQuadraticCovDef}, it could also be proved from \eqref{eq:ItoWithJumpsMultiDim_diff}.
\end{proof}

\begin{proposition}[from \cite{fp2004,jyc2009}]
\label{prop:CovariatioProperties}
The Quadratic Covariation has the following properties:
\begin{arabiclist}
\item The map $(x,y)\mapsto f(x,y)$ with $f(x,y)\equiv [x,y]$ is symmetric (direct application of the definition).
\item The map $(x,y)\mapsto f(x,y)$ with $f(x,y)\equiv [x,y]$ is bilinear, i.e.~for semimartingales $X^i, Y^j$ and $a_i, b_j\in \re$
\[
\left[\sum_i a_i X^i, \sum_j b_j Y^j\right]_t = \sum_i\sum_j a_i \,b_j\,  [X^i,  Y^j ]_t
\]
(by direct application of (\ref{eq:QuadrCovAsLimitOfIncrements})).
\item From the properties above we have the two polarization identities: 
\[
[X, Y] = \frac{1}{4}\,\Big([X + Y, X + Y] -[X - Y, X - Y ]\Big)= \frac{1}{2}\,\Big([X + Y, X + Y] -[X, X ]-[Y, Y ]\Big)
\]
\item $[X, Y]$ is a nonanticipating RCLL process with paths of finite variation (this follows from the polarization identity, as $[X,Y]$ is the difference of two increasing functions).
\item Take a time grid $\pi^k = \{0=t_0^k  < t_1^k < \cdots < t_{n+1}^k = T\}$, the discrete approximation below converges in probability to $[X, Y]$ uniformly on $[0, T]$:
\be\label{eq:QuadrCovAsLimitOfIncrements}
\sum_{t_i^k\in \pi^k, t_i^k <t} (X_{t_{i+1}^k} - X_{t_{i}^k} )(Y_{t_{i+1}^k} -Y_{t_{i}^k} ) \stackrel{
\stackrel{|\pi^k|\rightarrow 0}{ \p}}{\longrightarrow} [X, Y ]_t
\ee
over all partitions $\pi^k$. Some reference presents this limit as the definition of Quadratic Covariation.
\item The jumps of the Quadratic Covariation process occur only at points
where both processes have jumps,
\beq\label{eq:JumpsOfQuadraticCov}
\Delta[X,Y]_t = \Delta X_t \Delta Y_t. 
\eeq
\item If one of the processes X or Y is of (locally) finite variation, then the sum $\sum_{0<u\leq t} |\Delta X_u ||\Delta Y_u |$ is almost surely finite and
\be\label{eq:QCovOfFVProcess}
[X,Y ]_t = \sum_{0<u\leq t} \Delta X_u \Delta Y_u.
\ee
\item Defining $\TX_t:=\int_0^t H_u \diff X_u$  and $\widetilde{Y}_t :=\int_0^t K_u \diff Y_u$
then
\[
[\TX, \TY]_t = \left[\int_0^\cdot H_u \diff X_u, \int_0^\cdot K_u \diff Y_u\right]_t =
\int^t_0 H_u K_u \diff [X, Y]_u.
\]
For this reason, the following formal calculation rules can be applied: 
\be\label{eq:QCovariationOfIntegral}
\diff [\TX, \widetilde{Y}]_t =  [\diff\TX_t, \diff\widetilde{Y}_t] = [H_t \diff X_t, K_t \diff Y_t] = H_t\, K_t\diff  [ X,  Y]_t
\ee
\end{arabiclist}
\end{proposition}

%

\begin{proposition}[From \cite{jyc2009}]
\label{prop:IntegralWithPureJumpIntegrand}
Let $A$ be a finite variation process and $X$ a semimartingale:
\beqs
\int_0^t \Delta X_u \diff A_u = \sum_{0<u\leq t}\Delta X_u\Delta A_u = [X,A]_t
\eeqs
Then
\[
\diff (A_t X_t) = A_{t-}\diff X_t +X_t \diff A_t
\]
\end{proposition}


\begin{proposition}
Define with $\TX_t :=\sum_i \int_0^t H^i \diff X^i_u$ and $\TY_t :=\sum_j \int_0^t K^j_u \diff Y^j_u$ where $X^i,Y^j$ are semimartingales and $H^i,K^j$ are locally bounded predictable processes.
Then
\[
[\TX, \TY]_t = \sum_{ij}\int_0^t H^i_u K^j_u \diff [X^i,Y^j]_u
\]
or
\[
\diff [\TX, \TY]_u = \left[\sum_i H^i_u \diff X^i_u, \sum_j  K^j_u \diff Y^j_u\right] = \sum_{ij} H^i_u K^j_u \diff [X^i,Y^j]_u
\]
\label{prop:LinearityOfQuadraticCovariation}
\end{proposition}
\begin{proof}
We have:
\bes
[\TX, \TY]_t &:=& \left[\sum_i \int_0^\cdot H^i_u \diff X^i_u, \sum_j \int_0^\cdot K^j_u \diff Y^j_u\right]_t\\
&=&\sum_{i,j} \left[ \int_0^\cdot H^i_u \diff X^i_u, \int_0^\cdot K^j_u \diff Y^j_u\right]_t\\
&=&\sum_{i,j}  \int_0^t H^i_u\, K^j_u \diff [X^i, Y^j]_u
\ees
where in the second line we used the bilinearity of Quadratic Covariation \ref{prop:CovariatioProperties} and in the third line \eqref{eq:QCovariationOfIntegral}.
\end{proof}

\begin{proposition}[Covariation of local martingales, from \cite{fp2004}]
If $X$ and $Y$ are two locally square integrable local martingales. Then
$[X,Y]$ is the unique adapted RCLL process $A$ with paths of finite variation on compacts satisfying:
\begin{romanlist}
\item $XY -A$ is a local martingale;
\item $\Delta A = \Delta X\Delta Y$
\item $A_0=X_0 Y_0$. 
\end{romanlist}
\end{proposition}

%
%

\begin{definition}[Compensator, from \cite{jyc2009}]
An adapted increasing process $A$ is said to be a compensator for the semimartingale $X$ if $X - A$ is a local martingale.
\end{definition}

\noindent  For example if $X$ is a local martingale, the process $[X,X]$ is a compensator for $X^2$. 
In general, a semimartingale admits many compensators. If there exists a predictable
compensator, then it is unique (among predictable compensators). 

We now introduce the Predictable Quadratic Covariation, in some references also called Sharp Brackets/Angle Brackets/Conditional Quadratic Covariation.

\begin{definition}[Predictable Quadratic Covariation, from \cite{fp2004}]\label{def:predictableCovariation}
Let $X$ be a semimartingale such that its Quadratic Variation process $[X,X]$ is of locally integrable variation. Then the Predictable Quadratic Variation of $X$, denoted $\langle X, X\rangle$ exists and it is defined to be the compensator of $[X,X]$. The Predictable Quadratic Covariation $\langle X, Y\rangle$ can be defined as the compensator of $[X,Y]$ provided of course that $[X,Y]$ is of locally integrable variation.
\end{definition}

\begin{proposition}[From \cite{fk2005}]
If $X$ is a continuous semimartingale with integrable
Quadratic Variation, then $\langle X,X\rangle= [X, X]$, and there is no difference
between the sharp and the square bracket processes.
\end{proposition}

\noindent The Predictable Variation is inconvenient since, unlike the Quadratic Variation,
it doesn't always exist. Moreover, while  $[X, X]$, $[X,  Y]$, and $[Y, Y]$ all remain
invariant with a change to an equivalent probability measure,  the sharp brackets, in general, change with a change to an equivalent probability measure and may even no longer exist. 


\begin{proposition}[Covariation process decomposition, from \cite{fp2004}]\label{prop:CovariationProcessDecomposition}
If $X$ and $Y$ are semimartingales, their Quadratic Covariation is
\[
[X, Y ]_t = [ X, Y]^c_t + \sum_{0<u\leq t} (\Delta X_u)(\Delta Y_u) 
\]
or
\be\label{eq:CovariationFormula1}
\diff [X, Y ]_t = \diff [ X, Y]^c_t +  \Delta X_u \Delta Y_u.
\ee
where $[ X, Y]^c$ denotes the continuous local martingale part of $[ X, Y]$: see Definition \ref{def:ContMartPartOfSemiMart}.
\end{proposition}

\begin{proposition}[From \cite{fp2004}] 
For semimartingales $X,Y$, we have the following equalities:
\be\label{eq:CovariationFormula2}
[X,Y]^c = [X^c,Y^c]=\langle X^c,Y^c\rangle
\ee
If both $X$and $Y$ are continuous 
\[
[X,Y]^c = [X,Y]=\langle X,Y\rangle
\]
\end{proposition}

\begin{lemma}[Itô Formula, from \cite{jyc2009}]
Let $X = (X_1 ,\ldots,X_n )$ be a semimartingale vector process and $f: \re^+\times \re^n \mapsto \re$ with $f \in \mathcal{C}^{1,2}$. Then,
\begin{gather}
f(t, X_t) = f(0,X_0) +\int_0^t \partial_u  f(u, X_{u_-}) \diff u + \int_0^t \sum_{i=1}^n  \partial_{x_i} f(u, X_{u_-}) \diff X_u^i \label{eq:ItoWithJumpsMultiDim}\\
+ \frac{1}{2} \int_0^t \sum_{i,j=1}^n  \partial_{x_i x_j}^2 f(u, X_{u_-} ) \diff [X^{i}, X^{j}]^c_u 
+ \sum_{0< u \leq t} \left\{f(u,X_u) - f(u,X_{u_-}) - \sum_{i=1}^n  \partial_{x_i} f(u, X_{u_-})\Delta X_u^i\right\}\0
\end{gather}
or in differential form
\be
\diff f(u, X_u) &=& \partial_u  f(u, X_{u_-}) \diff u +\sum_{i=1}^n \partial_{x_i} f(u, X_{u_-}) \diff X_u^i+ \frac{1}{2} \sum_{i,j=1}^n  \partial_{x_i x_j}^2 f(u, X_{u_-} )  \diff [X^{i}, X^{j}]^c_u \0\\
&\phantom{=}&+  \left\{f(u,X_u) - f(u,X_{u_-}) - \sum_{i=1}^n \partial_{x_i} f(u, X_{u_-})\Delta X_u^i\right\}\label{eq:ItoWithJumpsMultiDim_diff}
\ee
\end{lemma}

\begin{proposition}
Define $Y_t := f(t,X_t)$ where $f: \re^+\times \re^n \mapsto \re$ with $f \in \mathcal{C}^{1,2}$ and let $X=(X_1 ,\ldots,X_n )$ be a semimartingale vector process. Then, for another scalar semimartingale process $Z$
\be\label{eq:QuadraticCovarianceOfFctOfX}
\diff [Y,Z]_t = \sum_{i=1}^n \partial_{x_i}   f(t, X_{t_-}) \diff [X^i, Z]_t^c + \Big\{f(u, X_{u}) - f(u, X_{u-})\Big\}\Delta Z_u 
\ee
which is the continuous Quadratic Variation result plus the co-jump term.
\end{proposition}
\begin{proof}
Denoting with $f^-:=f(u, X_{u-})$ and using the Itô formula
\beqs\begin{split}
[Y, Z]_t =& \Bigg[Z, \int_0^\cdot \sum_{i=1}^n \partial_{x_i} f^- \diff X_u^i + \frac{1}{2}\int_0^\cdot \sum_{i,j=1}^n  \partial_{x_i x_j}^2 f^- \diff [X^{i}, X^{j}]^c_u 
+ \sum_{0< u \leq\, \cdot} \Big\{f - f^- - \sum_{i=1}^n \partial_{x_i}   f^- \Delta X_u^i\Big\}\Bigg]_t\\
=& \int_0^t \sum_{i=1}^n \partial_{x_i} f^- \diff [X^i, Z]_u   + \frac{1}{2}\int_0^t \sum_{i,j=1}^n \partial_{x_i x_j}^2 f^- \diff [[X^i, X^j]^c, Z]_u \\
&+ \left[\sum_{0< u \leq \cdot} \Big\{f - f^- - \sum_{i=1}^n \partial_{x_i}   f^- \Delta X_u^i\Big\}, Z\right]_t\\
=& \int_0^t \sum_{i=1}^n \partial_{x_i} f^- \diff [X^i, Z]_u   + \left[\sum_{0< u \leq \cdot} \Big\{f - f^- - \sum_{i=1}^n \partial_{x_i}   f^- \Delta X_u^i\Big\},Z\right]_t
\end{split}\eeqs
where in the first equality we used that the time integral is continuous with finite variation, in the second equality we used the Proposition \ref{prop:LinearityOfQuadraticCovariation}, while in the third equality the fact that $[X^i, X^j]^c$ is continuous with finite variation. Define with
\[
\Delta A_u := f(u, X_{u}) - f(u, X_{u-}) - \sum_{i=1}^n \partial_{x_i}   f(u, X_{u-})  \Delta X_u^i
\]
and, under Stieltjes integration 
\[
A_t:= \sum_{0<u\leq t} \Delta A_u = \int_0^t \diff A_u
\]
is a (finite variation) pure jump RCLL process with jumps $\Delta A_u$ equal to zero except on $X$ jump times (recall that $f\in \mathcal{C}^{1,2})$. 
Therefore, the last addend of the above derivation is 
\bes
\left[A, Z\right]_t &=&  \left[\int_0^\cdot \diff A_u, \int_0^\cdot \diff Z_u\right]_t\\
&=& \int_0^t \diff [A,Z]_u\\
&=& \sum_{0< u \leq t} \Delta A_u \,\Delta Z_u
\ees
where the last line is since $A$ is of finite variation. The result follows from \eqref{eq:CovariationFormula1}.
\end{proof}

We then enounce the Girsanov Theorem generalized to semimartingales, following \cite{fp2004} p.~131 (see also \cite{jyc2009} p.~534). Let  $X$  be a semimartingale on a space $(\Omega, \fst, \p)$ satisfying the usual hypotheses.
Let $\q\sim\p$, then there exist a $\p$-integrable random variable $\frac{\diff \q}{\diff \p}$ such that $\e^{\p}[\frac{\diff \q}{\diff \p}]=1$. We let
\[
L_t := \e^{\p}_t\left[\frac{\diff \q}{\diff \p}\right]
\] 
be the right continuous version. Then  $(L_t)_{t\geq 0}$  is a uniformly integrable martingale, hence a semimartingale.  Note that since $\q$ is equivalent to $\p$ then $\frac{\diff \p}{\diff \q}$ is $\q$-integrable and 
$\frac{\diff \p}{\diff \q}=\left(\frac{\diff \q}{\diff \p}\right)^{-1}$. 

\begin{lemma}
Let $\q\sim\p$, and $L$ defined as above.  An adapted RCLL process $M$ is a $\q$  local-martingale if and only if  $M L$  is a  $\p$  local martingale.
\end{lemma}

\begin{theorem}[Generalized Girsanov Theorem]\label{th:GeneralizedGirsanov}
Let  $\q\sim\p$. Let  $X$  be a semimartingale under $\p$ with decomposition  $X  =  M^{\p}  + A^{\p}$. Then  $X$  is also a semimartingale under $\q$ and has a decomposition
$X =  M^{\q}  + A^{\q}$ with $M^{\p}_0= A^{\p}_0=M^{\q}_0= A^{\q}_0=0$ and where $M^{\q}$ is a $\q$-local martingale and $A^{\q}$ is a $\q$-finite variation process.
\begin{romanlist}
\item General result, optional version:
\[
\diff M^{\q}_t = \diff M^{\p}_t - \frac{\diff [M^{\p},L]_t}{L_t}
\]
and  
\[
\diff A^{\q}_t = \diff A^{\p}_t + \frac{\diff [M^{\p},L]_t}{L_t}
\]
\item Predictable version: If $[X,L]$ is $\p$-locally integrable (which implies that $\langle X,L\rangle$ exists) than
\[
\diff M^{\q}_t = \diff M^{\p}_t - \frac{\diff \langle M^{\p},L\rangle_t}{L_{t-}}
\]
and  
\[
\diff A^{\q}_t = \diff A^{\p}_t + \frac{\diff \langle M^{\p},L\rangle_t}{L_{t-}}
\]
\end{romanlist}
\end{theorem}
%

\begin{corollary}\label{cor:GirsanovForFVProc}
A semimartingale $X$ of finite variation with null local martingale part, i.e.~a semimartingale with decomposition $X=0+A$, has the same dynamics under $\p$ and under $\q$.
\end{corollary}

\subappendix{Basic Dividend Models}\label{app:BasicDividendModels}

\noindent As a follow up of Sections \ref{sec:ForwardPrice}-\ref{sec:nonnegasset} we analyze two basic deterministic dividend models: the continuous proportional case and the cash dividend case.
We calculate the forward price and check that the deflated gain process is effectively a local martingale in both cases. We then see how Proposition \ref{prop:bu18} applies to both cases.\\

\subsubappendix{Continuous Proportional Dividends\label{sec:contDivs}}

\noindent We set
\[
\diff D_t \setto q_t \,S_{t-} \diff t
\]
with $q$ deterministic, this means that the dividend in the $\diff t$ interval is proportional to the asset value at time $t$ and to the length of the interval itself. This is the easy case:
the Risk-Neutral dynamics, recalling \eqref{eq:RiskNeutralDynWithDivs}-\eqref{eq:RiskNeutraldriftWithDivs}, becomes
\bes
\diff S_t &=& (r_t-q_t) S_{t-} \diff t + \diff M_u
\ees
Moreover, using \eqref{eq:fwdWithDeterminiticRates}, the forward price becomes
\beqs
F_t(T) =  \frac{S_t  \,\esp^{-\int_t^T q_s\diff s}}{P_t(T) }
\eeqs
As an exercise, we can check that the deflated gain process \eqref{eq:GTilde} is indeed a martingale with the chosen dividend model. In this setting, the process
$Z_t : = \esp^{\int_0^t q_u \diff u} S_t$ has dynamics
\[
\diff Z_t = r_t Z_{t-} \diff t +\diff M_u
\]
and therefore $\frac{Z}{B}$ is a $\q$-martingale (being driftless and by item 10 of Proposition \ref{prop:SemiMgIntegralProperties}). We can use this property, and the fact that $q$ is deterministic, to prove that $\widetilde{G}_t=\e_t[\widetilde{G}_T]$
for this specific case.
We have
\bes
V_t^D(T) &:=& B_t \,\e_t\left[ \int_t^T \frac{ \diff D_u}{B_u}\right] = B_t \,\e_t\left[ \int_t^T \frac{q_u S_{u-}\diff u}{B_u}\right] \\ 
&=& B_t \,\int_t^T q_u\, \esp^{-\int_0^u q_v\diff v}\, \e_t\left[\frac{Z_{u-}}{B_u}\right] \diff u \\ 
&=& Z_t \int_t^T q_u \,\esp^{-\int_0^u q_v\diff v} \,  \diff u\\ 
&=& Z_t \left[ -\esp^{-\int_0^u q_v\diff v} \right]_{u=t}^{u=T} \\
 &=& S_t\left( 1 - \esp^{-\int_t^T q_u\diff u}\right) 
\ees
and since
\bes
\e_t \left[\frac{S_T}{B_T}\right]  &=& \e_t\left[\frac{ \esp^{-\int_0^T q_u\diff u}\, Z_T}{B_T}  \right] \\ 
&=&\esp^{-\int_0^T q_u\diff u} \,\e_t\left[\frac{ Z_T}{B_T} \right]  \\ 
&=&\esp^{-\int_0^T q_u\diff u} \,\frac{ Z_t}{B_t}\\ 
&=&  \esp^{-\int_t^T q_u\diff u}\,\frac{ S_t}{B_t}
\ees
we have that
\bes
\e_t \left[\widetilde{G}_T\right]  &=&\e_t\left[\frac{S_T}{B_T}\right] + \frac{V_t^D(T)}{B_t} +\int_0^t \frac{\diff D_u}{B_u}\\
&=& \esp^{-\int_t^T q_u\diff u}\,\frac{ S_t}{B_t} + \frac{S_t}{B_t} \left( 1 - \esp^{-\int_t^T q_v\diff v}\right) +\int_0^t \frac{\diff D_u}{B_u}\\ 
 &=& \frac{ S_t}{B_t} +\int_0^t \frac{\diff D_u}{B_u} =: \widetilde{G}_t
\ees

Finally, in case one is willing to build a non-negative asset in light of Proposition \ref{prop:bu18}, one obtains, under deterministic interest rates,
\bes
S_0^\star &=& S_0\, \esp^{-\int_0^{\widehat{T}} q_u\diff u} \\
S_t &=& S_0^\star \,B_t\,\esp^{ \int_t^{\widehat{T}} q_u\diff u}\, M^{S}_t = S_0 \,B_t\,\esp^{ -\int_0^{t} q_u\diff u}\, M^{S}_t = F_0(t)\, M^{S}_t 
\ees
Hence,
\bes
\diff S_t = r_t \,S_{t-} \diff t +F_0(t) \diff M^{S}_t - \diff D_t\\
\ees

\subsubappendix{Cash Dividends\label{Sec:AbsoluteDivs}}

\noindent We set
\[
D_t \setto \sum_{i=1}^{+\infty} \phi_{\tau_i} \1_{0<\tau_i\leq t}
\]
where $\phi_{\tau_i}$ is the dividend paid at time $\tau_i$ for an increasing sequence of stopping times $\tau_1<\tau_2<\ldots$. We have
\[
\diff D_u = \sum_{i=1}^{+\infty} \phi_{u}  \diff \Theta_{\tau_i}(u)
\]
where we recall that $\Theta_{T}(u)$ is the Heaviside function centered at $T$: this is a Stieltjes integral (see e.g.~\cite{ap2011}). 
In order to compact the notation we define the counting process $N_t := \sum_{i=1}^{+\infty} \1_{\tau_i\leq t}$ and rewrite
\[
D_t = \sum_{i=1}^{N_t} \phi_{\tau_i}\qquad \diff D_u = \phi_u\diff N_u
\]
Notice that $D_t = \sum_{0<u\leq t} \Delta D_u$ is a RCLL pure jump process with jumps 
\[
\Delta D_u = \sum_{i=1}^{+\infty} \phi_{u}  \,\Delta \Theta_{\tau_i}(u) = \phi_{u}  \,\Delta N_u
\] 
In this case  $\Delta S_t =(\ldots) - \Delta D_t$ from which it is clear that $S_t$ drops with dividends $\phi_{t}$ at times $t$ when $t=\tau_i$ for some $i$. 

Preferring to stay simple, we decide to set as deterministic both jump times $\tau_i$'s and jump size $\phi$, hence $N$ is a deterministic process: the generalization to a Compound Poisson Process should be straightforward. We have
\bes
V_t^D(T) &:=& B_t\,\e_t\left[  \int_t^T \frac{\diff D_u}{B_u}\right] =  B_t\,\e_t\left[  \sum_{t< u \leq T} \frac{\Delta D_u}{B_u}\right] \\
&=&  B_t\,\e_t\left[  \sum_{t< u \leq T} \frac{\phi_{u}  \,\Delta N_u}{B_u}\right] \\
&=& B_t\,\e_t\left[  \sum_{i=N_t+1}^{N_T}\frac{\phi_{\tau_i}}{B_{\tau_i}}\right]\\
&=&   \sum_{i=N_t+1}^{N_T} \phi_{\tau_i} \,P_t(\tau_i)\\
\ees
hence the forward price writes
\be\label{eq:DetermAbsDivsFwd}
F_t(T)=\frac{S_t - \sum_{i=N_t+1}^{N_T} \phi_{\tau_i} P_t(\tau_i)}{P_t(T)}
\ee
which is the correct forward with stochastic interest rates: here the replication strategy of the general case can be simplified. Instead of  directly selling the future dividend flow, the trader can sell at time $t$ a number $N_T-N_t$ of zero coupon bonds with the same expiries as the dividends ($\tau_i$) and with each notional equal to the future (known) cash dividend ($\phi_{\tau_i}$): when the dividend is paid at $\tau_i$ it is directly given to the $\tau_i$- Zero coupon holder to close the contract at its maturity $\tau_i$.

Finally, in case one is willing to build a non-negative asset in light of Proposition \ref{prop:bu18}, one can set
\bes
S_t &=& S_0^\star \,B_t\, M^{S}_t + \sum_{i=N_t + 1}^{N_{\widehat{T}}} \phi_{\tau_i} P_t(\tau_i)\\
S_0^\star &=& S_0 - \sum_{i=1}^{N_{\widehat{T}}} \phi_{\tau_i} P_0(\tau_i)\\
M^{D}_t &=& \frac{\sum_{i=N_t + 1}^{N_{\widehat{T}}} \phi_{\tau_i} P_t(\tau_i)}{B_t} + \sum_{i=1}^{N_{t}} \frac{\phi_{\tau_i}}{B_{\tau_i}}
\ees
In case of deterministic interest rates:
\bes
S_t &=& \left[ S_0 - \sum_{i=1}^{N_{\widehat{T}}} \frac{\phi_{\tau_i}}{B_{\tau_i}} \right] B_t\, M^{S}_t + \sum_{i=N_t + 1}^{N_{\widehat{T}}} \phi_{\tau_i} \frac{B_{t}}{B_{\tau_i}} = \left[ F_0(t) - V_t^D(\widehat{T}) \right] M^{S}_t + V_t^D(\widehat{T})
\ees
and $M^D_t = \sum_{i=1}^{N_{\widehat{T}}} \frac{\phi_{\tau_i}}{B_{\tau_i}}$ does not depend on $t$, hence $\diff M^D_t=0$. Therefore,
\[
\diff S_t = r_t S_{t-} \diff t + B_t \,S_0^\star\diff M^{S}_t - \diff D_t = r_t S_{t-} \diff t + [S_{t-} - V_{t-}^D(\widehat{T})]\,\frac{\diff M^{S}_t}{M^{S}_{t-}} - \diff D_t
\]
which corresponds to equation (11) of \cite{bu10}.\\

\subappendix{Technical Proofs}

\subsubappendix{Proof of Proposition \ref{prop:MuticurrencyCollateralPricing}\label{App:MuticurrencyCollateralPricing}}


\noindent From the previous section and \eqref{eq:GainProcessGeneralMeasure}:
\beqs\begin{split}
\widetilde{G}_t^{\textit{V-C}}  :=& \frac{ V_t^f }{\beta^{df}_t} -\frac{ C_t^g}{\beta^{dg}_t} + 
\int_0^t \frac{\diff \Pi_u^f}{\beta^{df}_{u-}} + \diff \left[ \Pi^f, \frac{1}{\beta^{df}} \right ]_u  + \int_0^t \frac{\diff C_u^g - c_u^g\, C_{u-}^g \diff u}{\beta^{dg}_{u-}} +\diff \left[ C^g,  \frac{1}{\beta^{dg}} \right ]_u \\
=& \frac{X^{f}_t V_t^f }{B_t} -\frac{X^{g}_t C_t^g}{B_t} + \int_0^t \frac{X^f_{u-} }{B_u}\diff \Pi_u^f + \sum_{0<u\leq t} \frac{\Delta X^f_u \Delta \Pi^f_u }{B_u} + \int_0^t \frac{X^g_{u-}}{B_u}\Big(\diff C_u^g - c_u^g \,C_{u-}^g \diff u \Big)\\
&+  \diff \left[ C^g,  \frac{1}{\beta^{dg}} \right ]_u
\end{split}\eeqs
where we used \eqref{eq:QCovOfFVProcess}. Now, 
\bes
Y&:=&\int_0^t \frac{X^f_{u-} }{B_u}\diff \Pi_u^f + \sum_{0<u\leq t} \frac{\Delta X^f_u \Delta \Pi^f_u }{B_u} \\
&=& \int_0^t \frac{X^f_{u-} }{B_u}(\diff \Phi_u^f +\psi^f_u\diff u) + \sum_{0<u\leq t} \frac{\Delta X^f_u \Delta \Phi^f_u }{B_u} \\
&=& \int_0^t \frac{X^f_{u} }{B_u} \,\psi^f_u\,\diff u + \sum_{0<u\leq t} \frac{\Delta \Phi^f_u  ( X^f_{u-} + \Delta X^f_u) }{B_u} \\
&=& \int_0^t \frac{X^f_{u} }{B_u} \,\psi^f_u\,\diff u + \sum_{0<u\leq t} \frac{X^f_{u}}{B_u} \,\Delta \Phi^f_u    \\
&=& \int_0^t \frac{X^f_{u} }{B_u} \diff \Pi_u^f
\ees
where the third equality by Proposition \ref{prop:LebesgueIntegralDiscontinuities}. 
Moreover, using Proposition \ref{prop:LinearityOfQuadraticCovariation}, and that $B$ is continuous with finite variation (see \eqref{eq:QCovOfFVProcess}):
\bes
\diff \left[ C^g,  \frac{1}{\beta^{dg}} \right ]_u &=& \diff \left[ C^g,   X^g B^{-1} \right ]_u \\
&=& \left[ \diff  C^g_u,   B^{-1}_u \diff X^g_u + X^g_{u-} \diff (B^{-1})_u \right ]\\
&=& B^{-1}_u \left[ \diff  C^g_u,    \diff X^g_u \right ]\\
&=& \frac{X^g_{u-} }{B_u} \left[\diff  C^g,   \frac{\diff  X^g }{X^g_{u-} }\right]\\
\ees
and, collecting the last three equations results, we end up with the definition of $\widetilde{G}^{\textit{V-C}}$ and $D^{\textit{V-C}}$.\\
Now, in order to derive $V^f$, we will use the fact that $\widetilde{G}^{\textit{V-C}}$ is a $\q$-martingale, in the same spirit as the derivation of \eqref{eq:SexpectedValueWithGenMeas}. Before, from a straightforward application of \eqref{eq:ItoWithJumpsOnProduct},
\bes
\diff \widetilde{C}_u &=& (X^g C^g)_{u-} \diff (B^{-1})_u + B^{-1}_u \diff (X^g C^g)_u\\
&=& B^{-1}_u \left\{ -r_u X^g_{u-} C^g_{u-} \diff u +  X^g_{u-}\diff C^g_u + C^g_{u-}\diff X^g_u + \diff [X^g, C^g]_u\right\}\\
&=& \frac{ X_{u-}^g}{B_u}\left\{\diff C_u^g - r_u  C^g_{u-} \diff u + C^g_{u-} \,\frac{\diff X_u^g}{X_{u-}^g} + \left[ \diff C^g_u,   \frac{\diff  X^g_u }{X^g_{u-} } \right ] \right\}
\ees
hence, from this result and from \eqref{eq:diffDVMinusC}
\be\label{eq:diffAOverB}
\frac{\diff D_u^{\textit{V-C}}}{B_u}  = \frac{X^f_{u} }{B_u}\diff \Pi_u^f +  \diff \widetilde{C}_u+\frac{ X_{u-}^g C^g_{u-}}{B_u}\left\{ (r_u-c^g_u)   \diff u- \frac{\diff X_u^g}{X_{u-}^g} \right\}
\ee
which can be substituted in \eqref{eq:defGTilde}:
\beqs
\widetilde{G}_t^{\textit{V-C}} = \widetilde{V}_t -\widetilde{C}_0  +  \int_0^t \frac{X^f_{u} \diff \Pi_u^f}{B_u} + \int_0^t \frac{X^g_{u-} C^g_{u-}}{B_u}\left\{ (r_u - c_u^g)  \diff u- \frac{\diff X_u^g}{X_{u-}^g}  \right\}.
\eeqs
Using the above expression twice in the equality $\e_t[\widetilde{G}_T^{\textit{V-C}}]=\widetilde{G}_t^{\textit{V-C}}$, one obtains
\bes
\widetilde{V}_t -\widetilde{C}_0 = \e_t\left[ \widetilde{V}_T -\widetilde{C}_0  +  \int_t^T \frac{X^f_{u} }{B_u}\diff \Pi_u^f  +\int_t^T \frac{X^g_{u-} C^g_{u-}}{B_u}\left\{  (r_u - c_u^g)  \diff u- \frac{\diff X_u^g}{X_{u-}^g} \right\} \right]
\ees
from which we have the thesis (using Proposition \ref{prop:LebesgueIntegralDiscontinuities}), since from \eqref{eq:DynXg},
\bes
\e_t\left[ \int_t^T \frac{C^g_{u-}}{B_u}   \diff X_u^g   \right] = \e_t\left[ \int_t^T \frac{C^g_{u-}}{B_u}  \,\mu_u^g  X_{u-}^g \diff u\right] + \e_t\left[ \int_t^T \frac{ C^g_{u-}}{B_u}  \diff M_u^g \right]
\ees
and the second addend is equal to zero by Property 10 or 11 of Proposition \ref{prop:SemiMgIntegralProperties} in case the integrand is well behaved: the proof of the pricing formula \eqref{eq:DerValueWithFXandCollateral} 
is concluded thanks to Proposition \ref{prop:LebesgueIntegralDiscontinuities}. The last statement is a direct application of the definition of $\widetilde{G}$ to the pricing formula.\\

\subsubappendix{Proof of Proposition \ref{prop:FXSDEInQf}\label{App:proofOfFXSDEInQf}}

\noindent We change the measure from $\q^{fb}$ to $\q$, recalling \eqref{eq:defBetaBankAccount}:
\[
L_t := \e_t\left[\frac{\diff \q^{fb}}{\diff\q\phantom{f}}\right] = \frac{ \beta^{fd}_t }{B_t} \,\frac{B_0}{\beta^{fd}_0 }=\frac{ B^{fb}_t X^f_t (B^{-1})_t}{ X^f_0}
\]
We have, under $\q$, since all bank accounts are continuous with finite variation,  
\bes
\diff L_t &=&    \frac{1}{X^f_0} \left\{ \beta^{fd}_{t-} \diff ((B)^{-1})_t + ((B)^{-1})_t \diff \beta^{fd}_t \right\}\\
&=&    \frac{1}{X^f_0 B_t} \left\{ - r_t\, \beta^{fd}_{t-}  \diff t + B^{fb}_t \diff  X^f_t + X^f_{t-} \diff B^{fb}_t\right\}\\
&=&L_{t-} \left\{(r^{fb}_t- r_t )\diff t +  \frac{\diff X^f_t}{X^f_{t-}} \right\}\\
&=&L_{t-} \left\{ \frac{\diff M^f_t}{X^f_{t-}} \right\}=L_{t} \left\{ \frac{\diff M^f_t}{X^f_{t}} \right\}
\ees
where the last equality is since the FX rate is the only jump component of $L$: we obtained that $L$ is a $\q$-martingale as expected. 
We now apply  the Girsanov Theorem \ref{th:GeneralizedGirsanov} (and use Proposition \ref{prop:LinearityOfQuadraticCovariation}): 
\bes
\diff M^{x,fb}_t &=& \diff M^x_t - \frac{\diff [M^x, L]_t}{L_t}\\
&=& \diff M^x_t - \frac{\diff [M^x, M^f]_t}{X^f_t}
\ees   
is a local martingale under $\q^{fb}$. In particular, from \eqref{eq:JumpsOfSemiMgIntegral} and \eqref{eq:JumpsOfQuadraticCov},
\beq\label{eq:JumpsOfMxfb}
\Delta M^{x,fb} = \Delta M^x_t - \frac{\Delta M^x \Delta M^f}{X^f_t}
\eeq
Therefore, recalling \eqref{eq:DynXg},
\bes
\diff X^x_t &\stackrel{\q}{=}& \mu^x_t X^x_{t-} \diff t + \diff M^x_t\\
&\stackrel{\q^{fb}}{=}&  \mu^x_t X^x_{t-} \diff t + \left(\diff M^{x,fb}_t + \frac{\diff [M^x, L]_t}{L_t}\right)
\ees
where we used Proposition \ref{prop:LinearityOfQuadraticCovariation}: the first result is proved. Observing \eqref{eq:DynOfFXUnderQfb}, under $\q^{fb}$,
\beqs
\Delta X^x_t =  \frac{\Delta [M^x, M^f]_t}{ X^f_t }+  \Delta M^{x,fb}_t = \frac{\Delta [M^x, M^f]_t}{ X^f_t } + \Delta M^x_t - \frac{\Delta M^x \Delta M^f}{X^f_t} = \Delta M^x_t 
\eeqs
so the FX jumps do not change measure.\\
For the second result, given \eqref{eq:DynOfFXUnderQfb} for $x=f$, we apply the Itô formula \eqref{eq:ItoWithJumpsMultiDim_diff} to $X^{df}=(X^{f})^{-1}$:
\bes
\diff X^{df} &=& -(X^f)^{-2}_- \diff X^f + (X^f)^{-3} _-\diff [X^f]^c + (X^f)^{-1} - (X^f)^{-1}_- +(X^f)^{-2}_- \,\Delta X^f \\
&=& -(X^{df})^{2}_- \diff X^f + (X^{df})^{3} _-\diff [X^f]^c + (X^f)^{-1} - (X^f)^{-1}_- +(X^{df})^{2}_- \,\Delta X^f\\
&=& (X^{df})^{2}_- \left\{  - \mu^f X^f_- \diff t -\frac{\diff [M^f]}{X^f}-  \diff M^{f,fb} + X^{df}_-\diff [M^f]^c +\Delta M^f\right\}  - \frac{\Delta X^f}{X^f X^f_-} \\
&=& (X^{df})^{2}_- \left\{  - \frac{\mu^f}{X^{df}_-} \diff t -X^{df} \diff [M^f]-  \diff M^{f,fb} + X^{df}_-\diff [M^f]^c +\Delta M^f - \frac{X^{df}}{X^{df}_-} \Delta M^f \right\}
\ees
Now, using \eqref{eq:CovariationFormula1}, Proposition \ref{prop:IntegralWithPureJumpIntegrand} and that the Quadratic Covariation has finite variation:
\begin{gather*}
-X^{df} \diff [M^f] + X^{df}_-\diff [M^f]^c +\Delta M^f - \frac{X^{df}}{X^{df}_-} \Delta M^f\\
=  - (X^{df}_-+\Delta X^{df})  \Big\{\diff [M^f]^c+ (\Delta M^f)^2\Big\} + X^{df}_-\diff [M^f]^c +\Delta M^f\left( 1  - \frac{X^f_-}{X^f}\right)\\
= - X^{df}_-(\Delta M^f)^2 -\Delta X^{df}(\Delta M^f)^2 +\Delta M^f \frac{\Delta X^f}{X^f} \\
= (\Delta M^f)^2 \Big\{ - X^{df}_- -\Delta X^{df} +X^{df} \Big\}=0
\end{gather*}
that completes the second result. For the last result, under $\q^{fb}$,
\be
\diff X^{xf} &:=&\diff ( X^x X^{df}) = X^x_{-} \diff X^{df} + X^{df}_{-} \diff X^x + \diff [ X^x, X^{df}]\0\\
&=& X^x_{-} \left\{ - \mu^f X^{df}_- \diff t  - (X^{df})^{2}_{-} \diff M^{f,fb} \right\} 
+ X^{df}_{-} \left\{  \mu^x X^x_-  \diff t + X^{df} \diff [M^x, M^f] +  \diff M^{x,fb} \right\} + \diff [ X^x, X^{df}]\0\\
&=& X^x_{-} X^{df}_- \left\{ (\mu^x- \mu^f)   \diff t  -  X^{df}_{-} \diff M^{f,fb}  + \frac{\diff M^{x,fb}}{X^x_{-}}
  \right\} + X^{df}_{-}\,X^{df} \diff [M^x, M^f] + \diff [ X^x, X^{df}]\label{eq:diffXxfIntermidiate}
\ee
Now, from Proposition \ref{prop:LinearityOfQuadraticCovariation} and \eqref{eq:JumpsOfMxfb},
\bes
\diff [ X^x, X^{df}] &=& [X^{df} \diff [M^x, M^f] +  \diff M^{x,fb}, - (X^{df})^{2}_{-} \diff M^{f,fb}] \\
&=& -(X^{df})^{2}_{-} \Big\{ X^{df} \, \Delta M^x \, \Delta M^f\,\Delta M^{f,fb}  + \diff[ M^{x,fb}, M^{f,fb}]\Big\}\\
&=& -(X^{df})^{2}_{-} \Big\{ X^{df} \, \Delta M^x \, \Delta M^f\,\left(\Delta M^f_t - \frac{\Delta M^f \Delta M^f}{X^f_t}\right)  + \diff[ M^{x,fb}, M^{f,fb}]\Big\}\\
&=& -(X^{df})^{2}_{-} \Big\{ X^{df} \, \Delta M^x \, (\Delta M^f)^2\,\left(1 - X^{df} \Delta M^f \right)  + \diff[ M^{x,fb}, M^{f,fb}]\Big\}.
\ees
Using again Proposition \ref{prop:LinearityOfQuadraticCovariation}:
\bes
\diff[ M^{x,fb}, M^{f,fb}]&=& \left[\diff M^x - \frac{\diff [M^x, M^f]}{X^f}, \diff M^f - \frac{\diff [M^f, M^f]}{X^f}\right]\\
&=& \diff [M^x, M^f] - X^{df}\left[\diff M^x, \diff [M^f, M^f]\right] - X^{df}\left[\diff [M^x, M^f], \diff M^f\right]\\
&\phantom{=}&+ (X^{df})^2 \left[\diff [M^x, M^f], \diff [M^f, M^f]\right]\\
&=& \diff [M^x, M^f] - 2 \,X^{df}\Delta M^x\,(\Delta M^f)^2  + (X^{df})^2 \Delta M^x\,(\Delta M^f)^3    
\ees
then
\bes
\diff[  X^x, X^{df}] &=& -(X^{df})^{2}_{-} \Big\{ -X^{df} \, \Delta M^x \, (\Delta M^f)^2 + \diff [M^x, M^f] \Big\}
\ees
and, recalling \eqref{eq:diffXxfIntermidiate} and using \eqref{eq:CovariationFormula1} and Proposition \ref{prop:IntegralWithPureJumpIntegrand},
\begin{gather*}
X^{df}_{-}\,X^{df} \diff [M^x, M^f] + \diff [ X^x, X^{df}]\\
= X^{df}_{-}\,(X^{df}_-+\Delta X^{df}) \diff [M^x, M^f]-(X^{df})^{2}_{-} \Big\{ -X^{df} \, \Delta M^x \, (\Delta M^f)^2 + \diff [M^x, M^f] \Big\}\\
= X^{df}_{-}\Big\{\Delta X^{df}(\diff [M^x, M^f]^c+\Delta M^x\Delta M^f)+ X^{df}_{-}\, X^{df} \, \Delta M^x \, (\Delta M^f)^2 \Big\}\\
= X^{df}_{-}\Big\{\Delta X^{df} \Delta M^x\Delta M^f +  X^{df}_{-}\, X^{df} \, \Delta M^x \, (\Delta M^f)^2 \Big\}\\
= X^{df}_{-}\left\{\frac{-\Delta M^f}{X^f X^f_-}\,\Delta M^x\Delta M^f +  X^{df}_{-}\, X^{df} \, \Delta M^x \, (\Delta M^f)^2 \right\}=0
\end{gather*}
so we have the thesis.\\

\subsubappendix{Proof of Proposition \ref{prop:VfCollateralPricing}\label{App:VfCollateralPricing}}
 
\noindent Recalling \eqref{eq:defBetaBankAccount}, we have:
\bes
\widetilde{G}^{\textit{f,V-C}}_t  &:=& \frac{ V_t^f }{B^{fb}_t} -\frac{ C_t^g}{\beta^{fg}_t} + 
\int_0^t \frac{\diff \Pi_u^f}{B^{fb}_{u-}}  + \int_0^t \frac{\diff C_u^g - c_u^g\, C_{u-}^g \diff u}{\beta^{fg}_{u-}} +\diff \left[ C^g,  \frac{1}{\beta^{fg}} \right ]_u \\
&=&  \widetilde{V}^f_t  - \widetilde{C}^f_t + \int_0^t \frac{\diff \Pi_u^f}{B^{fb}_{u}}+ \int_0^t \frac{X^{gf}_{u-}}{ B^{fb}_u}\Big(\diff C_u^g - c_u^g\, C_{u-}^g \diff u \Big)+  \diff \left[ C^g,  \frac{1}{\beta^{fg}} \right ]_u
\ees
using that $B^{fb}$ is continuous with finite variation
\bes
\diff \left[ C^g,  \frac{1}{\beta^{fg}} \right ]_u &=& \diff \left[ C^g,   X^{gf} (B^{fb})^{-1} \right ]_u \\
&=& (B^{fb})^{-1}_u \left[ \diff  C^g_u,    \diff X^{gf}_u \right ]\\
&=& \frac{X^{gf}_{u-} }{B^{fb}_u} \left[\diff  C^g_u,   \frac{\diff  X^{gf}_u }{X^{gf}_{u-} }\right]
\ees
which, together with the previous equation, recalling \eqref{eq:GTildef}, gives the definition of $D^{\textit{f,V-C}}$. The easiest way for obtaining the pricing equation is to directly change measure from \eqref{eq:DerValueWithFXandCollateral} to measure $\q^{fb}$, i.e.~from numéraire
$B\mapsto\beta^{fd}:=B^{fb}X^f$ using \eqref{eq:SexpectedValueWithGenMeas}:
\beqs\begin{split}
V_t^f =&\, \frac{\beta^{fd}_t}{X^{f}_t}\,\e_t^{fb}\Bigg[  \frac{\phi_T^f \,X_T^f}{\beta^{fd}_T} + 
\int_t^T \frac{X^f_{u} \diff \Pi_u^f}{\beta^{fd}_{u-}} + X^f_{u} \diff \left[\Pi^f, \frac{1}{\beta^{fd}}\right]_u- \int_t^T \frac{X^g_{u-} C^g_{u-}}{\beta^{fd}_{u-}} \, (c_u^g - r_u^{gb}) \diff u \Bigg]
\end{split}
\eeqs
and, using the fact that $\Pi^f$ has finite variation and \eqref{eq:QCovOfFVProcess}, recalling also Proposition \ref{prop:LebesgueIntegralDiscontinuities},
\bes
A &:=&\int_t^T \frac{X^f_{u} \diff \Pi_u^f}{\beta^{fd}_{u-}} + X^f_{u} \diff \left[\Pi^f, \frac{1}{\beta^{fd}}\right]_u\\
&=&\int_t^T \frac{X^f_{u} (\psi^f\diff u + \diff \Phi_u^f)}{B^{fb}_u X^f_{u-}} + \sum_{t<u\leq T} \frac{X^f_{u}}{B^{fb}_u}\, \Delta \Pi^f_u \left(\frac{1}{X^f_u}-\frac{1}{X^f_{u-}}\right)\\
&=&\int_t^T   \frac{\psi^f\diff u}{B^{fb}_u} + \sum_{t<u\leq T} \frac{X^f_u\,\Delta \Pi^f_u}{B^{fb}_u}\,  \left(\frac{1}{X^f_{u-}} + \frac{1}{X^f_u}-\frac{1}{X^f_{u-}}\right)\\
&=&\int_t^T   \frac{\diff \Pi_u^f}{B^{fb}_u}
\ees
so that the pricing equation is confirmed. The last statement is a direct application of the $\widetilde{G}$ definition to the pricing formula and the proof is concluded. One could also obtain the analogous of \eqref{eq:diffAOverB}, to be used in the next proposition: we have,
\bes
\diff \widetilde{C}^f_t &:=& \diff \left(C^{g} X^{gf} (B^{fb})^{-1}\right)_t  = (C^{g} X^{gf})_{u-} \diff ((B^{fb})^{-1})_u + ((B^{fb})^{-1})_u \diff (C^{g} X^{gf})_u\\
&=& ((B^{fb})^{-1})_u  \left\{ -r_t^{fb} X^{gf}_{u-} C^g_{u-} \diff u +  X^{gf}_{u-}\diff C^g_u + C^g_{u-}\diff X^{gf}_u + \diff [X^{gf}, C^g]_u\right\}\\
&=& \frac{ X_{u-}^{gf}}{B^{fb}_u}\left\{\diff C_u^g - r_u^{fb}  C^g_{u-} \diff u + C^g_{u-} \,\frac{\diff X_u^{gf}}{X_{u-}^{gf}} + \left[ \diff C^g_u,   \frac{\diff  X^{gf}_u }{X^{gf}_{u-} } \right ] \right\}
\ees
so
\be\label{eq:diffAOverBforeign}
\frac{\diff D^{\textit{f,V-C}}_u}{B^{fb}_u} = \frac{\diff \Pi_u^f}{B^{fb}_u} + \diff \widetilde{C}^f_t + \frac{ X_{u-}^{gf}\,C^g_{u-}}{B^{fb}_u} \left\{ (r_u^{fb}- c_u^g)   \diff u - \frac{\diff X_u^{gf}}{X_{u-}^{gf}} \right\}
\ee
which will be useful afterwards.\\

\subsubappendix{Proof of Proposition \ref{prop:DynamicsOfVf}}\label{app:ProofOfDynamicsOfVf}

\noindent The dynamics of $(V^f-C^f)$ is a direct consequence of the fact that $\widetilde{G}_t^{\textit{f,V-C}}$ is a $\q^{fb}$-martingale, see Proposition \ref{prop:RiskNeutralDrift} from which one also obtains that 
$\diff \widetilde{G}_t^{f,\textit{V-C}} = (B_t^{fb})^{-1} \diff \mathcal{M}^{f,\textit{V-C}}_t$. Moreover, by \eqref{eq:GTildef},
\[
\diff \widetilde{G}_t^{f,\textit{V-C}} = \diff \widetilde{V}_t^f -\diff \widetilde{C}^f_t + \frac{\diff D_t^{f,\textit{V-C}} }{B_t^{fb}}
\]
so
\bes
\diff \widetilde{V}^f_t &=& \diff \widetilde{G}_t^{f,\textit{V-C}} -\left( \frac{\diff D_t^{f,\textit{V-C}} }{B_t^{fb}} - \diff \widetilde{C}_t^f \right)\\
&=& \frac{1}{B_t^{fb}}\left\{ \diff \mathcal{M}^{f,\textit{V-C}}_t - \diff \Pi_t^f -   X_{t-}^{gf}\,C^g_{t-}  \left\{ (r_t^{fb}- c_t^g)   \diff t - \frac{\diff X_t^{gf}}{X_{t-}^{gf}}  \right\}\right\}
\ees
where we used \eqref{eq:diffAOverBforeign} at the second line. Now, using \eqref{eq:diffFXxf},
\bes
\diff V_t^f &=& \diff ( B^{fb} \widetilde{V}^f)_t = B_t^{fb} \diff \widetilde{V}_t^f  + \widetilde{V}_{t-}^f\diff B_t^{fb}\\
&=&  \diff \mathcal{M}^{\textit{f,V-C}}_t -  \diff \Pi_t^f  - X_{t-}^{gf}\,C^g_{t-}  \left\{ (r_t^{fb}- c_t^g)   \diff t - 
(\mu^g_t- \mu^f_t)   \diff t  +  \frac{\diff M^{f,fb}_t}{X^{f}_{t-}}  - \frac{\diff M^{g,fb}}{X^g_{t-}} \right\}+ r_t^{fb} V_{t-}^f \diff t\\
&=& r_t^{fb} V_{t-}^f \diff t +  \diff \mathcal{M}^{\textit{f,V-C}}_t -  \diff \Pi_t^f  - X_{t-}^{gf}\,C^g_{t-}  \left\{ (r_t^{gb} - c_t^g ) \diff t  +  \frac{\diff M^{f,fb}_t}{X^{f}_{t-}}  - \frac{\diff M^{g,fb}}{X^g_{t-}} \right\} 
\ees
that can be rearranged to obtain the dynamics of $V^f$.\\
We then change the measure from $\q^{fb}$ to $\q$, where
\[
L_t := \e_t^{fb}\left[\frac{\diff \q\phantom{fb}}{\diff\q^{fb}}\right] = \frac{\beta^{df}_t }{B^{fb}_t} \,\frac{B^{fb}_0}{\beta^{df}_0}=\frac{\beta^{df}_t}{B^{fb}_t\, \beta^{df}_0}
\]
We have, using Proposition \ref{prop:LebesgueIntegralDiscontinuities},
\bes
\diff L_t &=&    \frac{1}{\beta^{df}_0} \left\{\beta^{df}_{t-} \diff ((B^{fb})^{-1})_t + ((B^{fb})^{-1})_t \diff \beta^{df}_t \right\}\\
&=&L_t \left\{- r^{fb}_t \diff t +  \frac{\diff \beta^{df}_t}{\beta^{df}_{t}} \right\}.
\ees
We now apply the Girsanov Theorem \ref{th:GeneralizedGirsanov}: first we observe that $D^f$ is a semimartingale with zero local martingale part, so its dynamics does not change
with the change of measure, see Corollary \ref{cor:GirsanovForFVProc}. Moreover,
\[
\diff \mathcal{M}^{\textit{fqd}} _t = \diff \mathcal{M}^f_t - \frac{\diff [\mathcal{M}^f, L]_t}{L_t}
\]   
is a local martingale under $\q$. Therefore
\bes
\diff V^f_t &\stackrel{\q^{fb}}{=}& r^{fb}_t\,V_{t-}^f \diff t + \diff \mathcal{M}_t^f - \diff D^f_t\\
&\stackrel{\q}{=}&  r^{fb}_t\,V_{t-}^f \diff t + \left\{ \diff \mathcal{M}^{\textit{fqd}} _t +\frac{\diff [\mathcal{M}^f, L]_t}{L_t}\right\}- \diff D^f_t
\ees
so we have to calculate the Covariation term: using Proposition \ref{prop:LinearityOfQuadraticCovariation} and that $B$ is continuous with finite variation
\bes
\frac{\diff [\mathcal{M}^f, L]_t}{L_t} &=&  \left[\diff\mathcal{M}^f_t, - r^{fb}_t \diff t +  \frac{\diff \beta^{df}_t}{\beta^{df}_t}\right] = \frac{ \left[\diff\mathcal{M}^f_t, \diff\beta^{df}_t\right]}{\beta^{df}_t} \\
&=& \frac{\left[\diff \mathcal{M}^f_t, B_t \diff X^{df}_t +  X^{df}_t \diff B_t\right]}{\beta^{df}_t}=\frac{\diff \left[\mathcal{M}^f,  X^{df}  \right]_t}{ X^{df}_t}
\ees
Moreover, defining $g(x):=x^{-1}$ and using \eqref{eq:QuadraticCovarianceOfFctOfX}
\bes
\diff \left[\mathcal{M}^f,  g(X^f)\right] &=& g'(X_{-}^f) \diff[\mathcal{M}^f, X^f]^c + \Big\{g(X^f) - g(X_{-}^f)  \Big\}\Delta \mathcal{M}^f\\
&=& -\frac{1}{(X_{-}^f)^2}\, \diff[\mathcal{M}^f, X^f]^c + \left\{\frac{1}{X^f} - \frac{1}{X_{-}^f}\right\}\Delta \mathcal{M}^f\\
&=& -\frac{1}{(X_{-}^f)^2}\, \diff[\mathcal{M}^f, X^f]^c -  \frac{\Delta X^f\,\Delta \mathcal{M}^f}{X^f X_{-}^f} \\
&=& -\frac{1}{X_{-}^f }\,\left\{\frac{1}{X_{-}^f } \diff[\mathcal{M}^f, X^f]^c +  \frac{\Delta X^f\,\Delta \mathcal{M}^f}{X^f}\right\}
\ees
Therefore,  using Proposition \ref{prop:IntegralWithPureJumpIntegrand} and the fact that the Quadratic Covariation has finite variation,
\bes
\frac{\diff [\mathcal{M}^f, L]}{L} &=&  X^f \diff \left[\mathcal{M}^f,  g(X^f)\right]_t \\
&=& -\frac{1}{X_{-}^f }\,\left\{\frac{X_{-}^f +  \Delta X^f}{X_{-}^f } \diff[\mathcal{M}^f, X^f]^c +  \Delta X^f\,\Delta \mathcal{M}^f \right\}\\
&=& -\frac{1}{X_{-}^f }\,\left\{  \diff[\mathcal{M}^f, X^f]^c +  \Delta X^f\,\Delta \mathcal{M}^f \right\}\\
&=&-\frac{\diff[\mathcal{M}^f, X^f]}{X_{-}^f}
\ees
which concludes the proof.

\end{document}